\documentclass[12pt]{article}
\usepackage{amsmath}
\usepackage{graphicx}
\usepackage{natbib}
\usepackage{url} 
\usepackage{nicefrac, mathrsfs, dsfont}
\usepackage{amsthm}
\usepackage{amssymb}
\usepackage{thmtools}
\usepackage{thm-restate}
\usepackage{algorithm}
\usepackage[noend]{algpseudocode}
\usepackage{subcaption}
\usepackage{enumitem}
\usepackage{xcolor}

\declaretheorem[name=Theorem]{thm}

\declaretheorem[name=Lemma]{lem}

\newtheorem{assumption}{Assumption}

\newcommand{\blind}{0}

\newcommand{\argmax}[1]{\underset{#1}{\operatorname{arg}\,\operatorname{max}}\;}
\newcommand{\argmin}[1]{\underset{#1}{\operatorname{arg}\,\operatorname{min}}\;}
\renewcommand{\hat}{\widehat}

\begin{document}

\def\spacingset#1{\renewcommand{\baselinestretch}%
{#1}\small\normalsize} \spacingset{1}


\if0\blind
{
  \title{\bf Universal inference meets random projections: a scalable test for log-concavity}
  \author{Robin Dunn \hspace{.2cm}\\
    Advanced Methodology and Data Science, \\Novartis Pharmaceuticals Corporation \\
    \\
    Aditya Gangrade \\
    Electrical Engineering and Computer Science, University of Michigan,\\
    Department of Electrical \& Computer Engineering,
    Boston University \\
    \\ 
    Larry Wasserman \\
    Department of Statistics \& Data Science and \\ Machine Learning Department, Carnegie Mellon University \\
    \\ 
    Aaditya Ramdas \\
    Department of Statistics \& Data Science and \\ Machine Learning Department, Carnegie Mellon University \\}
  \maketitle
} \fi

\if1\blind
{
  \bigskip
  \bigskip
  \bigskip
  \begin{center}
    {\LARGE\bf Universal inference meets random projections: a scalable test for log-concavity}
\end{center}
  \medskip
} \fi

\bigskip
\begin{abstract}
Shape constraints yield flexible middle grounds between fully nonparametric and fully parametric approaches to modeling distributions of data. The specific assumption of log-concavity is motivated by applications across economics, survival modeling, and reliability theory. However, there do not currently exist valid tests for whether the underlying density of given data is log-concave. The recent universal inference methodology provides a valid test. The universal test relies on maximum likelihood estimation (MLE), and efficient methods already exist for finding the log-concave MLE. This yields the first test of log-concavity that is provably valid in finite samples in any dimension, for which we also establish asymptotic consistency results. Empirically, we find that a random projections approach that converts the $d$-dimensional testing problem into many one-dimensional problems can yield high power, leading to a simple procedure that is statistically and computationally efficient.
\end{abstract}

\noindent%
\vfill

\newpage
\spacingset{1.9} 

\section{Introduction}

Statisticians frequently use density estimation to understand the underlying structure of their data. To perform nonparametric density estimation on a sample, it is common for researchers to incorporate shape constraints \citep{koenker2018shape, carroll2011testing}. Log-concavity is one popular choice of shape constraint; a density $f$ is called log-concave if it has the form $f = e^g$ for some concave function $g$. This class of densities encompasses many common families, such as the normal, uniform (over a compact domain), exponential, logistic, and extreme value densities \citep[][Table 1]{bagnoli2005}. Furthermore, specifying that the density is log-concave poses a middle ground between fully nonparametric density estimation and use of a parametric density family. As noted in \cite{cule2010}, log-concave density estimation does not require the choice of a bandwidth, whereas kernel density estimation in $d$ dimensions requires a $d\times d$ bandwidth matrix.  

Log-concave densities have multiple appealing properties; \cite{an1997log} describes several. For example, log-concave densities are unimodal, they have at most exponentially decaying tails (i.e., $f(x) = O(\exp(-c\|x\|))$ for some $c>0$), and all moments of the density exist. Log-concave densities are also closed under convolution, meaning that if $X$ and $Y$ are independent random variables from log-concave densities, then the density of $X+Y$ is log-concave as well. A unimodal density $f$ is strongly unimodal if the convolution of $f$ with any unimodal density $g$ is unimodal. Proposition~2 of \cite{an1997log} states that a density $f$ is log-concave if and only if $f$ is strongly unimodal. 

In addition, log-concave densities have applications in many domains. \cite{bagnoli2005} describe applications of log-concavity across economics, reliability theory, and survival modeling. (The latter two appear to use similar methods in the different domains of engineering and medicine, respectively.) Suppose a survival density function $f$ is defined on $(a,b)$ and has a survival function (or reliability function) $\bar{F}(x) = \int_x^b f(t) dt$. If $f$ is log-concave, then its survival function is log-concave as well. The failure rate associated with $f$ is $r(x) = f(x) / \bar{F}(x) = -\bar{F}{\:'}(x) / \bar{F}(x)$. Corollary~2 of \cite{bagnoli2005} states that if $f$ is log-concave on $(a,b)$, then the failure rate $r(x)$ is monotone increasing on $(a,b)$. Proposition~12 of \cite{an1997log} states that if a survival function $\bar{F}(x)$ is log-concave, then for any pair of nonnegative numbers $x_1, x_2$, the survival function satisfies $\bar{F}(x_1 + x_2) \leq \bar{F}(x_1) \bar{F}(x_2)$. This property is called the new-is-better-than-used property; it implies that the probability that a new unit will survive for time $x_1$ is greater than or equal to the probability that at time $x_2$, an existing unit will survive an additional time $x_1$.

Given the favorable properties of log-concave densities and their applications across fields, it is important to be able to test the log-concavity assumption. Previous researchers have considered this question as well. \cite{cule2010} develop a permutation test based on simulating from the log-concave MLE and computing the proportion of original and simulated observations in spherical regions. \cite{chen2013smoothed} construct an approach similar to the permutation test, using a test statistic based on covariance matrices.  \cite{hazelton2011assessing} develops a kernel bandwidth test, where the test statistic is the smallest kernel bandwidth that produces a log-concave density. \cite{carroll2011testing} construct a metric for the necessary amount of modification to the weights of a kernel density estimator to satisfy the shape constraint of log-concavity, and they use the bootstrap for calibration. While these approaches exhibit reasonable empirical performance in some settings, none of the aforementioned papers have proofs of validity (or asymptotic validity) for their proposed methods. As one exception, \cite{an1997log} uses asymptotically normal test statistics to test implications of log-concavity (e.g., increasing hazard rate) in the univariate, nonnegative setting. A general valid test for log-concavity has proved elusive.

\begin{figure}[t]
\begin{center}
\includegraphics[scale=.7]{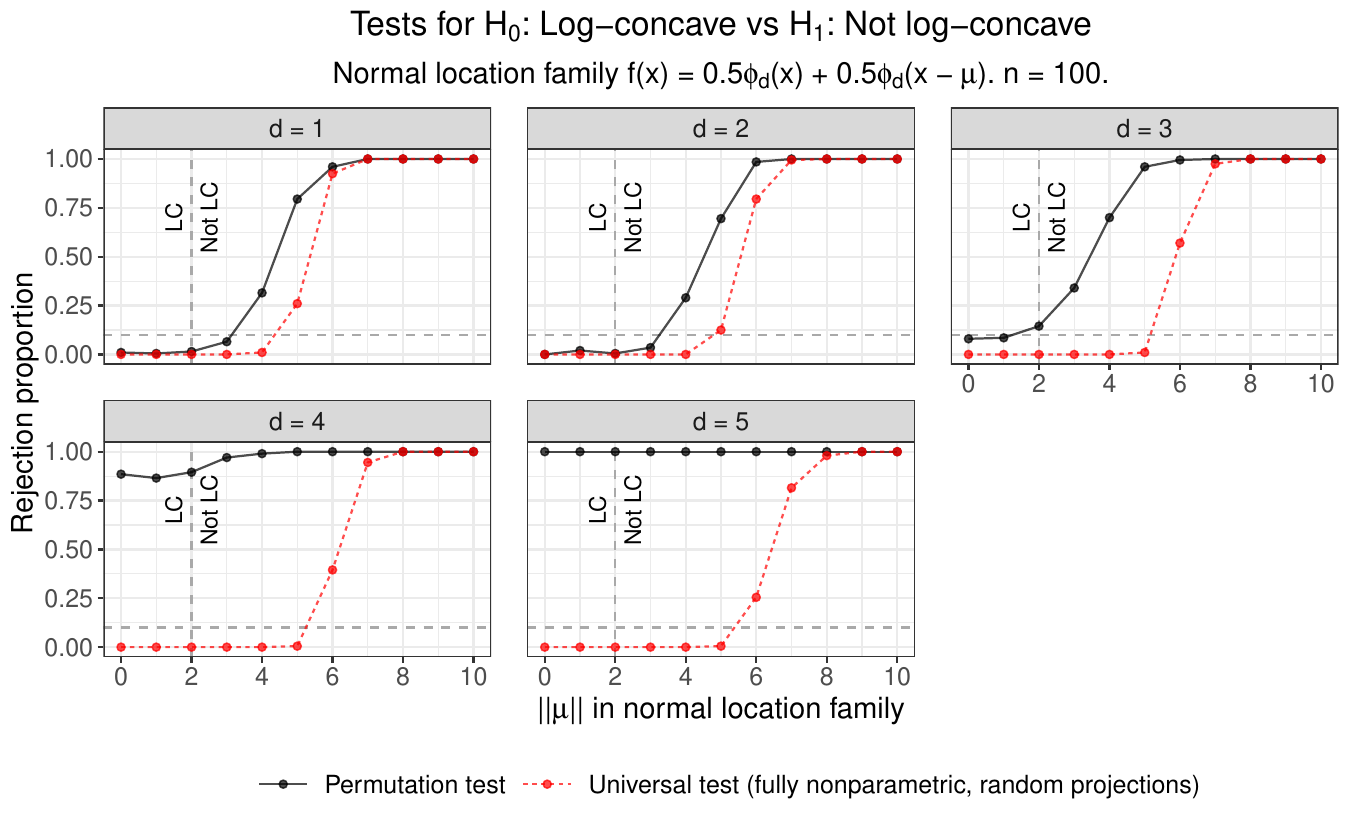}
\end{center}
\caption{Rejection proportions for tests of $H_0: f^*$ is log-concave versus $H_1: f^*$ is not log-concave. The permutation test from \cite{cule2010} is valid or approximately valid for $d\leq 3$, but it is not valid for $d\geq 4$. Our test that combines random projections and universal inference (Algorithm~\ref{alg:randproj}) is provably valid for all $n$ and $d$ while having high power.}
\label{fig:perm_randproj_n100}
\end{figure}

As an example, Figure~\ref{fig:perm_randproj_n100} shows the performance of the permutation test of log-concavity from \cite{cule2010} and a random projection variant of our universal test. Section~\ref{sec:permutation} explains the details of the permutation test, and Algorithm~\ref{alg:randproj} explains this universal test. Section~\ref{sec:example} provides more extensive simulations. If $\phi_d$ represents the $N(0, I_d)$ density and $\gamma\in (0,1)$, then $f(x) = \gamma\phi_d(x) + (1-\gamma)\phi_d(x-\mu)$ is log-concave only when $\|\mu\|\leq 2$. We simulate the permutation test in the $\gamma = 1/2$ setting for $1\leq d\leq 5$ and $n = 100$, testing the null hypothesis that the true underlying density $f^*$ is log-concave. We set $\mu = (\mu_1, 0, \ldots, 0)$, so that $\|\mu\| = |\mu_1|$. We use a significance level of $\alpha = 0.10$. Each point represents the proportion of times we reject $H_0$ over 200 simulations. 
Figure~\ref{fig:perm_randproj_n100} shows that the permutation test is valid at $d = 1$ and $d = 2$ and approximately valid at $d = 3$. Alternatively, at $d = 4$ and $d = 5$, this test rejects $H_0$ at proportions much higher than $\alpha$, even when the underlying density is log-concave ($\|\mu\| \leq 2$). In contrast, the universal test is valid in all dimensions. Furthermore, this universal test has high power for reasonable $\|\mu\|$ even as we increase $d$.

To develop a test for log-concavity with validity guarantees, we consider the universal likelihood ratio test (LRT) introduced in \cite{wasserman2020universal}. This approach provides valid hypothesis tests in any setting in which we can maximize (or upper bound) the null likelihood. Importantly, validity holds in finite samples and without regularity conditions on the class of models. Thus, it holds even in high-dimensional settings without assumptions. 

Suppose $\mathcal{F}_d$ is a (potentially nonparametric) class of densities in $d$ dimensions. The universal LRT allows us to test hypotheses of the form $H_0: f^*\in\mathcal{F}_d$ versus $H_1: f^*\notin\mathcal{F}_d$. In this paper,  $\mathcal{F}_d$ will represent the class of all log-concave densities in $d$ dimensions. 

Assume we have $n$ independent and identically distributed (iid) observations $Y_1, \ldots, Y_n$ with some true density $f^*$. To implement the split universal LRT, we randomly partition the indices from 1 to $n$, denoted as $[n]$, into $\mathcal{D}_0$ and $\mathcal{D}_1$. (Our simulations assume $|\mathcal{D}_0| = |\mathcal{D}_1| = n/2$, but any split proportion is valid.) Using the data indexed by $\mathcal{D}_1$, we fit \emph{any} density $\hat{f}_1$ of our choice, such as a kernel density estimator. The likelihood function evaluated on a density $f$ over the data indexed by $\mathcal{D}_0$ is denoted $\mathcal{L}_0(f) = \prod_{i\in\mathcal{D}_0} f(Y_i)$. Using the data indexed by $\mathcal{D}_0$, we fit $\hat{f}_0 = \argmax{f\in\mathcal{F}_d} \mathcal{L}_0(f)$, which is the null maximum likelihood estimator (MLE). The split LRT statistic is 
\[
T_n(f) = \mathcal{L}_0(\hat{f}_1) / \mathcal{L}_0(f).
\] 
The test rejects if $T_n(\hat f_0) \geq 1/\alpha$.

\begin{restatable}[\cite{wasserman2020universal}]{thm}{thmValidNonpar} \label{thm:valid_nonpar}
$T_n(\hat f_0)$ is an e-value, meaning that it has expectation at most one under the null. Hence, $1/T_n(\hat f_0)$ is a valid p-value, and rejecting the null when $T_n(\hat f_0)\geq 1/\alpha$ is a valid level-$\alpha$ test. That is, under $H_0: f^*\in\mathcal{F}_d$,
\[
\mathbb{P}(T_n(\hat f_0) \geq 1/\alpha) \leq \alpha.
\]
\end{restatable}

\cite{wasserman2020universal} prove Theorem~\ref{thm:valid_nonpar}, but Appendix~\ref{app:validity} contains a proof for completeness. 
It is also possible to invert this test, yielding a confidence set for $f^*$, but for nonparametric classes $\mathcal{F}_d$, these are not in closed form and are hard to compute numerically, so we do not pursue this direction further. 
Nevertheless, as long as we are able to construct $\hat{f}_0$ (or actually simply calculate or upper bound its likelihood), it is possible to perform the nonparametric hypothesis test described in Theorem~\ref{thm:valid_nonpar}. 

Prior to the universal LRT developed by \cite{wasserman2020universal}, there was no hypothesis test for $H_0: f^*$ is log-concave versus $H_1: f^*$ is not log-concave with finite sample validity, or even asymptotic validity. Since it is possible to compute the log-concave MLE on any sample of size $n \geq d+1$, the universal LRT described above provides a valid test as long as $|\mathcal{D}_0|\geq d+1$. The randomization in the splitting above can be entirely removed --- without affecting the validity guarantee --- at the expense of more computation. \cite{wasserman2020universal} show that one can repeatedly compute $T_n(\hat f_0)$ under independent random splits, and average all the test statistics; since each has expectation at most one under the null, so does their average.
It follows that the test based on averaging over multiple
splits still has finite sample validity.

Section~\ref{sec:solving} reviews critical work on the construction and convergence of log-concave MLE densities. Section~\ref{sec:tests} describes the permutation test from \cite{cule2010} and proposes several universal tests for log-concavity. The log-concave MLE does suffer from a curse of dimensionality, both computationally and statistically. Hence, our most important contribution is a scalable method using random projections to reduce the multivariate problem into many univariate testing problems, where the log-concave MLE is easy to compute. (This relies on the fact that if a density is log-concave then every projection is also log-concave.) 
Section~\ref{sec:example} compares these tests through a simulation study. Section~\ref{sec:theoretical} explains a theoretical result about the power of the universal LRT for tests of log-concavity. All proofs and several additional simulations are available in the appendices. Code to reproduce all analyses is available at 
\if0\blind{\url{https://github.com/RobinMDunn/LogConcaveUniv}.}\fi
\if1\blind{[redacted for blind review].}\fi

\section{Finding the Log-concave MLE} \label{sec:solving}

Suppose we observe an iid sample $X_1, \ldots, X_n \in \mathbb{R}^d$ from a $d$-dimensional density $f^*$, where $n \geq d + 1$. Recall that  $\mathcal{F}_d$ is the class of log-concave densities in $d$ dimensions.  The log-concave MLE is  $\hat{f}_n = \argmax{f\in\mathcal{F}_d}  \sum_{i=1}^n \log\{f(X_i)\}$. Theorem~1 of \cite{cule2010} states that with probability 1, $\hat{f}_n$ exists and is unique. Importantly, this does not require $f^* \in \mathcal{F}_d$.

The construction of $\hat{f}_n$ relies on the concept of a tent function $\bar{h}_y: \mathbb{R}^d \to \mathbb{R}$. For a given vector $y = (y_1, \ldots, y_n) \in \mathbb{R}^n$ and given the sample $X_1, \ldots, X_n$,  the tent function $\bar{h}_y$ is the smallest concave function that satisfies $\bar{h}_y(X_i) \geq y_i$ for $i = 1,\ldots, n$. Let $C_n$ be the convex hull of the observations $X_1, \ldots, X_n$. Consider the objective function $$\sigma(y_1, \ldots, y_n) = -\frac{1}{n} \sum_{i=1}^n y_i + \int_{C_n} \exp\{\bar{h}_y(x)\} dx.$$  Theorem~2 of \cite{cule2010} states that $\sigma$ is a convex function, and it has a unique minimum at the value $y^*\in \mathbb{R}^n$ that satisfies $\log(\hat{f}_n) = \bar{h}_{y^*}$. 

Thus, to find the tent function that defines the log-concave MLE, we need to minimize $\sigma$ over $y\in\mathbb{R}^n$. $\sigma$ is not differentiable, but Shor's algorithm \citep{shor2012minimization} uses a subgradient method to optimize convex, non-differentiable functions. This method is guaranteed to converge, but convergence can be slow. Shor's $r$-algorithm involves some computational speed-ups over Shor's algorithm, and \cite{cule2010} use this algorithm in their implementation. Shor's $r$-algorithm is not guaranteed to converge, but \cite{cule2010} state that they agree with \cite{kappel2000implementation} that the algorithm is ``robust, efficient, and accurate.'' The \texttt{LogConcDEAD} package for log-concave density estimation in arbitrary dimensions implements this method \citep{cule2009logconcdead}.

Alternatively, the \texttt{logcondens} package implements an active set approach to solve for the log-concave MLE in one dimension \citep{logcondens}. This algorithm is based on solving for a vector that satisfies a set of active constraints and then using the tent function structure to compute the log-concave density associated with that vector. See Section~3.2 of \cite{dumbgen2007active} for more details.

Figure~\ref{fig:logconc_densities_n5000_d1} shows the true $f^*$ and log-concave MLE ($\hat{f}_n$) densities of several samples from two-component Gaussian mixtures. The underlying density is $f^*(x) = 0.5\phi_d(x) + 0.5\phi_d(x-\mu).$ Again, this density is log-concave if and only if $\|\mu\| \leq 2$. (We develop this example further in Section~\ref{sec:example}.) In the $n=5000$ and $d=1$ setting, we simulate samples $X_1, \ldots, X_n \sim f^*$ and compute the log-concave MLE $\hat{f}_n$ on each random sample. These simulations use both the \texttt{LogConcDEAD} and \texttt{logcondens} packages to fit $\hat{f}_n$. \texttt{logcondens} only works in one dimension but is much faster than \texttt{LogConcDEAD}. The two packages produce densities with similar appearances. Furthermore, we include values of $n^{-1} \sum_{i=1}^n \log(f^*(x_i))$ on the true density plots and $n^{-1} \sum_{i=1}^n \log(\hat{f}_n(x_i))$ on the log-concave MLE plots. The log likelihood is approximately the same for the two density estimation methods. 

In the first two rows of Figure~\ref{fig:logconc_densities_n5000_d1}, the true density is log-concave and in this case, we see that $n^{-1} \sum_{i=1}^n \log(\hat{f}_n(x_i))$ is approximately equal to $n^{-1} \sum_{i=1}^n \log(f^*(x_i))$. When $\|\mu\| = 4$, the underlying density is not log-concave. The log-concave MLE at $\|\mu\| = 4$ and  $n=5000$ seems to have normal tails, but it is nearly uniform in the middle.

Appendix~\ref{app:addl_visualize} contains additional plots of the true densities and log-concave MLE densities when $n=50$ and $d=1$, $n=50$ and $d=2$, and $n=500$ and $d=2$. In the smaller sample $d=1$ setting, we still observe agreement between \texttt{LogConcDEAD} and \texttt{logcondens}. When $d=2$ and the true density is log-concave, the log-concave MLE is closer to the true density at larger $n$. Alternatively, when $d=2$ and the true density is not log-concave, the log-concave MLE density again appears to be uniform in the center.

\begin{figure}[t]
\begin{center}
\includegraphics[scale=.6]{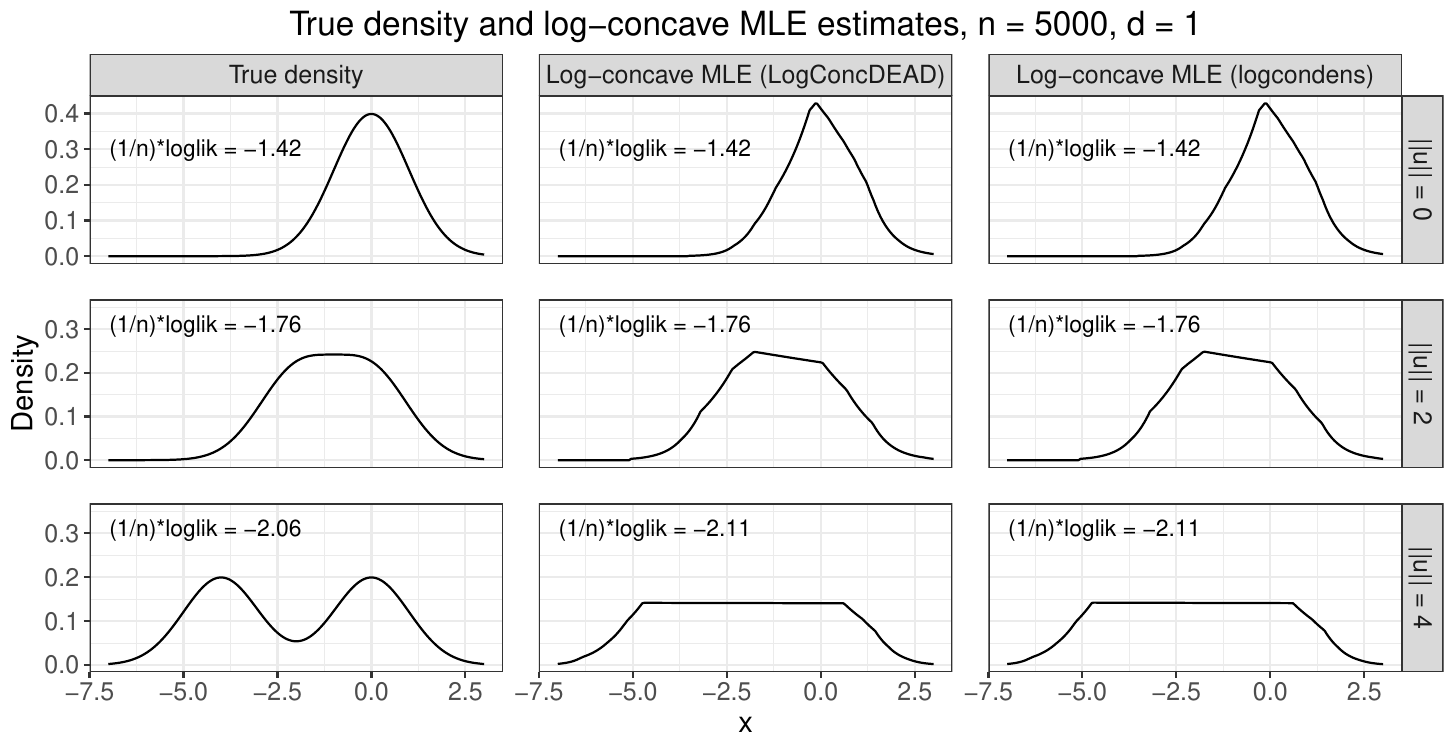}
\end{center}
\caption{Densities from fitting log-concave MLE on $n = 5000$ observations. The true density is the Normal mixture $f^*(x) = 0.5\phi_1(x) + 0.5\phi_1(x-\mu)$. In all settings, the \texttt{LogConcDEAD} and \texttt{logcondens} packages return similar results. In the $\|\mu\| = 0$ and $\|\mu\| = 2$ log-concave settings, the log-concave MLE is close to the true density. In the $\|\mu\| = 4$ non-log-concave setting, the log-concave densities appear to have normal tails and uniform centers.}
\label{fig:logconc_densities_n5000_d1}
\end{figure}

\cite{cule2010} formalize the convergence of $\hat{f}_n$. Let $D_\text{KL}(g \| f)$ be the Kullback-Leibler (KL) divergence of $g$ from $f$. Define $f^\text{LC} = \argmin{f\in\mathcal{F}_d} D_\text{KL}(f^* \| f)$ as the log-concave projection of $f^*$ onto the set of all log-concave densities $\mathcal{F}_d$ \citep{barber2021local, samworth2018recent}. In the simplest case, if $f^*\in\mathcal{F}_d$, then $f^\text{LC} = f^*$. Regardless of whether $f^*\in\mathcal{F}_d$, suppose $f^*$ satisfies the following conditions: $\int_{\mathbb{R}^d} \|x\| f^*(x) dx < \infty$, $\int_{\mathbb{R}^d} f^* \log_+(f^*) < \infty$ (where $\log_+(x) = \max\{\log(x), 0\}$), and the support of $f^*$ contains an open set. By Lemma~1 of \cite{cule2010theoretical}, there exists some $a_0 > 0$ and $b_0 \in\mathbb{R}$ such that $f^\text{LC}(x) \leq \exp(-a_0 \|x\| + b_0)$ for any $x\in\mathbb{R}^d$. Theorem~3 of \cite{cule2010} states that for any $a<a_0$, $$\int_{\mathbb{R}^d} \exp(a\|x\|) |\hat{f}_n(x) - f^\text{LC}(x)| dx \to 0 \qquad \text{almost surely.}$$ This means that the integrated difference between $\hat{f}_n$ and $f^\text{LC}$ converges to 0 even when we multiply the tails by some exponential weight. Furthermore, Theorem~3 of \cite{cule2010} states that if $f^\text{LC}$ is continuous, then $$\sup_{x\in\mathbb{R}^d} \left\{ \exp(a\|x\|) |\hat{f}_n(x) - f^\text{LC}(x)|\right\} \to 0 \qquad \text{almost surely.}$$ 

In the case where $f^*\in\mathcal{F}_d$, it is possible to describe rates of convergence of the log-concave MLE in terms of the Hellinger distance. The squared Hellinger distance is $$h^2(f, g) = \int_{\mathbb{R}^d} (f^{1/2} - g^{1/2})^2.$$ As stated in \cite{chen2021new} and shown in \cite{kim2016global} and \cite{kur2019optimality}, the rate of convergence of $\hat{f}_n$ to $f^*$ in sqaured Hellinger distance is  
\[\sup_{f^* \in \mathcal{F}_d} \mathbb{E}[h^2(\hat{f}_n - f^*)] \leq K_d \cdot \begin{cases} 
      n^{-4/5} & d = 1 \\
      n^{-2/(d+1)}\log(n) & d \geq 2
   \end{cases},
\]
where $K_d > 0$ depends only on $d$. 

\section{Tests for Log-concavity} \label{sec:tests}

We first describe a permutation test as developed in \cite{cule2010}, and then we propose several universal inference tests. The latter are guaranteed to control the type I error at level $\alpha$ (theoretically and empirically), while the former is not always valid even in simulations, as already demonstrated in Figure~\ref{fig:perm_randproj_n100}.

\subsection{Permutation Test~\citep{cule2010}} \label{sec:permutation}

\cite{cule2010} describe a permutation test of the hypothesis $H_0: f^*\in\mathcal{F}_d$ versus $H_1: f^*\notin\mathcal{F}_d$.
First, this test fits the log-concave MLE $\hat{f}_n$ on $\mathcal{Y} = \{Y_1, \ldots, Y_n\}$. Then it draws another sample $\mathcal{Y}^* = \{Y_1^*, \ldots, Y_n^*\}$ from $\hat{f}_n$. Next, it computes a test statistic based on the empirical distributions of $\mathcal{Y}$ and $\mathcal{Y}^*$. As the permutation step, the procedure repeatedly ``shuffles the stars'' to permute the observations in $\mathcal{Y} \cup \mathcal{Y}^*$ into two sets of size $n$, and it re-computes the test statistic on each permuted sample. We reject $H_0$ if the original test statistic exceeds the $1-\alpha$ quantile of the test statistics computed from the permuted samples. We explain the permutation test in more detail in Appendix~\ref{app:perm_test_desc}.

Intuitively, this test assumes that if $H_0$ is true, $\mathcal{Y}$ and $\mathcal{Y}^*$ will be similar. Then the original test statistic will not be particularly large relative to the test statistics computed from the permuted samples. Alternatively, if $H_0$ is false, $\mathcal{Y}$ and $\mathcal{Y}^*$ will be dissimilar, and the converse will hold. This approach is not guaranteed to control the type I error level. Figure~\ref{fig:perm_randproj_n100} shows cases both where the permutation test performs well and where the permutation test's false positive rate is much higher than $\alpha$. 

\subsection{Universal Tests in $d$ Dimensions} \label{sec:ddim}

Alternatively, we can use universal approaches to test for log-concavity. Theorem~\ref{thm:valid_nonpar} justifies the universal approach for testing whether $f^*\in\mathcal{F}_d$.  Recall that the universal LRT provably controls the type I error level in finite samples. To implement the universal test on a single subsample, we partition $[n]$ into $\mathcal{D}_0$ and $\mathcal{D}_1$. Let $\hat{f}_0$ be the maximum likelihood log-concave density estimate fit on $\{Y_i : i\in\mathcal{D}_0\}$. Let $\hat{f}_1$ be any density estimate fit on $\{Y_i : i\in \mathcal{D}_1\}$. The universal test rejects $H_0$ when $$T_n = \prod_{i\in\mathcal{D}_0} \{\hat{f}_1(Y_i) / \hat{f}_0(Y_i)\} \geq 1/\alpha.$$ 

The universal test from Theorem~\ref{thm:valid_nonpar} holds when $T_n$ is replaced with an average of test statistics, each computed over random partitions of $[n]$. Algorithm~\ref{alg:LCsubsampling} explains how to use subsampling to test $H_0: f^*\in\mathcal{F}_d$ versus $H_1: f^*\notin\mathcal{F}_d$. The $j^{th}$ random partition of $[n]$ produces a test statistic $T_{n,j}$. The subsampling approach rejects $H_0$ when $B^{-1} \sum_{j=1}^B T_{n,j} \geq 1/\alpha$. Note that each test statistic $T_{n,j}$ is nonnegative.  In cases where we have sufficient evidence against $H_0$, it may be possible to reject $H_0$ at some iteration $b < B$. That is, for any $b$ such that $1\leq b < B$, $\sum_{j=b}^B T_{n,j} \geq 0$. If there is a value of $b<B$ such that $B^{-1} \sum_{j=1}^b T_{n,j} \geq 1/\alpha$,  then it is guaranteed that $T_n = B^{-1} \sum_{j=1}^B T_{n,j} \geq 1/\alpha$. Algorithms~\ref{alg:LCsubsampling}--\ref{alg:randproj} incorporate this fact by rejecting early if we have sufficient evidence against $H_0$.

\begin{algorithm}[ht] \label{alg:subsampling}
\caption{For $H_0: f^*\in\mathcal{F}_d$ versus $H_1: f^*\notin\mathcal{F}_d$, compute the subsampling test statistic or run the test.}
\hspace*{\algorithmicindent} \textbf{Input:} $n$ iid $d$-dimensional observations $Y_1, \ldots, Y_n$ from unknown density $f^*$,\\
\hspace*{\algorithmicindent} \hskip 12pt number of subsamples $B$, significance level $\alpha$, any density estimation approach. \\
\hspace*{\algorithmicindent} \textbf{Output:} The subsampling test statistic $T_n$ or the test result.
\begin{algorithmic}[1]
\For {$b=1,2,\ldots,B$}
\State Randomly partition $[n]$ into $\mathcal{D}_{0,b}$ and $\mathcal{D}_{1,b}$ such that $|\mathcal{D}_{0,b}| = |\mathcal{D}_{1,b}| = n/2$.
\State Where $\mathcal{L}_{0,b}(f) = \prod_{i\in\mathcal{D}_{0,b}} f(Y_i)$, compute $\hat{f}_{0,b} = \argmax{f\in\mathcal{F}_d} \mathcal{L}_{0,b}(f)$.
\State Fit a density $\hat{f}_{1,b}$ on $\{Y_i: i \in \mathcal{D}_{1,b}\}$, using the input density estimation approach.
\State Compute $T_{n,b} = \mathcal{L}_{0,b}(\hat{f}_{1,b}) / \mathcal{L}_{0,b}(\hat{f}_{0,b})$.
\If{$B^{-1} \sum_{j=1}^b T_{n,j} \geq 1/\alpha$} \textbf{return} rejection of hypothesis. \EndIf
\EndFor
\State \textbf{return} the subsampling test statistic $T_n = B^{-1} \sum_{j=1}^B T_{n,j}$.
\end{algorithmic}
\label{alg:LCsubsampling}
\end{algorithm}

Both \texttt{logcondens} ($d = 1$) and \texttt{LogConcDEAD} ($d \geq 1$) compute the log-concave MLE $\hat{f}_0$. The choice of $\hat{f}_1,$ which can be any density, is flexible, and we explore several options.

\noindent \textbf{Full Oracle} The full oracle approach uses the true density $f^*$ in the numerator, i.e., in Algorithm~\ref{alg:LCsubsampling}, the input density estimation approach is to set $\hat{f}_{1,b} = f^*$. This method is a helpful theoretical comparison, since it avoids the depletion in power that occurs when $\hat{f}_{1,b}$ does not approximate $f^*$ well. We would expect the power of this approach to exceed the power of any approach that estimates a numerator density on $\{Y_i : i \in \mathcal{D}_{1,b}\}$. 

\noindent \textbf{Partial Oracle} The partial oracle approach uses a $d$-dimensional parametric MLE density estimate in the numerator. Suppose we know (or we guess) that the true density is parameterized by some unknown real-valued vector $\theta^*\in\mathbb{R}^p$ such that $f^* = f_{\theta^*}$. In Algorithm~\ref{alg:LCsubsampling}, the input density estimation approach is to set $\hat{f}_{1,b} = f_{\hat{\theta}_{1,b}}$, where $\hat{\theta}_{1,b}$ is the MLE of $\theta$ over $\{Y_i : i \in \mathcal{D}_{1,b}\}$. If the true density is from the parametric family $(f_\theta: \theta\in\mathbb{R}^p)$, we would expect this method to have good power relative to other density estimation methods. 

\noindent \textbf{Fully Nonparametric} The fully nonparametric method uses a $d$-dimensional kernel density estimate (KDE) in the numerator. In Algorithm~\ref{alg:LCsubsampling}, the input density estimation approach is to set $\hat{f}_{1,b}$ to the kernel density estimate computed on $\mathcal{D}_{1,b}$. Kernel density estimation involves the choice of a bandwidth. The \texttt{ks} package \citep{ks} in \texttt{R} can fit multidimensional KDEs and has several bandwidth computation procedures. These options include a plug-in bandwidth \citep{wand1994multivariate, duong2003plug, chacon2010multivariate}, a least squares cross-validated bandwidth \citep{bowman1984alternative, rudemo1982empirical}, and a smoothed cross-validation bandwidth \citep{jones1991simple, duong2005cross}. In the parametric density case, we would expect the fully nonparametric method to have lower power than the full oracle method and the partial oracle method. If we do not want to make assumptions about the true density, this may be a good choice.

\subsection{Universal Tests with Dimension Reduction}

Suppose we write each random variable $Y \in \mathbb{R}^d$ as $Y = (Y^{(1)}, \ldots, Y^{(d)})$. As noted in \cite{an1997log}, if the density of $Y$ is log-concave, then the marginal densities of $Y^{(1)}, \ldots, Y^{(d)}$ are all log-concave. In the converse direction, if marginal densities of $Y^{(1)}, \ldots, Y^{(d)}$ are all log-concave and $Y^{(1)}, \ldots, Y^{(d)}$ are all independent, then the density of $Y$ is log-concave. Proposition 1 of \cite{cule2010} uses a result from \cite{prekopa1973logarithmic} to deduce a more general result. We restate Proposition 1(a) in Theorem~\ref{thm:projection}. 

\begin{restatable}[Proposition 1(a) of \cite{cule2010}]{thm}{thmProjection} \label{thm:projection}
Suppose $Y\in\mathbb{R}^d$ is a random variable from a distribution having density $f^*$ with respect to Lebesgue measure. Let $V$ be a subspace of $\mathbb{R}^d$, and denote the orthogonal projection of $y$ onto $V$ as $P_V(y)$. If $f^*$ is log-concave, then the marginal density of $P_V(Y)$ is log-concave and the conditional density $f^*_{Y\mid P_V(Y)}(\cdot\mid t)$ of $Y$ given $P_V(Y) = t$ is log-concave for each $t$.
\end{restatable}

When considering how to test for log-concavity, \cite{an1997log} notes that univariate tests for log-concavity could be used in the multivariate setting. For our purposes, we use Theorem~\ref{thm:projection}'s implication that if $f^*$ is log-concave, then the one-dimensional projections of $f^*$ are also log-concave. We develop new universal tests on these one-dimensional projections. 
 
To reduce the data to one dimension, we take one of two approaches.

\subsubsection{Dimension Reduction Approach 1: Axis-aligned Projections}

We can represent any $d$-dimensional observation $Y_i$ as $Y_i = (Y_i^{(1)}, Y_i^{(2)}, \ldots, Y_i^{(d)})$. Algorithm~\ref{alg:ddimreduce} describes an approach that computes a test statistic for each of the $d$ dimensions.

\begin{algorithm}[ht]
\caption{For $H_0: f^*\in\mathcal{F}_d$ versus $H_1: f^*\notin\mathcal{F}_d$, compute the axis-aligned projection test statistics or run the test.}
\hspace*{\algorithmicindent} \textbf{Input:} $n$ iid $d$-dimensional observations $Y_1, \ldots, Y_n$ from unknown density $f^*$, \\
\hspace*{\algorithmicindent} \hskip 12pt number of subsamples $B$, significance level $\alpha$. \\
\hspace*{\algorithmicindent} \textbf{Output:} $d$ test statistics $T_n^{(k)}$, $k=1,\ldots,d$, or the test result.
\begin{algorithmic}[1]
\For {$k=1,2,\ldots,d$}
\For {$b=1,2,\ldots,B$}
\State Randomly partition $[n]$ into $\mathcal{D}_{0,b}$ and $\mathcal{D}_{1,b}$ such that $|\mathcal{D}_{0,b}| = |\mathcal{D}_{1,b}| = n/2$.
\State Estimate a one-dimensional density $\hat{f}_{1,b,k}$ on $\{Y_i^{(k)}: i\in\mathcal{D}_{1,b}\}$.
\State Estimate the log-concave MLE $\hat{f}_{0,b,k}$ on $\{Y_i^{(k)}: i \in \mathcal{D}_{0,b}\}$.
\If{$B^{-1} \sum_{j=1}^b \prod_{i \in \mathcal{D}_{0,j}} \{\hat{f}_{1,j,k}(Y_i^{(k)}) / \hat{f}_{0,j,k}(Y_i^{(k)}) \} \geq d/\alpha$} \textbf{stop}, reject null. \EndIf
\EndFor
\State Compute the test statistic $T_{n}^{(k)} = B^{-1} \sum_{b=1}^B \prod_{i \in \mathcal{D}_{0,b}} \{\hat{f}_{1,b,k}(Y_i^{(k)}) / \hat{f}_{0,b,k}(Y_i^{(k)}) \}$
\EndFor
\State \textbf{return} the test statistics $T_n^{(k)}$, $k = 1,\ldots, d$. 
\end{algorithmic}
\label{alg:ddimreduce}
\end{algorithm}

We reject $H_0: f^* \in \mathcal{F}_d$ if at least one of the $d$ test statistics $T_{n}^{(1)}, \ldots, T_n^{(d)}$ exceeds $d/\alpha$. Instead of checking this condition at the very end of the algorithm, we check this along the way, and we stop early to save computation if this condition is satisfied (line 6).
This rejection rule has valid type I error control because under $H_0$, 
$$\mathbb{P}(\{T_{n}^{(1)} \geq d/\alpha\} \: \cup \: \{T_{n}^{(2)} \geq d/\alpha\} \: \cup \: \cdots \: \cup \: \{T_{n}^{(d)} \geq d/\alpha\}) \leq \sum_{k=1}^d \mathbb{P}(T_{n}^{(k)} \geq d/\alpha) \leq d(\alpha / d) = \alpha.$$

If we do not simply want an accept-reject decision but would instead like a real-valued measure of evidence, then we can note that $p_n := d \min_{k \in [d]}1/T_n^{(k)}$ is a valid p-value. Indeed the above equation can be rewritten as the statement $\mathbb{P}(p_n \leq \alpha)\leq \alpha$, meaning that under the null, the distribution of $p_n$ is stochastically larger than uniform.

To run the test, we must fit some one-dimensional density $\hat{f}_{1,b,k}$ on $\{Y_i^{(k)} : i \in \mathcal{D}_{1,b}\}$. We consider two density estimation methods; the same applies to the next subsection. Thus, for the universal LRTs with dimension reduction, we consider four total combinations of two dimension reduction approaches and two density estimation methods.

\noindent \textbf{Density Estimation Method 1: Partial Oracle} This approach uses parametric knowledge about the true density. The numerator $\hat{f}_{1,b,k}$ is the parametric MLE fit on $\mathcal{D}_{1,b}$. 

\noindent \textbf {Density Estimation Method 2: Fully Nonparametric} This approach does not use any prior knowledge about the true density. Instead, we use kernel density estimation (e.g., \texttt{ks} package with plug-in bandwidth) to fit $\hat{f}_{1,b,k}$.

\subsubsection{Dimension Reduction Approach 2: Random Projections}

We can also construct one-dimensional densities by projecting the data onto a vector drawn uniformly from the unit sphere. Algorithm~\ref{alg:randproj} shows how to compute the random projection test statistic $T_n$. As discussed in Section~\ref{sec:ddim}, Theorem~\ref{thm:valid_nonpar} justifies the validity of this approach. In short, each individual projection test statistic $T_{n,j}$ is an e-value, meaning that it has expectation of at most one under the null. Thus the average of $T_{n,j}$ values is also an e-value. Since each $T_{n,j}$ is nonnegative, if there is some $k < n_\text{proj}$ such that $(1/n_\text{proj}) \sum_{j=1}^k T_{n,j} \geq 1/\alpha$, then we can reject $H_0$ without computing all $n_\text{proj}$ test statistics. 

\begin{algorithm}[ht]
\caption{For $H_0: f^*\in\mathcal{F}_d$ versus $H_1: f^*\notin\mathcal{F}_d$, compute the random projection test statistic or run the test.}
\hspace*{\algorithmicindent} \textbf{Input:} $n$ iid $d$-dimensional observations $Y_1, \ldots, Y_n$ from unknown density $f^*$, \\
\hspace*{\algorithmicindent} \hskip 12pt number of subsamples $B$, significance level $\alpha$, number of random projections $n_\text{proj}$. \\
\hspace*{\algorithmicindent} \textbf{Output:} The random projection test statistic $T_n$ or the test result.
\begin{algorithmic}[1]
\For {$k=1,2,\ldots,n_\text{proj}$}
\State Draw a vector $V$ uniformly from the $d$-dimensional unit sphere. To obtain $V$, \par
\hskip 1pt draw $X \sim N(0, I_d)$ and set $V = X / \|X\|$.
\State Project each $Y$ observation onto $V$. The projection of $Y_i$ is $P_V(Y_i) = Y_i^T V$.
\For {$b=1,2,\ldots,B$}
\State Randomly partition $[n]$ into $\mathcal{D}_{0,b}$ and $\mathcal{D}_{1,b}$ such that $|\mathcal{D}_{0,b}| = |\mathcal{D}_{1,b}| = n/2$.
\State Estimate a one-dimensional density $\hat{f}_{1,b,k}$ on $\{P_V(Y_i): i\in\mathcal{D}_{1,b}\}$.
\State Estimate the log-concave MLE $\hat{f}_{0,b,k}$ on $\{P_V(Y_i): i \in\mathcal{D}_{0,b}\}$.
\EndFor
\State Compute the test statistic $T_{n,k} = B^{-1} \sum_{b=1}^B \prod_{i \in \mathcal{D}_{0,b}} \{\hat{f}_{1,b,k}(P_V(Y_i)) / \hat{f}_{0,b,k}(P_V(Y_i)) \}$.
\If{$n_\text{proj}^{-1} \sum_{j=1}^k T_{n,j} \geq 1/\alpha$} \textbf{stop}, reject null. \EndIf
\EndFor
\State \textbf{return} the random projection test statistic $T_n = n_\text{proj}^{-1} \sum_{j=1}^{n_\text{proj}} T_{n,j}$.
\end{algorithmic}
\label{alg:randproj}
\end{algorithm}

We expect random projections (with averaging) to work better when the deviations from log-concavity are ``dense,'' meaning there is a small amount of evidence to be found scattered in different directions. In contrast, the axis-aligned projections (with Bonferroni) presented earlier are expected to work better when there is a strong signal along one or a few dimensions, with most dimensions carrying no evidence (meaning that the density is indeed log-concave along most axes).

\section{Example: Testing Log-concavity of Normal Mixture} \label{sec:example}

We test the permutation approach and the universal approaches on a normal mixture distribution, which is log-concave only at certain parameter values. Naturally, when testing or fitting log-concave distributions in practice, one would eschew all parametric assumptions, so the restriction to normal mixtures is simply for a nice simulation example.  See Appendix~\ref{app:beta} for another such example over Beta densities. 

Let $\phi_d$ be the $N(0, I_d)$ density. \cite{cule2010} note a result that we state in Fact~\ref{fact:logconcmixture}.
\begin{restatable}[]{fact}{factLogconcmixture} \label{fact:logconcmixture}
For $\gamma \in (0, 1)$, the normal location mixture $f(x) = \gamma \phi_d(x) + (1-\gamma) \phi_d(x - \mu)$ is log-concave only when $\|\mu\| \leq 2$.
\end{restatable}

\cite{cule2010} use the permutation test to test $H_0: f^*\in\mathcal{F}_d$ versus $H_1: f^*\notin\mathcal{F}_d$ in this setting. We explore the power and validity of both the permutation and the universal tests over $\gamma = 0.5,$ varying $\|\mu\|,$ and dimensions $d\in\{1,2,3,4\}$. 

\subsection{Full Oracle (in $d$ dimensions) has Inadequate Power} \label{sec:fulloracle}


We compare the permutation test from Section~\ref{sec:permutation} to the full oracle universal test, which uses the true density in the numerator, from Section~\ref{sec:ddim}. Figure~\ref{fig:perm_randproj_n100} showed that the permutation test is not valid in this setting for $d\geq 4$, $n=100$, and $B = 99$. In Appendix~\ref{app:perm_test_addl}, we show that the permutation test's rejection proportion is similar if we increase $B$ to $B\in\{100, 200, 300, 400, 500\}$. In addition, we show that if we increase $n$ to 250, the rejection proportion is still much higher than 0.10 for $\|\mu\| \leq 2$ at $d = 4$ and $d = 5$. In the same appendix, we also show that the discrete nature of the test statistic is not the reason for the test's conservativeness (e.g., $d=1$) or anticonservativeness (e.g., $d=5$).

To compare the permutation test to the full oracle universal test, we again set $\mu = (\mu_1, 0, \ldots, 0)$. Figure~\ref{fig:reject_perm_true_n100} shows the power for $d\in \{1, 2, 3, 4\}$ on $n=100$ observations. For the universal test, we subsample $B=100$ times. Each point is the rejection proportion over 200 simulations. For some $\|\mu\|$ values in the $d=1$ case, the full oracle test has higher power than the permutation test.  For most $(d, \|\mu\|)$ combinations, though, the full oracle test has lower power than the permutation test. Unlike the permutation test, though, the universal test is provably valid for all $d$. 
In Figure \ref{fig:reject_perm_true_n100}, we see that as $d$ increases, we need larger $\|\mu\|$ for the universal test to have power. More specifically, $\|\mu\|$ needs to grow exponentially with $d$ to maintain a certain level of power. (See Figure~\ref{fig:power_vary_d} in Appendix~\ref{app:power_mu_d}.)

\begin{figure}[ht]
\begin{center}
\includegraphics[scale=.65]{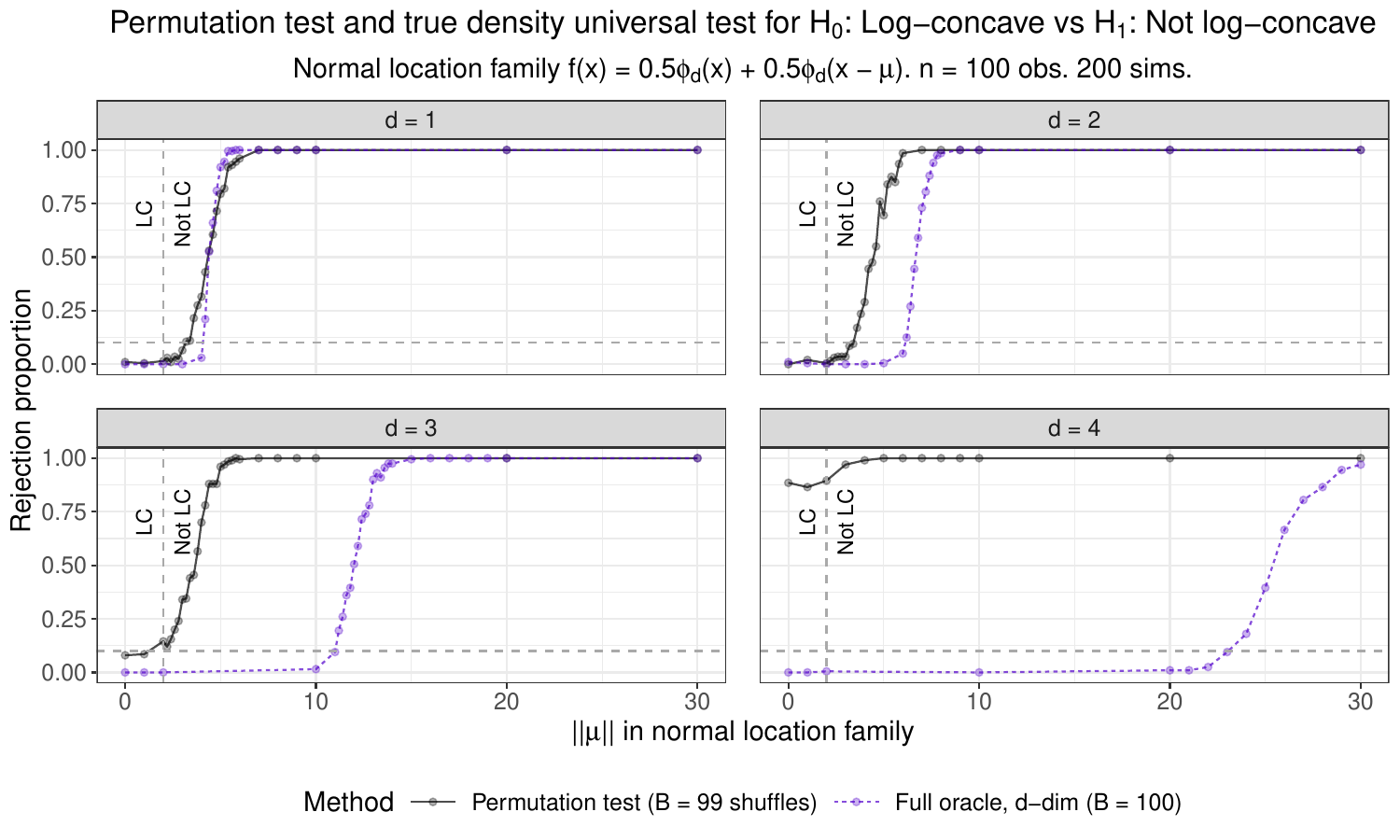}
\end{center}
\caption{Rejection proportions for tests of $H_0: f^*$ is log-concave versus $H_1: f^*$ is not log-concave. When $d=1$, the permutation test and the full oracle universal test have similar power. The full oracle universal approach remains valid in higher dimensions, but it has low power for moderate $\|\mu\|$ when $d\geq 3$.}
\label{fig:reject_perm_true_n100}
\end{figure} 

In Appendix~\ref{app:trace}, we briefly discuss the trace test from Section~3 of \cite{chen2013smoothed} as an alternative to the permutation test. Similar to the permutation test, simulations suggest that the trace test is valid in the $d=2$ setting, but it does not control type I error in the general $d$-dimensional setting.

\subsection{Superior Performance of Dimension Reduction Approaches}

We have seen that the full oracle universal LRT requires $\|\mu\|$ to grow exponentially to maintain power as $d$ increases. We turn to the dimension reduction universal LRT approaches, and we find that they produce higher power for smaller $\|\mu\|$ values.

We implement all four combinations of the two dimension reduction approaches (axis-aligned and random projections) and two density estimation methods (partial oracle and fully nonparametric). We compare them to three $d$-dimensional approaches: the permutation test, the full oracle test, and the partial oracle test. The full oracle $d$-dimensional approach uses the split LRT with the true density in the numerator and the $d$-dimensional log-concave MLE in the denominator. The partial oracle approaches use the split LRT with a two component Gaussian mixture in the numerator and the log-concave MLE in the denominator. We fit the Gaussian mixture using the EM algorithm, as implemented in the \texttt{mclust} package \citep{mclust}. The fully nonparametric approaches use a kernel density estimate in the numerator and the log-concave MLE in the denominator. We fit the kernel density estimate using the plug-in bandwidth in the \texttt{ks} package \citep{ks}.

Figures~\ref{fig:unequal_spaced_both} and \ref{fig:equal_spaced_both} compare the four dimension reduction approaches and the $d$-dimensional approaches. The six universal approaches subsample at $B = 100$. The random projection approaches set $n_{\text{proj}} = 100$. The permutation test uses $B = 99$ permutations to determine the significance level of the original test statistic. Both figures use the normal location model $f^*(x) = 0.5 \phi_d(x) + 0.5 \phi_d(x -\mu)$ as the underlying model. However, Figure~\ref{fig:unequal_spaced_both} uses $\mu = -(\|\mu\|, 0, \ldots, 0)$, while Figure~\ref{fig:equal_spaced_both} uses $\mu = -(\|\mu\|d^{-1/2}, \|\mu\|d^{-1/2}, \ldots, \|\mu\|d^{-1/2})$. The axis-aligned projection method has higher power in the first setting, but the other methods do not have differences in power between the two settings.

Figures~\ref{fig:unequal_spaced} and \ref{fig:equal_spaced} compare all seven methods. There are several key takeaways. The universal approaches that fit one-dimensional densities (axis-aligned projections and random projections) have higher power than the universal approaches that fit $d$-dimensional densities. (When $d=1$, the ``Partial oracle, axis-aligned projections,'' ``Partial oracle, $d$-dim,'' and ``Partial oracle, random projections'' approaches are the same, except that the final method uses $Bn_\text{proj}$ subsamples rather than $B$ subsamples.) As intuition for this behavior, even one projection with particularly strong evidence against log-concavity may provide sufficient evidence to reject log-concavity. See Appendix~\ref{app:full_vs_proj} for more discussion. The permutation test is not always valid, especially for $d \geq 4$. 

To compare the four universal approaches that fit one-dimensional densities, we consider Figures~\ref{fig:unequal_spaced_zoom} and \ref{fig:equal_spaced_zoom}, which zoom in on a smaller range of $\|\mu\|$ values for those four methods. In both Figure~\ref{fig:unequal_spaced_zoom} and \ref{fig:equal_spaced_zoom}, for a given dimension reduction approach (axis-aligned projections or random projections), the partial oracle approach has slightly higher power than the fully nonparametric approach. For a given density estimation approach, the dimension reduction approach with higher power changes based on the setting. When $\mu = -(\|\mu\|, 0, \ldots, 0)$ (Figure~\ref{fig:unequal_spaced_zoom}), the axis-aligned projection approach has higher power than the random projections approach. This makes sense because a single dimension contains all of the signal. When $\mu = -(\|\mu\|d^{-1/2}, \|\mu\|d^{-1/2}, \ldots, \|\mu\|d^{-1/2})$ (Figure~\ref{fig:equal_spaced_zoom}), the random projections approach has higher power than the axis-aligned projection approach. This makes sense because all directions have some evidence against $H_0$, and there exist linear combinations of the coordinates that have higher power than any individual axis-aligned dimension.

\begin{figure}
\centering
\begin{subfigure}[b]{1\textwidth}
\centering
\includegraphics[scale=.5]{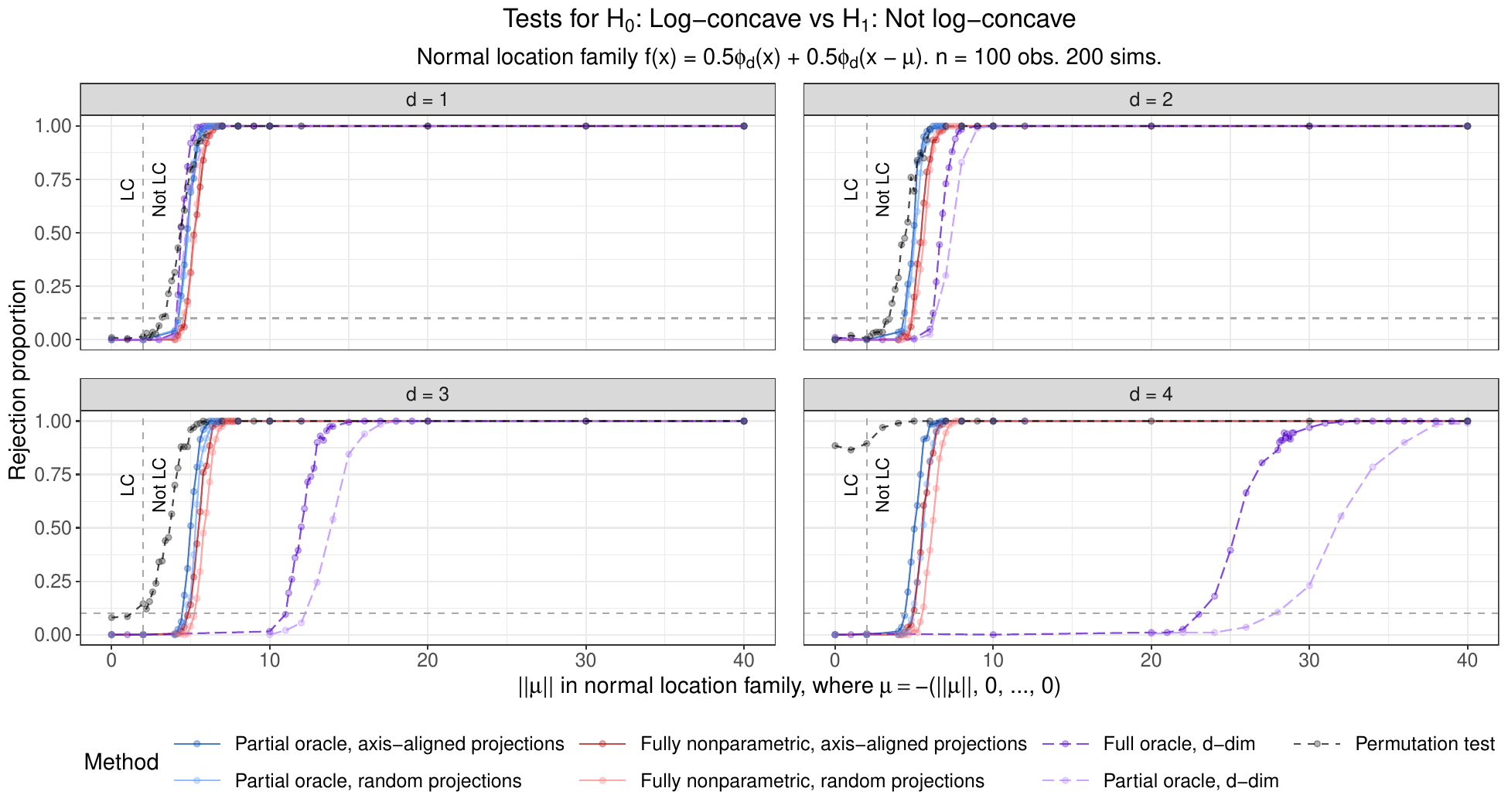}
\caption{The projection-based universal tests decrease the gap in power between the permutation test and the $d$-dimensional universal tests for $d\geq 2$. The power of the projection tests only exhibits a moderate curse of dimensionality.}
\label{fig:unequal_spaced}
\end{subfigure} \\
\begin{subfigure}[b]{1\textwidth}
\centering
\includegraphics[scale=.55]{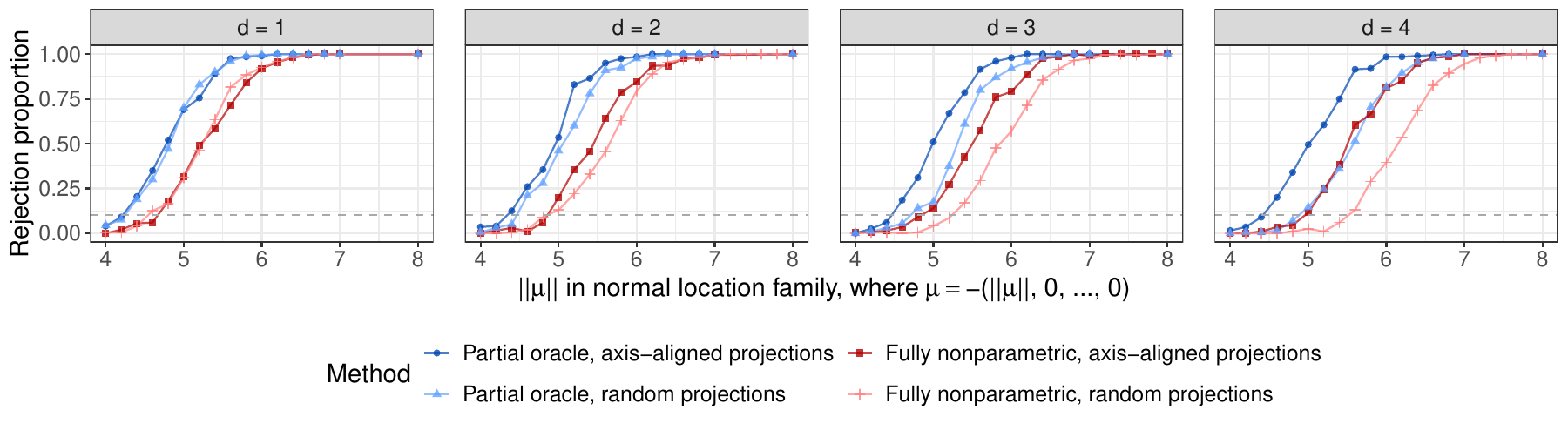}
\caption{Partial oracle numerators have higher power than fully nonparametric numerators. Since all signal against $H_0$ is in the first component of $\mu$, the axis-aligned projection tests have higher power than the random projection tests within each choice of numerator.}
\label{fig:unequal_spaced_zoom}
\end{subfigure} 
\caption{Power of tests of $H_0: f^*$ is log-concave versus $H_1: f^*$ is not log-concave. $\mu$ vector for second component is $\mu = -(\|\mu\|, 0, \ldots, 0)$.}
\label{fig:unequal_spaced_both}
\end{figure} 

\begin{figure}
\centering
\begin{subfigure}[b]{1\textwidth}
\centering
\includegraphics[scale=.5]{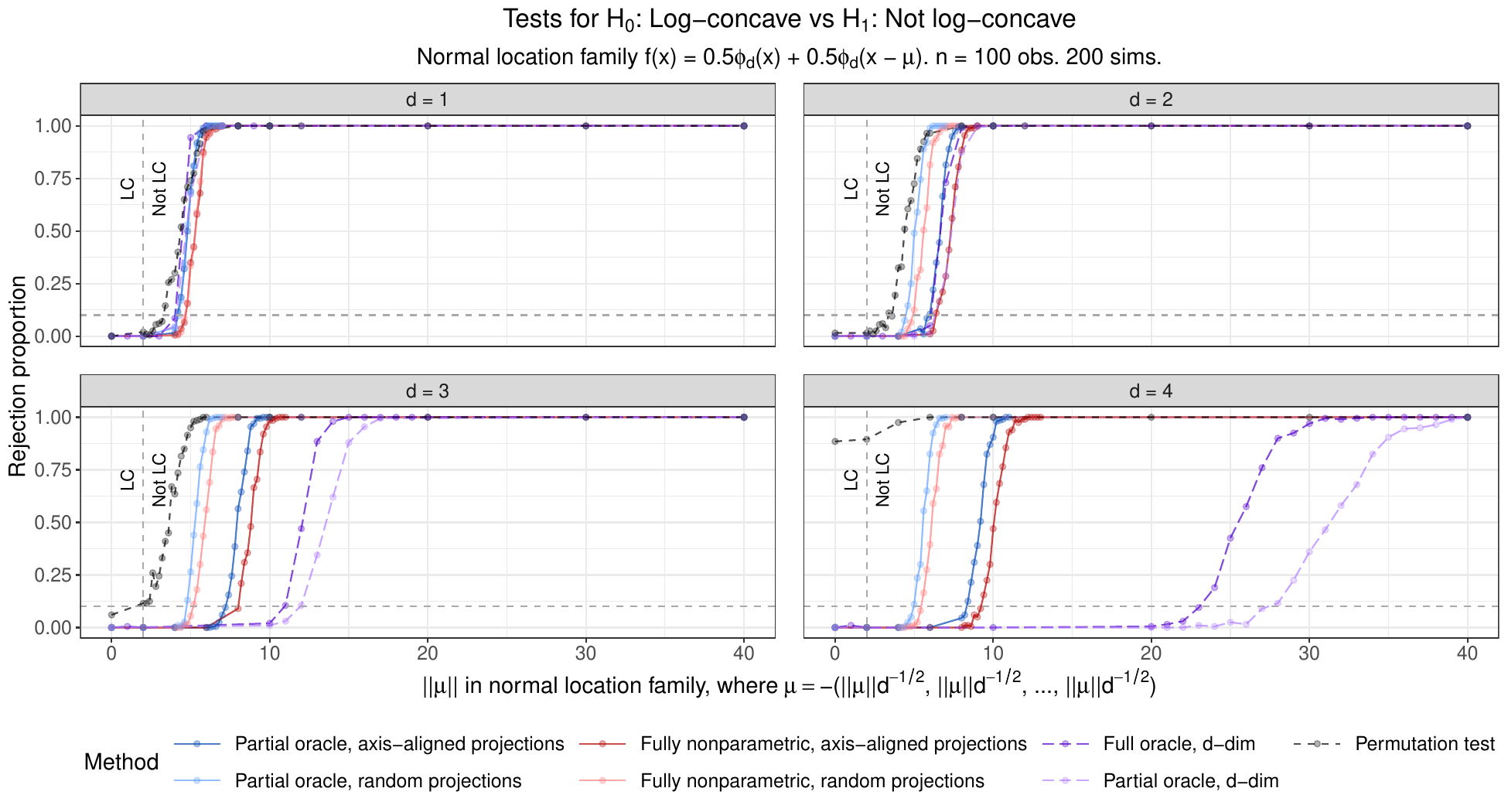}
\caption{In this second choice of $\mu$ structure, the projection-based universal tests also decrease the gap in power between the permutation test and the $d$-dimensional universal tests for $d\geq 2$. The power of the projection tests only exhibits a moderate curse of dimensionality.}
\label{fig:equal_spaced}
\end{subfigure} \\
\begin{subfigure}[b]{1\textwidth}
\centering
\includegraphics[scale=.55]{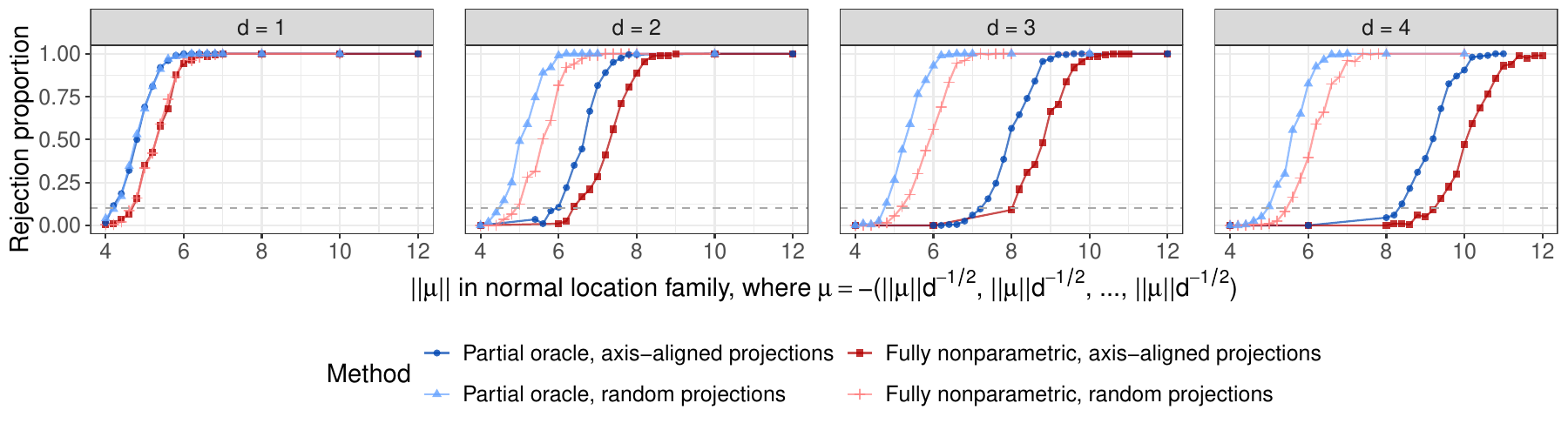}
\caption{Since the signal against $H_0$ is equally distributed across the components of $\mu$, the random projection tests have higher power than the axis-aligned projection tests.}
\label{fig:equal_spaced_zoom}
\end{subfigure} 
\caption{Power of tests of $H_0: f^*$ is log-concave versus $H_1: f^*$ is not log-concave. $\mu$ vector for second component is $\mu = -(\|\mu\|d^{-1/2}, \|\mu\|d^{-1/2}, \ldots, \|\mu\|d^{-1/2})$.}
\label{fig:equal_spaced_both}
\end{figure} 

\subsection{Time Benchmarking}

\begin{table}[h!]
\centering
\begin{tabular}{p{3.5cm}p{3cm}p{2.6cm}p{2.6cm}p{2.7cm}}
  \hline \hline
 Method & $d = 1$ & $d = 2$ & $d = 3$ & $d = 4$ \\ 
  \hline
 Partial oracle, \newline random projections & (150, 50, 0.94) & (160, 110, 1.7) & (150, 140, 3.2) & (160, 140, 3.7) \\ \hline
 Fully NP, \newline random projections & (120, 71, 0.67) & (110, 100, 1.3) & (120, 110, 2.6) & (120, 120, 4.2) \\ \hline
 Permutation test \newline & (51, 50, 51) & (52, 52, 52) & (54, 54, 53) & (61, 62, 62) \\ \hline
 Partial oracle, \newline axis projections & (8, 3.6, 0.1) & (16, 5.7, 0.19) & (24, 12, 0.26) & (32, 25, 0.34) \\ \hline
 Fully NP, \newline axis projections & (7.6, 5.8, 0.098) & (15, 11, 0.17) & (23, 19, 0.25) & (30, 23, 0.33) \\ \hline
 Partial oracle,\newline d-dim & (6.6, 3.9, 0.093) & (21, 22, 0.68) & (71, 74, 75) & (300, 330, 350) \\ \hline
 Full oracle,\newline d-dim & (6.1, 1.2, 0.073) & (20, 21, 0.3) & (70, 72, 74) & (300, 340, 340) \\ 
   \hline
\end{tabular}
\caption{Average run time (in seconds) of log-concave tests at $(\|\mu\| = 0, \|\mu\| = 5, \|\mu\| = 10)$, using the same parameters as Figure~\ref{fig:unequal_spaced_both}. The universal methods run faster when there is more evidence against $H_0$, which allows for early rejection. It is faster to compute log-concave MLEs in one dimension (e.g., projection methods) than in general $d$ dimensions.}
\label{tab:benchmarking}
\end{table}

Table~\ref{tab:benchmarking} displays the average run times of the seven methods that we consider. Each cell corresponds to an average (in seconds) over 10 simulations at $(\|\mu\| = 0, \|\mu\| = 5, \|\mu\| = 10)$. Except for the restriction on  $\|\mu\|$, these simulations use the same parameters as Figure~\ref{fig:unequal_spaced_both}. We arrange the methods in rough order from longest to shortest run times in $d=1$. The random projection methods have some of the longest run times at $\|\mu\| = 0$. If the random projection methods do not have sufficient evidence to reject $H_0$ early, they will construct $Bn_\text{proj}$ test statistics. Each of those $Bn_\text{proj}$ test statistics requires fitting a one-dimensional log-concave MLE and estimating a partial oracle or fully nonparametric numerator density. The permutation test is faster than the random projection tests in the $\|\mu\| = 0$ setting, but it does not stop early for larger $\|\mu\|$. The axis projection and $d$-dimensional universal approaches have similar run times for $d\leq 2$. (In fact, for $d=1$ the partial oracle axis projection and partial oracle $d$-dimensional methods are the same.) The axis projection methods compute a maximum of $Bd$ test statistics, and the $d$-dimensional methods compute a maximum of $B$ test statistics. Since the $d$-dimensional universal approaches repeatedly fit $d$-dimensional log-concave densities, they are the most computationally expensive approaches for large $d$. 

\subsection{Counterpoint: Non-log-concave Density with Log-concave Marginals}\label{sec:counterpoint}
In the previous normal location mixture example,
the projection methods have power because the projection distributions are not all log-concave. While Theorem~\ref{thm:projection} states that log-concavity of a density implies log-concavity of the lower dimensional projections, the converse is not guaranteed to hold.  As an example, suppose $f^*: \mathbb{R}^2 \to \mathbb{R}$ is the normal mixture density given by $f^*(x) = (2 \cdot 2\pi)^{-1} (\exp(-\|x\|^2/2) + \sigma^{-2} \exp(-\|x\|^2/2\sigma^2))$, where $\sigma = \sqrt{3}$. This density is not log-concave, but all of its one-dimensional projections are log-concave. (See Appendix~\ref{app:gaussianproj}.) Figure~\ref{fig:normal_mixture_diff_variance} shows power results from simulations of several full-dimensional and projection approaches to test $H_0: f^* \in \mathcal{F}_d$ versus $H_1: f^* \notin \mathcal{F}_d$ across varying $n$. The $d$-dimensional approaches have some power to detect $H_1$, but even at $n = 1600$ the estimated power is only about 0.13. The projection methods do not have power since all projections are log-concave. Hence, the projection methods will not always have higher power than the $d$-dimensional approaches, but in this example even the $d$-dimensional approaches do not have high power.
If one is not sure whether the projected or full-dimensional test will have higher power, one can simply run both, average the resulting test statistics, and threshold the average at $1/\alpha$. Since the average of e-values is an e-value, such a test is still valid, and the test is consistent if either of the original tests is consistent.

\begin{figure}[H]
\centering
\includegraphics[scale=.65]{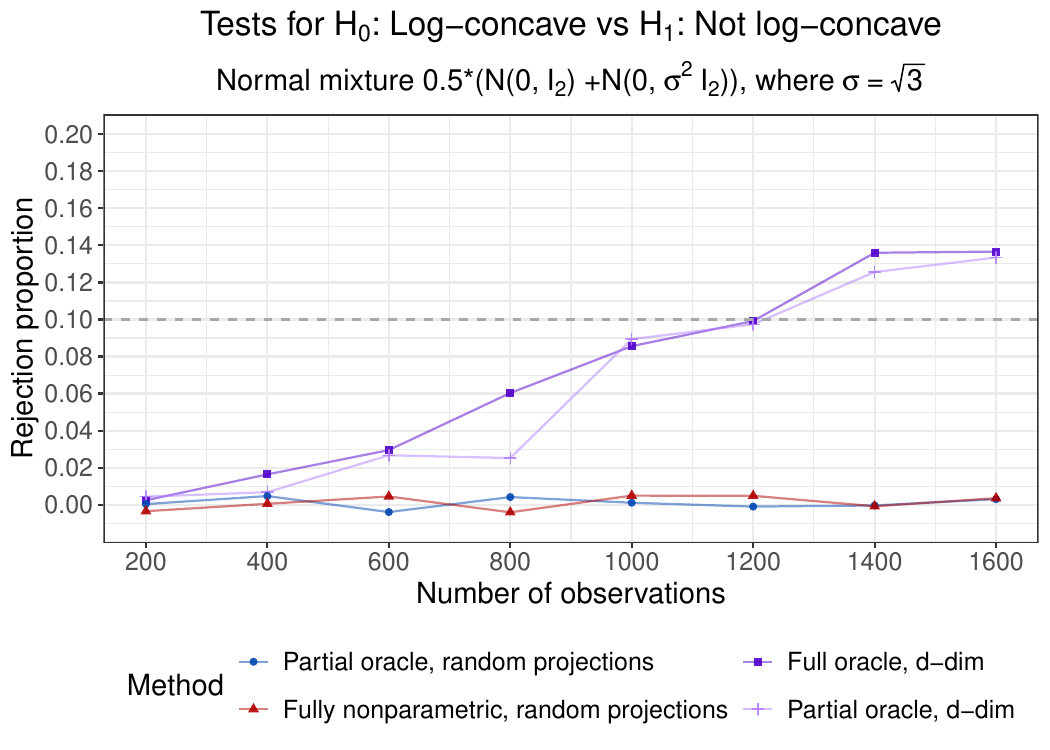}
\caption{Power of the universal test for log-concavity in a setting where the two-dimensional density is not log-concave but all one-dimensional projections are log-concave. The estimated power is the rejection proportion over 200 simulations at $\alpha = 0.10$ for a given value of $n$. Values of rejection proportions are jittered by at most 0.005 for plotting. Full-dimensional approaches have power above $\alpha$ at $n\in\{1400, 1600\}$. Projection approaches never reject log-concavity.}
\label{fig:normal_mixture_diff_variance}
\end{figure} 

\newcommand{\onehatf}{\hat{f}_{1,n}}
\newcommand{\zerohatf}{\hat{f}_{0,n}}
\newcommand{\flc}{f^{\text{LC}}}
\newcommand{\dee}{\mathcal{D}}
\newcommand{\PP}{\mathbb{P}}
\newcommand{\PPP}[1]{\mathbb{P}\left( #1\right)}
\newcommand{\dkl}{D_{\mathrm{KL}}}

\section{Theoretical Power of Log-concave Universal Tests} \label{sec:theoretical}

Our simulations have shown that the universal LRTs can have high power to test $H_0: f^* \in \mathcal{F}_d$ versus $H_1: f^* \notin \mathcal{F}_d$. We complement this with a proof of the consistency of this test. For the rest of this section, think of $f^* \notin \mathcal{F}_d$.

First, we review and introduce some notation.  Let $\onehatf$ be an estimate of $f^*$, fit on $\dee_{1,n}$. Let $\zerohatf$ be the log-concave MLE of $f^*$ fit on $\dee_{0,n},$ i.e., $\zerohatf = \argmax{f\in\mathcal{F}_d} \sum_{i \in \dee_{0,n}} \log(f(Y_i))$. The universal test statistic is
\begin{equation}
    T_n = \prod_{i \in \mathcal{D}_{0,n}} \frac{\onehatf(Y_i)}{\zerohatf(Y_i)}, \label{eq:teststat}
\end{equation}
and we reject $H_0$ if $T_n \geq 1/\alpha$. 

Let $\flc$ denote the log-concave projection of $f^*$, i.e., $\flc = \argmin{f\in\mathcal{F}_d} \dkl(f^* \| f),$ where for densities $p,q, \dkl(p\|q)$ is the KL divergence $\dkl(p\|q) := \int p(x) \log (p(x) / q(x))\,\mathrm{d}x.$ We further define the Hellinger divergence $h(p,q) := \|\sqrt{p/2} - \sqrt{q/2}\|_2.$ $h$ is well-defined for nonnegative functions (and not only densities), and is a metric on such functions.


Finally, we define some set notation: $\mathrm{supp}(f^*) := \{ x : f^*(x) > 0\} $ denotes the support of the measure induced by $f^*,$ and for a measurable set $S \subset \mathbb{R}^d,$ $\mathrm{int}(S)$ and $\partial S$ respectively denote its topological interior and boundary.

\subsection{Assumptions}

We reiterate that the validity of the universal LRT only requires the assumption of iid data. However, for the universal LRT to be powerful (when $f^* \notin \mathcal{F}_d$), we need a few relatively mild conditions.

\begin{assumption}\label{assumption:regularity}
\emph{(Regularity of $f^*$) } Let $P^*(\mathrm{d}x) := f^*(x) \mathrm{d}x.$ We assume that if $X \sim f^*,$ then $\mathbb{E}[\|X\|] < \infty$, $\mathbb{E}[ \max(\log f^*(X), 0) ] < \infty,$ $\mathrm{int}(\mathrm{supp}(f^*)) \neq \varnothing,$ and that for every hyperplane $H$, $P^*(H) < 1.$ 
\end{assumption}

Assumption~\ref{assumption:regularity} enforces standard conditions imposed in the log-concave estimation literature \citep{cule2010theoretical}. In particular, the finiteness of $\mathbb{E}[\|X\|]$ and $\mathbb{E}[\max(\log f^*,0)]$ yield the existence of the log-concave projection $\flc$ (defined earlier), and the remaining conditions impose weak regularity properties that ensure convergence of $\zerohatf$ to $\flc$. 

For power, we fundamentally need our estimate $\onehatf$ to be close enough to the true (non-log-concave, under the alternative) distribution $f^*$, relative to $\mathcal{F}_d$. An assumption of KL-consistency, that is, assuming $\dkl(f^*\|\onehatf) \to 0$ as $n \to \infty$, would certainly suffice. The following condition is weaker, though, and it only requires $\dkl(f^*\|\onehatf)$ to get smaller than a critical Hellinger distance of $f^*$ from log-concavity.

\begin{assumption}\label{assumption:estimability}
\emph{(Estimability of $f^*$) } We assume that $\onehatf$ is a good estimator of $f^*$, in the sense that if we use $n/2$ iid draws from $f^*$ to construct $\onehatf$, then for all $\theta > 0,$ \[\lim_{n\to\infty} \PPP{ \frac{\dkl(f^*\|\onehatf)}{h^2(f^*, \flc)} < \frac{1}{200}, \int f^*(x) \log^2(f^*(x) / \onehatf(x))\,\mathrm{d}x < \theta n} = 1.\]
\end{assumption}

Assumption~\ref{assumption:estimability} is satisfied if $\onehatf$ estimates $f^*$ in a KL divergence sense better than $\flc$ approximates $f^*$ in a Hellinger sense, and if the variance of the log-ratio $\log (f^* / \onehatf)$ does not grow too fast. Estimation in KL divergence is largely driven by the tail behavior of $f^*$; the $\onehatf$ estimation procedure needs to ensure that $\onehatf$ does not underestimate the tails of $f^*$. In many settings, we will have $\dkl(f^*\|\onehatf) \to 0$, thus satisfying the assumption. However, we do not require $\dkl(f^*\|\onehatf)$ to go to $0$ --- it is enough for the divergence to get smaller than $h^2(f^*,\flc)/200$. While this criterion is still stringent enough to be practically relevant, it makes the theoretically favorable point that the universal LRT for log-concavity does not require consistent density estimation of $f^*$ in a strong (KL) sense. In our argument, the role of Assumption \ref{assumption:estimability} is to control how large $f^*/\onehatf$ can get in a manner similar to Theorems~3 and 4 of \cite{wong1995probability}.

\subsection{Result and Proof Sketch}

We are now in a position to state our main result. We provide a proof sketch, leaving the details to Appendix~\ref{app:consistency}.

\begin{restatable}[]{thm}{thmLogconcpower} \label{thm:logconcpower}

Suppose Assumptions \ref{assumption:regularity} and \ref{assumption:estimability} hold. Then the universal likelihood ratio test for log-concavity with test statistic (\ref{eq:teststat}) is consistent. That is, $ \lim_{n \to \infty} \PP_{H_1}(T_n \geq 1/\alpha) = 1$.

\end{restatable}

\begin{proofsketch}

We assume throughout that $H_1$ is true, meaning that $f^*$ is not log-concave. For brevity, we use $\PP$ to denote $\PP_{H_1}$. We begin by decomposing $T_n$ into \[ T_n = \underbrace{\prod_{i \in \mathcal{D}_{0,n}} \frac{\hat{f}_{1,n}(Y_i)}{f^*(Y_i)}}_{=: 1/R_n} \cdot \underbrace{\prod_{i \in \mathcal{D}_{0,n}} \frac{f^*(Y_i)}{\hat{f}_{0,n}(Y_i)}}_{=: S_n} = S_n/R_n.\]

Let $\varepsilon := h(f^*, \flc).$ Further, suppose $n \geq 100 \log(1/\alpha)/\varepsilon^2.$ Now observe that \begin{align*}
    \{ T_n < 1/\alpha\} &\subset \{R_n > \exp(n\varepsilon^2/100) \} \cup \{S_n < \exp(n\varepsilon^2/50) \},
\end{align*} since outside this union, $S_n/R_n \geq \exp(n\varepsilon^2/100) \geq 1/\alpha$. Thus, it suffices to argue that \begin{equation*} 
    \PPP{R_n > \exp(n \varepsilon^2/100}) + \PPP{S_n < \exp(n \varepsilon^2/50}) \to 0.
\end{equation*} 
The two assumptions contribute to bounding these terms. In particular, Assumption~\ref{assumption:regularity} implies that $S_n$ is big, while Assumption~\ref{assumption:estimability} implies that $R_n$ is small.

\textbf{$S_n$ is big}. Observe that since $h(f^*, \flc) > 0,$ the likelihood ratio $\prod_{i \in \dee_{0,n}} f^*(Y_i) / \flc(Y_i)$ tends to be exponentially large with high probability. We can use Markov's inequality and the properties of Hellinger distance to show that for $\xi > 0$ (and particularly for large $\xi$), \[\PPP{\prod_{i \in \dee_{0,n}} f^*(Y_i)/\flc(Y_i) < \xi} \le \sqrt{\xi} \mathbb{E}_{Y \sim f^*}\left[ \sqrt{\flc(Y)/f^*(Y)}\right]^{n/2} = \sqrt{\xi} (1-h^2(f^*, \flc))^{n/2}.\] We would thus expect that the same holds for the ratio of interest $\prod_{i \in \dee_{0,n}} f^*(Y_i)/\zerohatf(Y_i)$. The regularity conditions from Assumption \ref{assumption:regularity} enable this, by establishing that $\zerohatf \to \flc$ in the strong sense that for large $n$, $\zerohatf$ lies in a small bracket containing $\flc$. This pointwise control allows us to use results from empirical process theory to show that for large enough $n$, $\prod_{i \in \dee_{0,n}} f^*(Y_i)/\zerohatf(Y_i)$ grows at an exponential rate similar to $\prod_{i \in \dee_{0,n}} f^*(Y_i)/\flc(Y_i)$.

\textbf{$R_n$ is small}. The smallness of $R_n$ relies on two facts: (1) $\onehatf$ is fit on $\{Y_i : i \in \dee_{1,n}\}$ and evaluated on an independent dataset $\{Y_i: i \in \dee_{0,n}\},$ and (2) under Assumption~\ref{assumption:estimability}, $\onehatf$ approximates $f^*$ in a strong sense with high probability. 

Concretely, we observe that since $\onehatf$ is determined given the data in $\dee_{1,n},$ we may condition on $\dee_{1,n}$ and study the tail behavior of $R_n$ on $\dee_{0,n}.$ Further observing that $\mathbb{E}_{Y \sim f^*}[\log (f^*(Y)/\onehatf(Y))] =  D(f^*\|\onehatf),$ applying Tchebycheff's inequality to $\log R_n$ yields \[ \PPP{R_n > \exp(n \varepsilon^2/100) \,|\, \{Y_i : i \in \dee_{1,n}\}} \le \frac{\int f^*(x) \log^2(f^*(x)/\onehatf(x))\mathrm{d}x}{ n ((\varepsilon^2/50 - D(f^*\|\onehatf))_+)^2},\] where $(z)_+ = \max(0,z)$. Assumption~\ref{assumption:estimability} lets us argue that with high probability over $\{Y_i : i \in \mathcal{D}_{1,n}\},$ this upper bound vanishes as $n \to \infty$.

\end{proofsketch} 

\section{Conclusion} \label{sec:conclusion}

We have implemented and evaluated several universal LRTs to test for log-concavity. These methods provide the first tests for log-concavity that are valid in finite samples under only the assumption that the data sample is iid. The tests include a full oracle (true density) approach, a partial oracle (parametric) approach, a fully nonparametric approach, and several LRTs that reduce the $d$-dimensional test to a set of one-dimensional tests. 
For reference,
we compared these tests to a permutation test 
although that test is not guaranteed to be valid. In one dimension, the universal tests can have higher power than the permutation test. In higher dimensions, the permutation test may falsely reject $H_0$ at a rate much higher than $\alpha$, but the universal tests are still valid in higher dimensions. As seen in the Gaussian mixture case, dimension reduction universal approaches can have notably stronger performance than the universal tests that work with $d$-dimensional densities. 

Several open questions remain. 
Theorem~\ref{thm:logconcpower} presented a set of conditions under which the universal LRT has power that converges to 1 as $n\to\infty$. As discussed, it may be possible to weaken some of these conditions. In addition,  future work may seek to theoretically derive the power as a function of the dimension, number of observations, and signal strength. As shown in one example (Figure~\ref{fig:power_vary_d} of Appendix~\ref{app:addl_sims}), the signal may need to grow exponentially in $d$ to maintain the same power.

\bigskip
\begin{center}
{\large\bf SUPPLEMENTARY MATERIAL}
\end{center}

\begin{description}

\item[Appendix:] The appendix contains proofs of all theoretical results (Appendix~\ref{app:proofs}), additional simulations and visualizations for the two-component normal mixture setting (Appendix~\ref{app:addl_sims}), discussions on the relative power of full-dimensional and projection tests (Appendix~\ref{app:full_vs_proj}), simulations to test log-concavity when data arise from a Beta distribution (Appendix~\ref{app:beta}), and additional details on the permutation test and trace test for log-concavity (Appendix~\ref{app:perm_trace}). (pdf file)

\item[R code:] The R code to reproduce the simulations and figures is available at
\if0\blind{\url{https://github.com/RobinMDunn/LogConcaveUniv}.}\fi
\if1\blind{[redacted for blind review].}\fi

\end{description}

\bigskip
\begin{center}
{\large\bf ACKNOWLEDGEMENTS}
\end{center}
\if0\blind{This work used the Extreme Science and Engineering Discovery Environment (XSEDE) \citep{xsede}, which is supported by National Science Foundation grant number ACI-1548562. Specifically, it used the Bridges system \citep{xsedebridges}, which is supported by NSF award number ACI-1445606, at the Pittsburgh Supercomputing Center (PSC). This work made extensive use of the \texttt{R} statistical software \citep{Rcore}, as well as the \texttt{clustermq} \citep{clustermq}, \texttt{data.table} \citep{datatable}, \texttt{fitdistrplus} \citep{fitdistrplus}, \texttt{kde1d} \citep{kde1d}, \texttt{ks} \citep{ks},  \texttt{LogConcDEAD} \citep{cule2009logconcdead}, \texttt{logcondens} \citep{logcondens}, \texttt{MASS} \citep{mass}, \texttt{mclust} \citep{mclust}, \texttt{mvtnorm} \citep{mvtnorm, mvtnormbook}, and \texttt{tidyverse} \citep{tidyverse} packages.}\fi
\if1\blind{[Redacted for blind review]}\fi

\bigskip
\begin{center}
{\large\bf FUNDING}
\end{center}
\if0\blind{RD is currently employed at Novartis Pharmaceuticals Corporation. This work was primarily conducted while RD was at Carnegie Mellon University. RD's research was supported by the National Science Foundation Graduate Research Fellowship Program under Grant Nos.\ DGE 1252522 and DGE 1745016. AR's research is supported by the National Science Foundation under Grant Nos.\ DMS (CAREER) 1945266 and DMS 2310718. Any opinions, findings, and conclusions or recommendations expressed in this material are those of the authors and do not necessarily reflect the views of the National Science Foundation.}\fi
\if1\blind{[Redacted for blind review]}\fi

\bibliographystyle{jasa3}

\bibliography{references}

\newpage

\appendix

\setcounter{page}{1}

\section{Proofs of Theoretical Results} \label{app:proofs}

We provide proofs for theoretical statements made in the main text. For convenience, statements are reproduced. 

\subsection{Validity of the Universal Likelihood Ratio Test} \label{app:validity}

\thmValidNonpar*

\begin{proof}
This result is due to \cite{wasserman2020universal}. First, we use only the data $\{Y_i: i \in \mathcal{D}_1\}$ to fit a density $\hat{f}_1$. Let $\mathcal{M}^*$ be the support of the distribution $P^*$ with density $f^*$, and let $\hat{\mathcal{M}}_1$ be the support of the distribution with density $\hat{f}_1$. We see
\begin{align*}
\mathbb{E}_{P^*}&\left[T_n(f^*) \mid \{Y_i\}_{i\in\mathcal{D}_1}\right] = \mathbb{E}_{P^*}\left[\frac{\mathcal{L}_0(\hat{f}_1)}{\mathcal{L}_0(f^*)} \: \Bigg| \: \{Y_i\}_{i\in\mathcal{D}_1} \right] = \mathbb{E}_{P^*}\left[\prod_{i\in\mathcal{D}_0} \frac{\hat{f}_1(Y_i)}{f^*(Y_i)} \: \Bigg| \: \{Y_i\}_{i\in\mathcal{D}_1} \right] \\
&\overset{iid}{=} \prod_{i\in\mathcal{D}_0} \mathbb{E}_{P^*}\left[\frac{\hat{f}_1(Y_i)}{f^*(Y_i)} \: \Bigg| \: \{Y_i\}_{i\in\mathcal{D}_1} \right] = \prod_{i \in \mathcal{D}_0} \left\{\int_{\mathcal{M}^*} \frac{\hat{f}_1(y_i)}{f^*(y_i)} f^*(y_i) dy_i \right\} = \prod_{i \in \mathcal{D}_0} \left\{\int_{\mathcal{M}^*} \hat{f}_1(y_i) dy_i \right\} \\
&\leq  \prod_{i\in\mathcal{D}_0} \left\{\int_{\hat{\mathcal{M}}_1} \hat{f}_1(y_i) dy_i \right\} = 1.
\end{align*}
This implies that $\mathbb{E}_{P^*}[T_n(f^*)] = \mathbb{E}_{P^*} \left[ \mathbb{E}_{P^*}\left[ T_n(f^*) \mid \{Y_i\}_{i\in\mathcal{D}_1} \right] \right] \leq 1$. Furthermore, recall that $T_n(f^*) = \mathcal{L}_0(\hat{f}_1) / \mathcal{L}_0(f^*)$ and $T_n(\hat{f}_0) = \mathcal{L}_0(\hat{f}_1) / \mathcal{L}_0(\hat{f}_0)$, where $\hat{f}_0 = \argmax{f\in\mathcal{F}_d} \mathcal{L}_0(f)$. Thus, under $H_0: f^*\in\mathcal{F}_d$, it holds that $T_n(\hat{f}_0) \leq T_n(f^*)$ and $\mathbb{E}_{P^*}[T_n(\hat{f}_0)] \leq 1$.  

Applying Markov's inequality and the above fact, under $H_0$,
\begin{equation*}
\mathbb{P}_{P^*}\left(T_n(\hat{f}_0) \geq 1/\alpha \right) \leq \alpha \mathbb{E}_{P^*}[T_n(\hat{f}_0)] \leq \alpha.
\end{equation*}

\end{proof}

\subsection{Non-log-concave Gaussian mixture distributions with log-concave projections}  \label{app:gaussianproj}

We show that there exist $d$-dimensional Gaussian mixture densities that are not log-concave but whose lower dimensional projections are all log-concave.  This result heavily exploits the property that full rank affine transformations preserve log-concavity or lack thereof. This follows directly from the properties of concave functions, and we state it here for completeness. 
\begin{lem} \label{lem:affine_full_rank_preserves}
    Let $A \in \mathbb{R}^{d \times d}$ be any full rank matrix, and $b \in \mathbb{R}^d$ any vector. A random vector $X$ with a density in $\mathbb{R}^d$ is log-concave if and only if $AX + b$ is log-concave.  
\end{lem} 
\begin{proof}
    Suppose that $X\in\mathbb{R}^d$ is a random vector with a log-concave density. Then we can write its density as $p_X(x) = \exp(\varphi(x))$ for some concave function $\varphi$. Denoting $Y =  A X + b,$ we observe that \[ p_Y(y) = |\det A|^{-1} p_X( A^{-1}(y - b)) = |\det A|^{-1} \exp(\varphi( A^{-1}(y-b))). \] But the composition of concave and linear functions is concave, and so $p_Y$ is of the required form for $Y$ to have a log-concave density. In addition, $X = A^{-1}Y - A^{-1}b,$ and so a similar argument applies in the reverse.
\end{proof}

We now construct the family of densities that we use to prove the key result. Let \[ \gamma_{d, \sigma}(x) := \frac{1}{(2\pi \sigma^2)^{d/2}} \exp\left( - \|x\|^2/2\sigma^2\right)\] denote the isotropic centered $d$-dimensional Gaussian density with variance $\sigma^2 I_d$. We are concerned with the behavior of the family 
\begin{align*}
    \mathcal{G} &:= \{ f_{d, \sigma}(x) : d \in \mathbb{N}, \sigma > 1\}, \textrm{ where }
    f_{d,\sigma}(x) := \frac12 ( \gamma_{d, 1}(x) + \gamma_{d, \sigma}(x)).
\end{align*}
For this family, we show the result of Theorem~\ref{thm:prop_gaussian}. \begin{thm}\label{thm:prop_gaussian}
For any natural number $d \geq 2$, there exists some $\sigma > 1$ such that (a) $f = f_{d,\sigma}$ is not log-concave and (b) if $X \sim f$ and $P: \mathbb{R}^d \to \mathbb{R}^j$ is any surjective affine map, then $PX$ has a log-concave law if and only if $j < d$.
\end{thm}

\begin{proof}
    
Firstly note that any affine projection can be seen as a composition of a translation and a projection. Since translations are affine maps preserving the dimension, by Lemma \ref{lem:affine_full_rank_preserves} it suffices to consider linear projections. Fix natural numbers $d\geq 2$ and $j < d$. In the following, we abuse notation and treat $P$ as a $\mathbb{R}^{j \times d}$ matrix with full row rank. 
    
Next, observe that if \( X \sim f_{d, \sigma},\) then due to the full row rank of $P$, $Y := PX$ is distributed as \[ Y \sim \frac12 \left(N(0, PP^T) + N(0, \sigma^2 PP^T)\right). \] This follows from the standard properties of multivariate Gaussians and the representation of mixtures using independent components. That is, we can write $X = B G + (1-B) \sigma H$  for mutually independent standard Gaussians $G,H$ and a fair bit $B \sim \text{Bernoulli}(1/2)$. Then $PX$ follows the stated distribution of $Y$.
    
$P$ is full row rank, so there exists some matrix $M \in \mathbb{R}^{j \times j}$ (specifically the canonical inverse square root $(P P^T)^{-1/2}$) of full rank such that \[ Z:= MY \sim \frac{1}{2} \left( N(0, I_j) + N(0, \sigma^2 I_j)\right).\] That is, $Z \sim f_{j,\sigma}$. 
(If $d = 2$, $j = 1$, and $P$ is a $1\times 2$ matrix whose squared entries sum to 1, then $PP^T = 1$ which means that $M = 1$.) Since $M$ is full rank, by Lemma \ref{lem:affine_full_rank_preserves}, $Y$ is log-concave if and only if $Z$ is log concave. Thus, to argue our claim, it suffices to prove that there exists $\sigma > 1$ such that $f_{d,\sigma}$ is not log-concave but $f_{j,\sigma}$ is log-concave.
    
The claim now follows by exploiting the following characterization of $f_{k, \sigma}$, $k\in\mathbb{N}$.   
    
\begin{lem}\label{lem:gaussian_analysis}
    There exists a strictly decreasing sequence $\{\tau_k\}_{k \in \mathbb{N}}$ taking values in $(1,\infty)$ such that $f_{k, \sigma}$ is log-concave if and only if $\sigma \le \tau_k.$
\end{lem}
    
Taking Lemma~\ref{lem:gaussian_analysis} as given, we can select any $\sigma \in (\tau_{d}, \tau_{d-1})$. This interval must exist since the sequence $\{\tau_k\}_{k \in \mathbb{N}}$ is strictly decreasing. At this choice of $\sigma$, $f_{d,\sigma}$ is not log-concave, but $f_{j,\sigma}$ is log-concave since $\sigma < \tau_{d-1} \leq \tau_j$.
\end{proof}

The argument thus concludes with the proof of Lemma \ref{lem:gaussian_analysis}, with which we now proceed.

\begin{proof}[Proof of Lemma \ref{lem:gaussian_analysis}\newline]

\textbf{Step 1:} \textit{Establish the form of the Hessian $h_{j,\sigma}(x)$ of $\log f_{j, \sigma}(x)$. $f_{j,\sigma}$ is log-concave if and only if this Hessian is negative semi-definite for all $x\in\mathbb{R}^j$.}

Since the density $f_{j,\sigma}$ is smooth, its log-concavity is determined by the Hessian $h_{j, \sigma}$ of its logarithm. Let $\nabla^2 f_{j,\sigma}(x)$ denote the Hessian matrix of $f_{j,\sigma}(x)$, and let $\nabla f_{j,\sigma}(x)$ denote the gradient vector of $f_{j,\sigma}(x)$. Then the Hessian of $\log f_{j,\sigma}$ is given by
$$h_{j, \sigma}(x) :=  \left[f_{j,\sigma}(x) \nabla^2 f_{j,\sigma}(x) - \nabla f_{j,\sigma}(x) \nabla f_{j,\sigma}(x)^T\right] / f_{j,\sigma}^2(x).$$
Thus, $f_{j,\sigma}$ is log-concave if and only if $h_{j, \sigma}(x)$ is negative semi-definite for all $x \in \mathbb{R}^j$.

\textbf{Step 2:} \textit{Construct another matrix that is negative semidefinite if and only if $h_{j,\sigma}$ is negative semidefinite.}

Observe that $\gamma_{j,\sigma}(x)$ has the following gradient vector $\nabla \gamma_{j,\sigma}(x)$ and Hessian matrix $\nabla^2 \gamma_{j, \sigma}(x)$:
\begin{align*}
    \nabla \gamma_{j,\sigma}(x) &= - \sigma^{-2} \gamma_{j,\sigma}(x) x\\
    \nabla^2 \gamma_{j, \sigma}(x) &= \sigma^{-4} \gamma_{j, \sigma}(x) xx^T - \sigma^{-2} \gamma_{j,\sigma}(x) I_j.
\end{align*}  
By linearity of differentiation, we see that 
\begin{align*}
   4 f_{j,\sigma}^2(x) h_{j, \sigma}(x) &= 4 \left[f_{j,\sigma}(x) \nabla^2 f_{j,\sigma}(x) - \nabla f_{j,\sigma}(x) \nabla f_{j,\sigma}(x)^T\right] \\
   &= 4 [\left\{\gamma_{j,1}(x)/2 + \gamma_{j,\sigma}(x)/2 \right\} \times \\
   &\qquad \left\{\gamma_{j, 1}(x) xx^T/2 - \gamma_{j,1}(x) I_j/2 + \sigma^{-4} \gamma_{j, \sigma}(x) xx^T/2 - \sigma^{-2} \gamma_{j,\sigma}(x) I_j/2 \right\} ] - \\
   &\qquad 4[\left\{- \gamma_{j,1}(x) x/2 - \sigma^{-2} \gamma_{j,\sigma}(x) x/2\right\} \left\{- \gamma_{j,1}(x) x^T/2 - \sigma^{-2} \gamma_{j,\sigma}(x) x^T/2 \right\}] \\
   &= \left\{\gamma_{j,1}(x) + \gamma_{j,\sigma}(x)\right\} \left\{\gamma_{j, 1}(x) + \sigma^{-4} \gamma_{j,\sigma}(x)\right\} xx^T - \\ 
   &\qquad \left\{\gamma_{j,1}(x) + \gamma_{j,\sigma}(x)\right\} \left\{\gamma_{j,1}(x) + \sigma^{-2} \gamma_{j,\sigma}(x)\right\} I_j - \left\{\gamma_{j, 1}(x) + \sigma^{-2} \gamma_{j, \sigma}(x)\right\}^2 xx^T.
\end{align*}
Hence,
\begin{align}
  &\frac{4 (1-\sigma^{-2})^{-2}}{\gamma_{j,1}(x) \gamma_{j,\sigma}(x)} f_{j,\sigma}^2(x) h_{j, \sigma}(x) \notag \\
  &= \frac{(1-\sigma^{-2})^{-2}}{\gamma_{j,1}(x) \gamma_{j,\sigma}(x)} \left[\left\{\gamma_{j,1}(x) + \gamma_{j,\sigma}(x)\right\} \left\{\gamma_{j, 1}(x) + \sigma^{-4} \gamma_{j,\sigma}(x)\right\} - \left\{\gamma_{j, 1}(x) + \sigma^{-2} \gamma_{j, \sigma}(x)\right\}^2 \right] xx^T \notag \\
  &\qquad - (1-\sigma^{-2})^{-2}\frac{\left\{\gamma_{j,1}(x) + \gamma_{j,\sigma}(x)\right\} \left\{\gamma_{j,1}(x) + \sigma^{-2}\gamma_{j,\sigma}(x) \right\}}{\gamma_{j,1}(x) \gamma_{j,\sigma}(x)} I_j \notag \\
  &= \frac{(1-\sigma^{-2})^{-2}}{\gamma_{j,1}(x) \gamma_{j,\sigma}(x)} \Big[\gamma_{j,1}^2(x) + \sigma^{-4} \gamma_{j,1}(x) \gamma_{j,\sigma}(x) + \gamma_{j,1}(x) \gamma_{j,\sigma}(x) + \sigma^{-4} \gamma_{j,\sigma}^2(x) - \notag \\
  &\hspace*{8em} \gamma_{j,1}^2(x) - 2\sigma^{-2}\gamma_{j,1}(x) \gamma_{j,\sigma}(x) - \sigma^{-4} \gamma_{j,\sigma}^2(x) \Big] xx^T \notag \\
  &\qquad - (1-\sigma^{-2})^{-2}\frac{\left\{\gamma_{j,1}(x) + \gamma_{j,\sigma}(x)\right\} \left\{\gamma_{j,1}(x) + \sigma^{-2}\gamma_{j,\sigma}(x) \right\}}{\gamma_{j,1}(x) \gamma_{j,\sigma}(x)} I_j \notag \\ 
  &= \frac{(1-\sigma^{-2})^{-2}}{\gamma_{j,1}(x) \gamma_{j,\sigma}(x)} \Big[\sigma^{-4} \gamma_{j,1}(x) \gamma_{j,\sigma}(x) + \gamma_{j,1}(x) \gamma_{j,\sigma}(x) - 2\sigma^{-2}\gamma_{j,1}(x) \gamma_{j,\sigma}(x) \Big] xx^T \notag \\
  &\qquad - (1-\sigma^{-2})^{-2}\frac{\left\{\gamma_{j,1}(x) + \gamma_{j,\sigma}(x)\right\} \left\{\gamma_{j,1}(x) + \sigma^{-2}\gamma_{j,\sigma}(x) \right\}}{\gamma_{j,1}(x) \gamma_{j,\sigma}(x)} I_j \notag \\
  &= xx^T - (1-\sigma^{-2})^{-2}\frac{\left\{\gamma_{j,1}(x) + \gamma_{j,\sigma}(x)\right\} \left\{\gamma_{j,1}(x) + \sigma^{-2}\gamma_{j,\sigma}(x) \right\}}{\gamma_{j,1}(x) \gamma_{j,\sigma}(x)} I_j. \label{eq:h_times_factor}
\end{align} 
    
The factor multiplying $h_{j, \sigma}(x)$ above is non-negative. Therefore, the matrix in expression~\ref{eq:h_times_factor} is negative semidefinite if and only if $h_{j, \sigma}(x)$ is negative semidefinite. 

\textbf{Step 3:} \textit{Determine a necessary and sufficient condition for all eigenvalues of this matrix to be non-positive (i.e., for $h_{j,\sigma}$ to be negative semidefinite).}

We can observe the spectrum of the matrix from expression~\ref{eq:h_times_factor}. Since $xx^T$ has the spectrum $\|x\|^2$ once and $0$ $j-1$ times, the spectrum of this matrix consists of $\tilde A_{j,\sigma}(x)$ once and $-\tilde B_{j,\sigma}(x)$ $j-1$ times, where 
\begin{align*} 
\tilde{A}_{j, \sigma}(x) &= \|x\|^2 - \tilde B_{j,\sigma}(x) \\ 
\tilde B_{j, \sigma}(x) &=(1-\sigma^{-2})^{-2}\frac{(\gamma_{j,1}(x) + \gamma_{j,\sigma}(x)) ( \gamma_{j,1}(x) + \sigma^{-2}\gamma_{j,\sigma}(x) )}{\gamma_{j,1}(x) \gamma_{j,\sigma}(x)}. 
\end{align*}
    
The matrix in expression~\ref{eq:h_times_factor} is negative semi-definite if and only if all of its eigenvalues are non-positive. Since $-\tilde B_{j, \sigma}(x)$ is evidently negative, we conclude that $f_{j,\sigma}$ is log-concave if and only if $\tilde{A}_{j,\sigma}(x) \le 0$ for every $x$. The subsequent argument develops the range of $\sigma$ in which this holds. To this end, first observe that both $\tilde{A}_{j,\sigma}$ and $\tilde{B}_{j,\sigma}$ depend on $x$ only through $\|x\|^2$. Denoting $\rho = \|x\|^2$ and expanding out the ratio of Gaussian densities above, we may define 
\begin{align*}
    A_{j,\sigma}(\rho) &= \rho - B_{j,\sigma}(\rho) \\
    B_{j,\sigma}(\rho) &= (1-\sigma^{-2})^{-2} (e^{\rho/2} + \sigma^j e^{\rho/2\sigma^2})(e^{-\rho/2} + \sigma^{-2-j} e^{-\rho/2\sigma^2}).
\end{align*}
For $\rho = \|x\|^2$, $B_{j,\sigma}(\rho) = \tilde B_{j, \sigma}(x)$ because 
\begin{align*}
    \tilde B_{j, \sigma}(x) &= (1-\sigma^{-2})^{-2} \: \frac{(\gamma_{j,1}(x) + \gamma_{j,\sigma}(x)) ( \gamma_{j,1}(x) + \sigma^{-2}\gamma_{j,\sigma}(x) )}{\gamma_{j,1}(x) \gamma_{j,\sigma}(x)} \\
    &= (1-\sigma^{-2})^{-2} \: \frac{\exp(-\|x\|^2/2) + \sigma^{-j}\exp(-\|x\|^2/2\sigma^2)  }{\sigma^{-j} \exp(-\|x\|^2/2) \exp(-\|x\|^2/2\sigma^2)} \times \\
    &\qquad \left\{\exp(-\|x\|^2/2) + \sigma^{-2-j}\exp(-\|x\|^2/2\sigma^2)\right\} \\
    &= (1-\sigma^{-2})^{-2} \: \frac{\left\{\exp(-\rho/2) + \sigma^{-j}\exp(-\rho/2\sigma^2)\right\} \left\{\exp(-\rho/2) + \sigma^{-2-j}\exp(-\rho/2\sigma^2)\right\}}{\sigma^{-j} \exp(-\rho/2) \exp(-\rho/2\sigma^2)} \\
    &= (1-\sigma^{-2})^{-2} \left\{\sigma^j e^{\rho/2\sigma^2} + e^{\rho/2} \right\} \left\{e^{-\rho/2} + \sigma^{-2-j}e^{-\rho/2\sigma^2}\right\} \\
    &= B_{j,\sigma}(\rho).
\end{align*}
Then $A_{j,\sigma}(\rho)$ is clearly equal to $\tilde{A}_{j, \sigma}(x)$ for $\rho = \|x\|^2$.

Again, $f_{j,\sigma}$ is log concave if and only if for all $\rho \ge 0,$ $A_{j,\sigma}(\rho) \le 0,$ or equivalently if and only if $M(j,\sigma) \le 0,$ where \[ M(j,\sigma) := \sup_{\rho \ge 0} A_{j,\sigma}(\rho). \]

\textbf{Step 4:} \textit{Derive an explicit expression for $M(j, \sigma)$. The density $f_{j,\sigma}$ is log-concave if and only if $M(j, \sigma) \leq 0$.}

Let us first simplify $B_{j,\sigma}(\rho)$ as 
\begin{align*} 
B_{j,\sigma}(\rho) &= (1-\sigma^{-2})^{-2} \left( 1 + \sigma^{-2} + \sigma^{-1} \left[\sigma^{j+1} e^{-\rho(1-\sigma^{-2})/2} + \sigma^{(-j-1)} e^{\rho(1-\sigma^{-2})/2} \right] \right)\\ 
&= (1-\sigma^{-2})^{-2} \left( 1 + \sigma^{-2} + \sigma^{-1} \left[e^{-\rho (1-\sigma^{-2})/2 + (j+1)\log\sigma} + e^{\rho (1-\sigma^{-2})/2 - (j+1)\log\sigma} \right]  \right)\\ 
&= (1-\sigma^{-2})^{-2} \left( 1 + \sigma^{-2} + 2\sigma^{-1} \text{cosh}\left( \rho(1-\sigma^{-2})/2 - (j+1) \log \sigma \right) \right). 
\end{align*}
    
Notice that $B_{j,\sigma}$ is a convex function of $\rho,$ and thus $A_{j,\sigma}$ is a concave function of $\rho$. Therefore, $A_{j,\sigma}$ admits a unique maximum $\rho^*(j,\sigma),$ which satisfies the equation 
$$\frac{\partial}{\partial \rho} A_{j,\sigma}(\rho) \equiv 1 - (1-\sigma^{-2})^{-2}  \cdot 2 \sigma^{-1} \cdot \frac{1-\sigma^{-2}}{2} \cdot \text{sinh}(\rho^*(j,\sigma) (1-\sigma^{-2})/2 - (j+1) \log \sigma) = 0.$$
We can see that 
$$\rho^*(j,\sigma) =  \frac{2(j+1) \log \sigma}{1-\sigma^{-2}} +  \frac{2}{1-\sigma^{-2}} \text{arcsinh}( \sigma - \sigma^{-1})$$ 
solves this expression (where $\text{arcsinh}$ is the inverse of the $\text{sinh}$ function) because 
\begin{align*}
    1 &- (1-\sigma^{-2})^{-2}  \cdot 2 \sigma^{-1} \cdot \frac{1-\sigma^{-2}}{2} \cdot \text{sinh}(\rho^*(j,\sigma) (1-\sigma^{-2})/2 - (j+1) \log \sigma) \\
    &= 1 - (1-\sigma^{-2})^{-1}  \cdot \sigma^{-1} \cdot \text{sinh}((j + 1) \log\sigma + \text{arcsinh}(\sigma - \sigma^{-1}) - (j+1) \log \sigma) \\
    &= 1 - (1-\sigma^{-2})^{-1}  \cdot \sigma^{-1} \cdot (\sigma - \sigma^{-1}) \\
    &= 1 - (1-\sigma^{-2})^{-1} \cdot (1 - \sigma^{-2}) \\
    &= 0.
\end{align*}

Observe further that since $\text{cosh}^2(x) = 1 +\text{sinh}^2(x),$ we can also infer that  
\begin{align*}
   \text{cosh}( \rho^*(j,\sigma) (1-\sigma^{-2})/2 -(j+1)\log \sigma) &=  \sqrt{ 1 + \text{sinh}^2(\rho^*(j,\sigma) (1-\sigma^{-2})/2 -(j+1)\log \sigma)} \\
   &= \sqrt{ 1 + (\sigma - \sigma^{-1})^2}.  
\end{align*}
Plugging $\rho^*(j,\sigma)$ into $A_{j,\sigma},$ we conclude that \[ M(j,\sigma) = \frac{2(j+1) \log \sigma}{1-\sigma^{-2}} +  \frac{2\text{arcsinh}( \sigma - \sigma^{-1})}{1-\sigma^{-2}}  - \frac{ \left( 1 + \sigma^{-2} + 2\sigma^{-1} \sqrt{1 + (\sigma - \sigma^{-1})^2}\right)}{(1-\sigma^{-2})^{2}}.
    \]

\textbf{Step 5:} \textit{Show that $M(j, \sigma)$ strictly increases with $\sigma$ for $\sigma > 1$ by determining that the derivatives of each term with respective to $\sigma$ are positive.}

To show that $M(j, \sigma)$ strictly increases with $\sigma$ for $\sigma > 1$, we consider the derivatives of the three terms of $M(j,\sigma)$ with respect to $\sigma$. The derivative of the first term is 
\begin{align*}
    \frac{\partial}{\partial \sigma} \frac{2(j+1) \log \sigma}{1-\sigma^{-2}} &= 2(j+1) \frac{(1-\sigma^{-2})(1/\sigma) - (\log \sigma) (2\sigma^{-3})}{(1 - \sigma^{-2})^2} \\
    &= 2(j+1) \frac{\sigma^4 (1-\sigma^{-2})(1/\sigma) - \sigma^4 (\log \sigma) (2\sigma^{-3})}{(\sigma^2 - 1)^2} \\
    &= 2(j+1) \frac{\sigma^3 - \sigma - 2\sigma \log \sigma }{(\sigma^2 - 1)^2} \\
    &= 2(j+1) \frac{\sigma (\sigma^2 - 2 \log \sigma - 1)}{(\sigma^2 - 1)^2},
\end{align*}
which is positive for $\sigma > 1$. The derivative of the second term is
\begin{align}
   \frac{\partial}{\partial \sigma}  &\frac{2\text{arcsinh}( \sigma - \sigma^{-1})}{1-\sigma^{-2}} \notag \\
   &= 2 \frac{(1-\sigma^{-2})((\sigma - \sigma^{-1})^2 + 1)^{-1/2} (1 + \sigma^{-2}) - 2\sigma^{-3} \text{arcsinh}(\sigma - \sigma^{-1}) }{(1-\sigma^{-2})^2} \notag \\
   &= 2 \frac{(\sigma^{4}-1)(\sigma^2 - 2 + \sigma^{-2} + 1)^{-1/2} - 2\sigma \text{arcsinh}(\sigma - \sigma^{-1}) }{(\sigma^2 - 1)^2} \notag \\
   &= 2 \frac{(\sigma^{4}-1)(\sigma^2 + \sigma^{-2} - 1)^{-1/2}  (\sigma^2 + \sigma^{-2} - 1)^{1/2} - 2\sigma (\sigma^2 + \sigma^{-2} - 1)^{1/2} \text{arcsinh}(\sigma - \sigma^{-1}) }{(\sigma^2 - 1)^2 (\sigma^2 + \sigma^{-2} - 1)^{1/2}} \notag \\
   &= 2 \frac{\sigma^{4} - 1 - 2 (\sigma^4 - \sigma^2 + 1)^{1/2} \text{arcsinh}(\sigma - \sigma^{-1}) }{(\sigma^2 - 1)^2 (\sigma^2 + \sigma^{-2} - 1)^{1/2}}. \label{eq:M_term2}
\end{align}
The denominator of this expression is clearly positive. We can also show that the numerator is positive. For the derivations that follow, we will always assume that $\sigma > 1$. First, we know that $(\sigma^4 - \sigma^2 + 1)^{1/2} < \sigma^2$. Second, note that
$$\sigma^4 - 2\sigma^3 + 2\sigma - 1 = (\sigma-1)^3 (\sigma + 1) > 0.$$ 
Dividing both sides by $2\sigma^2$ and rearranging, we see that
$$\text{sinh}(2\log\sigma) = \sigma^2/2 - \sigma^{-2}/2 > \sigma - \sigma^{-1}.$$
Since arcsinh is monotone increasing, this implies that $2\log\sigma > \text{arcsinh}(\sigma - \sigma^{-1})$. In addition, $\log \sigma < \sigma - 1$. Returning to the numerator of expression~\ref{eq:M_term2}, we see that
\begin{align*}
\sigma^{4} &- 1 - 2 (\sigma^4 - \sigma^2 + 1)^{1/2} \text{arcsinh}(\sigma - \sigma^{-1}) \\
&> \sigma^{4} - 1 - 2\sigma^2 \text{arcsinh}(\sigma - \sigma^{-1}) \\
&> \sigma^4 - 1 - 4\sigma^2 \log\sigma \\
&> 0.
\end{align*}
It is possible to show the final inequality by verifying that the derivative of $\sigma^4 - 1 - 4\sigma^2 \log\sigma$ is positive and that this expression equals 0 at $\sigma = 1$. Hence, the derivative of the second term in $M(j,\sigma)$ is also positive for $\sigma > 1$.

Now we consider the third term of $M(j,\sigma)$. Where $z$ represents $\sigma^{-2}$, this term is equal to the function \[ \zeta(z) := -\frac{1 + z + 2 \sqrt{z + (1-z)^2}}{(1-z)^2}.\] Directly computing the derivative yields
\begin{align*}
&\frac{\partial}{\partial(z)} \zeta(z) \\
&= -\frac{(1-z)^2 \left[1 + 2(1/2)(z + (1-z)^2)^{-1/2}(1-2(1-z))\right] - \left[1 + z + 2 \sqrt{z + (1-z)^2}\right](-2(1-z))}{(1-z)^4} \\
&= - \frac{(1-z)\left[1 + (z^2 - z + 1)^{-1/2} (2z - 1) \right] + \left[2 + 2z + 4(z^2 - z + 1)^{1/2} \right]}{(1-z)^3} \\
&= -\frac{3 + z + (z^2 - z + 1)^{-1/2} (2z - 1 - 2z^2 + z) + 4(z^2 - z + 1)^{1/2}}{(1-z)^3} \\
&= -\frac{3 + z + (z^2 - z + 1)^{-1/2} (- 2z^2 + 3z - 1) + 4(z^2 - z + 1)^{1/2}}{(1-z)^3} \\
&= -\frac{(3 + z)(z^2 - z + 1)^{1/2} + (- 2z^2 + 3z - 1) + 4(1 - z + z^2)}{(1-z)^3 (1-z+z^2)^{1/2}} \\
&= -\frac{(3 + z)(z^2 - z + 1)^{1/2} + 2z^2 - z + 3}{(1-z)^3 (1-z+z^2)^{1/2}}.
\end{align*}
The denominator terms are both positive for $z\in [0,1)$. In addition, the numerator terms $(3 + z)(z^2 - z + 1)^{1/2}$, $2z^2$, and $3-z$ are also positive for $z\in [0, 1)$. Since there is a coefficient of $-1$ in front of the fraction, we conclude that $\zeta(z)$ strictly decreases as $z$ increases, for $z\in [0,1)$. Equivalently, since $z$ represents $\sigma^{-2}$, the third term of $M(j, \sigma)$ strictly increases as $\sigma$ increases for $\sigma > 1$, and the derivative of the third term is positive.

\textbf{Step 6:} \textit{Determine that $M(j, \sigma)$ must have a unique root in $\sigma$ for each $j$ and that this sequence of roots is strictly decreasing in $j$. A  $j$-dimensional density in the proposed mixture is log-concave if and only if its $\sigma$ value is less than or equal to the root in $j$ dimensions.}

Observe that $M(j, \sigma)$ is continuous, and $M(j,\sigma) \to -\infty$ as $\sigma \to 1$ from above (third term dominating). In addition, $M(j, \sigma) \to \infty$ as $\sigma \to \infty$, since the first two terms tend to $\infty$ and the third term converges to $-3$. Thus, for every $j,$ $M(j,\sigma)$ is a strictly increasing function with a unique root. Denote this root as $\tau_j$. We can conclude that 
\begin{align*}
    M(j, \sigma) \le  0 &\iff \sigma \le \tau_j,
\end{align*} 
and in particular $f_{j,\sigma}$ is log-concave if and only if $\sigma \le \tau_j$. Finally, we note that this $\tau_j$ must strictly decrease with $j$, since $M(j, \sigma)$ is strictly increasing in $j$ for fixed $\sigma$. \qedhere    
    
The $\tau_j$ are quite amenable to computation since they are the roots of a strictly increasing function $M(j, \sigma)$. Concretely, we can approximate \[ \tau_1 \approx 1.80302, \tau_2 \approx 1.66083, \text{ and } \tau_3 \approx 1.57175. \] Observe in particular that $\tau_2 < \sqrt{3} < \tau_1$. Hence, the two-dimensional Gaussian mixture given by $(\gamma_{2,1}(x) + \gamma_{2,\sqrt{3}}(x))/2$ is not log-concave, but its one-dimensional projections are log-concave.
\end{proof}

\subsection{Consistency of  Universal Inference for Log-Concavity} \label{app:consistency}

We begin by commenting further on the assumptions from Section~\ref{sec:theoretical}. Then we prove Theorem~\ref{thm:logconcpower} along the lines sketched in the main text.

\subsubsection{Further Discussion of Assumptions}

\noindent \textbf{On Assumption~\ref{assumption:estimability} and rates.} The constant $1/200$ in the assumption is largely a matter of convenience. It arises by making a choice of the constellation of constants that appear in the results of \citet{wong1995probability}. While we have not attempted to optimize the same, it is plausible (by checking the limiting behavior of the results of \citet{wong1995probability}) that this can be improved to at least $1/12$. 

Additionally we observe that using similar methods, a bolstering of Assumption~\ref{assumption:estimability} in terms of rates of decay of the probabilities in question, combined with known control on the metric entropy of log-concave distributions, should yield a rate statement of the form ``If $\flc \in \mathcal{F} \subset \mathcal{F}_d$ and  $h^2(f^*, \flc) \gtrsim \varepsilon_n,$ then the power of the test is $1 - o(1)$.'' The $o(1)$ term would depend on the strength of this assumption (and otherwise be exponentially small), while $\varepsilon_n$ would depend on the complexity of the class $\mathcal{F}.$ For $\mathcal{F}$ consisting of log-concave laws with near-identity covariance, such control is available \citep[e.g.,][]{kur2019optimality}, and the convex-ordering of log-concave projections (e.g., Corollary 5.3 of the survey by \cite{samworth2018recent}) should allow such claims for $f^*$ with near-identity covariance.

\subsubsection{Proof of Consistency}


\thmLogconcpower*

\begin{proof}

We first recall the approach from the proof sketch in the main text. We assume throughout that $H_1$ is true, meaning that $f^*$ is not log-concave. For brevity, we use $\PP$ to denote $\PP_{H_1}$. We begin by decomposing $T_n$ into \[ T_n = \underbrace{\prod_{i \in \mathcal{D}_{0,n}} \frac{\hat{f}_{1,n}(Y_i)}{f^*(Y_i)}}_{=: 1/R_n} \cdot \underbrace{\prod_{i \in \mathcal{D}_{0,n}} \frac{f^*(Y_i)}{\hat{f}_{0,n}(Y_i)}}_{=: S_n} = S_n/R_n.\]

Let $\varepsilon := h(f^*, \flc).$ Further, suppose $n \geq 100 \log(1/\alpha)/\varepsilon^2.$ Now observe that \begin{align*}
    \{ T_n < 1/\alpha\} &\subset \{R_n > \exp(n\varepsilon^2/100) \} \cup \{S_n < \exp(n\varepsilon^2/50) \},
\end{align*} since outside this union, $S_n/R_n \geq \exp(n\varepsilon^2/100) \geq 1/\alpha$. Thus, it suffices to argue that \begin{equation} \label{eqn:target}
    \PPP{R_n > \exp(n \varepsilon^2/100}) + \PPP{S_n < \exp(n \varepsilon^2/50}) \to 0.
\end{equation} 

The two assumptions contribute to bounding these terms. In particular, Assumption~\ref{assumption:regularity} implies that $S_n$ must be big, while Assumption~\ref{assumption:estimability} implies that $R_n$ is small.

\textbf{$S_n$ is big}. Observe that since $h(f^*, \flc) = \varepsilon > 0$ and $h(f^*, \flc) \leq 1$, the likelihood ratio $\prod_{i \in \dee_{0,n}} f^*(Y_i) / \flc(Y_i)$ tends to be exponentially large with high probability.   We can use Markov's inequality and the properties of Hellinger distance to show that for any $\xi > 0$ (and particularly for large $\xi$), 
\begin{align*}
&\PPP{\prod_{i \in \dee_{0,n}} f^*(Y_i)/\flc(Y_i) < \xi} = \PPP{\prod_{i \in \dee_{0,n}} \sqrt{\flc(Y_i) / f^*(Y_i)} > 1/\sqrt{\xi}} \\
&\le \sqrt{\xi} \mathbb{E}_{Y \sim f^*}\left[ \sqrt{\flc(Y) / f^*(Y)}\right]^{n/2} = \sqrt{\xi} (1-h^2(f^*, \flc))^{n/2} \to 0.
\end{align*}
It may thus be expected that the same holds true for the ratio of interest $\prod_{i \in \dee_{0,n}} f^*(Y_i)/\zerohatf(Y_i)$. The regularity conditions from Assumption~\ref{assumption:regularity} enable precisely this, by establishing that $\zerohatf \to \flc$ in the strong sense that for large $n$, $\zerohatf$ lies in a small bracket containing $\flc$. This pointwise control then enables the use of classical results from empirical process theory to show that for large enough $n$, $\prod_{i \in \dee_{0,n}} f^*(Y_i)/\zerohatf(Y_i)$ grows at an exponential rate similar to $\prod_{i \in \dee_{0,n}} f^*(Y_i)/\flc(Y_i)$.

Concretely, recall that for a pair of functions $u \le v$, the bracket $[u,v]$ is defined as the set of all functions $f$ such that $u \le f \le v$ everywhere (denoted $f\in[u,v]$). The next lemma uses the characterizations of convergence of $\zerohatf$ due to \citet{cule2010theoretical}. 
\begin{restatable}[]{lem}{lemMleCloseToMProjection}\label{lemma:mle_close_to_M_projection}
    Under the conditions of Assumption~\ref{assumption:regularity}, for any $\eta > 0,$ there exist nonnegative functions $u_\eta \le v_\eta$ such that $\flc \in [u_\eta, v_\eta],$  \[ \int_{\mathbb{R}^d} \left(v_\eta(x) - u_\eta(x)\right) \,\mathrm{d} x \le \eta, \] and \[ \PPP{ \exists n_0 : \forall n \ge n_0, \zerohatf \in [u_\eta, v_\eta] } = 1.\]
\end{restatable}

This lemma offers strong pointwise control on the values that $\zerohatf$ can possibly take. This immediately allows us to use the following result, which is a simplification of Theorem~1 from \citet{wong1995probability}.
\begin{restatable}[]{lem}{lemPointwiseControlImpliesBlowup}\label{lemma:pointwise_control_implies_blowup}
    There exists $\eta_0 > 0,$ depending on $\varepsilon,$ such that if $[u_{\eta_0},v_{\eta_0}]$ is a bracket constructed to satisfy Lemma~\ref{lemma:mle_close_to_M_projection}, then \[ \PPP{ \inf_{g \in [u_{\eta_0}, v_{\eta_0}]} \prod_{i \in \dee_{0,n}} \frac{f^*(Y_i)}{g(Y_i)} \le \exp(n \varepsilon^2/50) } \le 4\exp( - C n \varepsilon^2),  \] where $C > 2^{-14}$ is a universal constant. 
\end{restatable}

We prove both Lemma~\ref{lemma:mle_close_to_M_projection} and Lemma~\ref{lemma:pointwise_control_implies_blowup} later in this section. The claim that $\PPP{S_n < \exp(n \varepsilon^2/50) }$ converges to 0 follows from a combination of the above statements. Indeed, choose an appropriate $\eta_0,$ and define the event $\mathscr{E}_n := \{ \zerohatf \in [u_{\eta_0}, v_{\eta_0}]\}.$ Then 
\begin{align*}
\PPP{S_n < \exp(n \varepsilon^2/50) } &\le \PPP{S_n \le \exp(n\varepsilon^2/50) , \mathscr{E}_n} + (1-\PP(\mathscr{E}_n)) \\
&= \PPP{\prod_{i \in \dee_{0,n}} \frac{f^*(Y_i)}{\zerohatf(Y_i)} \leq \exp(n \varepsilon^2/50), \: \zerohatf \in [u_{\eta_0}, v_{\eta_0}] } + (1-\PP(\mathscr{E}_n)) \\
&\le  \PPP{ \inf_{g \in [u_{\eta_0}, v_{\eta_0}]} \prod_{i \in \dee_{0,n}} \frac{f^*(Y_i)}{g(Y_i)} \le \exp(n \varepsilon^2/50) } + (1 - \PP(\mathscr{E}_n)) \\
&= o(1) + 1 - \PP(\mathscr{E}_n). 
\end{align*} 

Under Assumption~\ref{assumption:regularity}, we can apply Lemma~\ref{lemma:mle_close_to_M_projection} to show that $\lim_{n\to\infty} \PP(\mathscr{E}_n) = 1$. Then we conclude that $S_n$ is asymptotically large enough to enable control via (\ref{eqn:target}).

\textbf{$R_n$ is small}. The smallness of $R_n$ relies on the fact that $\onehatf$ approximates $f^*$ well in a strong sense. To this end, observe that the function $\onehatf$ is purely determined by the data in the split $\dee_{1,n}$. Let us abbreviate $\mathcal{Y}_{1,n} = \{Y_i: i \in \dee_{1,n}\}$. We may write \[ \PPP{R_n > \exp(n \varepsilon^2/100)} = \mathbb{E}_{\mathcal{Y}_{1,n}}\left[ \PPP{R_n > \exp(n \varepsilon^2/100) \mid \mathcal{Y}_{1,n}} \right].\]

With this in hand, we observe that conditional on $\mathcal{Y}_{1,n},$ both $f^*$ and $\onehatf$ are fixed functions. Further, due to the independence of $\mathcal{Y}_{1,n}$ and $\{Y_i : i \in \dee_{0,n}\}$, an application of Markov's inequality yields the following statement, where the upper bound is in terms of the $\sigma(\mathcal{Y}_{1,n})$-measurable\footnote{Strictly speaking, it is possible that the estimation of $\onehatf$ is a randomized procedure. As long as the extraneous randomness is independent of $\{Y_i : i \in \dee_{0,n}\},$ this does not affect the details of this argument, beyond the fact that one would also need to condition on the randomness of this algorithm.} $\onehatf$.


\begin{restatable}[]{lem}{lemRnSmallMarkov}\label{lem:R_n_small_markov} Let $(x)_+ := \max(x,0),$ and interpreting $x/0$ as $\infty,$ 
    \[ \PPP{R_n > \exp(n \varepsilon^2/100 ) \mid \mathcal{Y}_{1,n}} \le \frac{\int f^*(x) \log^2(f^*(x)/\onehatf(x)) \,\mathrm{d}x }{n (( \varepsilon^2/50 - \dkl(f^*\|\onehatf))_+)^2}.   \] 
\end{restatable}

We prove Lemma~\ref{lem:R_n_small_markov} later in this section. For any $\theta > 0$, define the event \begin{align*}
    \mathscr{F}_{0, \theta} &= \{\dkl(f^*\|\onehatf) \le \varepsilon^2/100\} \cap \left\{ \int f^*(x) \log^2( f^*(x)/\onehatf(x)) \mathrm{d}x \le \theta n\right\}.
\end{align*} 
We find by the law of total probability and Lemma~\ref{lem:R_n_small_markov} that 
\begin{align}
    \PPP{R_n > \exp(n \varepsilon^2/100)} &\le (1 - \PPP{\mathscr{F}_{0, \theta}} ) + \PPP{R_n \ge \exp(n  \varepsilon^2/100) , \mathscr{F}_{0, \theta}} \notag \\
    &= (1 - \PPP{\mathscr{F}_{0, \theta}} ) + \PPP{R_n \ge \exp(n  \varepsilon^2/100) \mid \mathscr{F}_{0, \theta}} \PPP{\mathscr{F}_{0, \theta}} \notag \\
    &\le (1 - \PPP{\mathscr{F}_{0, \theta}} ) + \mathbb{E}\left[I(R_n \ge \exp(n  \varepsilon^2/100)) \mid \mathscr{F}_{0, \theta} \right] \notag \\
    &= (1 - \PPP{\mathscr{F}_{0, \theta}} ) + \mathbb{E}\left[ \mathbb{E}\left[ I(R_n \ge \exp(n  \varepsilon^2/100)) \mid \mathcal{Y}_{1,n} \right] \mid \mathscr{F}_{0, \theta} \right] \notag \\
    &\leq (1 - \PPP{\mathscr{F}_{0, \theta}} ) + \mathbb{E}\left[ \frac{\int f^*(x) \log^2(f^*(x)/\onehatf(x)) \,\mathrm{d}x }{n (( \varepsilon^2/50 - \dkl(f^*\|\onehatf))_+)^2} \: \Big| \: \mathscr{F}_{0, \theta}\right] \notag \\
    &\le (1 - \PPP{\mathscr{F}_{0, \theta}} ) + \frac{\theta}{(\varepsilon^2/100)^2} , \label{eqn:r_n_is_small_ell_0_inter} 
\end{align}
where $(\ref{eqn:r_n_is_small_ell_0_inter})$ holds by the definition of $\mathscr{F}_{0, \theta}.$ Since the left hand side is independent of $\theta,$ 
\begin{align*}
    \lim_{n\to\infty} \PPP{R_n > \exp(n \varepsilon^2/100)} &= \inf_{\theta > 0} \lim_{n\to\infty} \PPP{R_n > \exp(n \varepsilon^2/100)} \\ 
    &\leq \inf_{\theta > 0} \lim_{n\to\infty} \left[ (1 - \PPP{\mathscr{F}_{0, \theta}} ) + \frac{\theta}{(\varepsilon^2/100)^2}\right] \\
    &= \inf_{\theta > 0} \lim_{n\to\infty} (1 - \PPP{\mathscr{F}_{0, \theta}} ) + \inf_{\theta > 0} \frac{\theta}{(\varepsilon^2/100)^2} \\
    &= \inf_{\theta > 0} \lim_{n\to\infty} (1 - \PPP{\mathscr{F}_{0, \theta}} ). 
\end{align*}
Under Assumption~\ref{assumption:estimability}, $\lim_{n\to\infty}  \PP(\mathscr{F}_{0,\theta}) = 1$ for all $\theta > 0$. Thus we conclude that $$\lim_{n\to\infty} \PPP{R_n > \exp(n \varepsilon^2/100)} = 0,$$ and $R_n$ is asymptotically small enough to enable power control via (\ref{eqn:target}).
\end{proof}

It remains for us to prove the lemmata invoked in the above argument, which we now proceed to do. The three statements concern qualitatively distinct aspects of the argument. Lemma~\ref{lemma:mle_close_to_M_projection} is a structural result about the log-concave MLE, Lemma~\ref{lemma:pointwise_control_implies_blowup} is more generic and concerns the behavior of likelihood ratios in classes of bounded complexity, while Lemma~\ref{lem:R_n_small_markov} is an application of Markov's inequality that exploits closeness in KL-divergence. 

\paragraph{Pointwise Convergence of Log-Concave MLEs.} In this section, we restate and prove Lemma~\ref{lemma:mle_close_to_M_projection}. The argument here is essentially a slight refinement of the convergence analysis of \citet{cule2010theoretical}. 

\lemMleCloseToMProjection*

\begin{proof}[Proof of Lemma~\ref{lemma:mle_close_to_M_projection}]

We shall heavily exploit the results of \citet{cule2010theoretical}. Let $U := \mathrm{int}(\mathrm{supp}(\flc)),$ and let $\mathscr{F}_a$ be the event that $\int_{\mathbb{R}^d} e^{a\|x\|}|\zerohatf(x) - \flc(x)| dx \to 0.$ We know that there exists  $a > 0$ such that $\PPP{\mathscr{F}_a} = 1$ \citep[][Thm. 4]{cule2010theoretical}. 

    

    Now, as argued in the proof of Theorem 4 of Cule and Samworth (page 264, paragraph starting with ``we claim that''), there exists some $a_1 > 0$ and $b_1 \in \mathbb{R}$ such that for every $x\in\mathbb{R}^d$, 
    \begin{equation}
     \sup_{n} \zerohatf(x) \le e^{-a_1\|x\| + b_1}. \label{eq:fhat_0n_UB}
    \end{equation}
    Furthermore, as stated in Theorem 4 of Cule and Samworth, since $\flc$ is log-concave, there exists some $a_2 > 0$ and $b_2 \in \mathbb{R}$ such that for all $x\in\mathbb{R}^d$, $\flc(x) \le e^{-a_2\|x\| + b_2}$. Let $a_0 := \min(a_1,a_2)$ and $b_0 := \max(b_1, b_2),$ so that $\max( \flc(x), \sup_n \zerohatf(x)) \le e^{-a_0\|x\| + b_0}$ for every $x$. Let $R_\eta$ be such that \[ \int_{\|x\| > R_\eta} e^{-a_0\|x\| + b_0} \mathrm{d}x \le \eta/3.\]
    
    Further, given $\mathscr{F}_a,$ we note that on any compact subset of $U$, $\zerohatf \to \flc$ uniformly, as argued by \citet{cule2010theoretical} using Theorem 10.8 of \citet{rockafellar1997convex}.\footnote{While \citet{cule2010theoretical} explicitly argue this only for balls contained within $U$, this in fact follows for any compact subset. The gist of the argument is as follows. The results of \citet{cule2010theoretical} imply that $\log \zerohatf$ converges to $\log \flc$ on all positive Lebesgue measure sets (and thus on a dense subset of $U$). Since $\log \flc$ and $\log \zerohatf$ are concave functions, Theorem 10.8 of \citet{rockafellar1997convex} implies that convergence is uniform on any compact subset of $U.$ Finally, since both $\log\zerohatf$ and $\log \flc$ are uniformly bounded (by $b_0$ in our notation), this uniform convergence extends to $\zerohatf$ and $\flc$ due to the uniform continuity of the exponential function on sets of the form $[-\infty, C]$ for $C < \infty$.} 
    
    Let $V:= U \cap \{x \in \mathbb{R}^d: \|x\| \le R_\eta\}$. Since $V$ has finite Lebesgue measure, by a standard consequence of regularity of the Lebesgue measure, it contains a compact set $W$ such that $\mathrm{Leb}_d(V \setminus W) \le e^{-b_0}\eta/3$. Fix such a $W$. 
    Finally, note that since $W$ is compact, $\zerohatf \to \flc$ uniformly on $W$, and thus there exists some $n_0$ such that under $\mathscr{F}_a$, it follows that if $n > n_0$, then 
    \begin{equation}
       \max_{x \in W} |\zerohatf(x) - \flc(x)| \le \frac{\eta}{6\mathrm{Leb}_d(\{x \in \mathbb{R}^d : \|x\| \le R_\eta\})}. \label{eq:W_condition}
    \end{equation}
    
    We are now in a position to construct the functions $u_\eta$ and $v_\eta$. Indeed, observe that in each of the sets constructed above, we have explicit bounds available on $\zerohatf$. In particular, we let \begin{align*}
        u_\eta(x) &:= \begin{cases} \flc(x) - \frac{\eta}{6\mathrm{Leb}_d(\{x \in \mathbb{R}^d : \|x\| \le R_\eta\})} & x \in W \\ 0 & x \not\in W\end{cases},\\
        v_\eta(x) &:= \begin{cases} \flc(x) + \frac{\eta}{6\mathrm{Leb}_d(\{x \in \mathbb{R}^d : \|x\| \le R_\eta\})} & x \in W \\ e^{b_0} & x \in V \setminus W \\ e^{-a_0\|x\| + b_0} & x \in U \cap \{\|x\| > R_\eta\} \\ 0 & x \in U^c\end{cases}.
    \end{align*}
    
    We first note that $v_\eta \ge \flc \ge u_\eta$ everywhere, and \begin{align*}
        &\int_{\mathbb{R}^d} \left(v_\eta(x) - u_\eta(x) \right) \,\mathrm{d}x \\
        &= \left( \int_W + \int_{V \setminus W} + \int_{U \cap \{\|x\| > R_\eta\}} \right) (v_\eta(x) - u_\eta(x)) dx \\
        &\le \frac{\eta}{3} \frac{\mathrm{Leb}_d(W)}{\mathrm{Leb}_d(\{x \in \mathbb{R}^d : \|x\| \le R_\eta\})} + e^{b_0} \mathrm{Leb}_d(V \setminus W) + \int_{\{x \in\mathbb{R}^d : \|x\| > R_\eta\}} e^{-a_0\|x\| + b_0} dx \\
        &\le \eta/3 + \eta/3 + \eta/3  = \eta.
    \end{align*} 
    
    For our choice of $n_0$, under the event $\mathscr{F}_a$, (\ref{eq:fhat_0n_UB}) and (\ref{eq:W_condition}) guarantee that $u_\eta(x) \le \zerohatf(x) \le v_\eta(x)$ for all $n \geq n_0$. Since $\PPP{\mathscr{F}_a} = 1,$ the conclusion follows.    
\end{proof}

\paragraph{Pointwise convergence of $\zerohatf$ implies asymptotic blowup.} In this section, we restate and prove Lemma~\ref{lemma:pointwise_control_implies_blowup}. We use the results of \citet{wong1995probability}, which are stated in terms of the Hellinger bracketing entropy. Since we are only interested in consistency, the hypotheses underlying these can be satisfied by taking a small enough 
$\eta_0 > 0$ and analyzing the behavior over the corresponding bracket $[u_{\eta_0}, v_{\eta_0}]$, along with exploiting the fact that if $h(\flc, f^*)$ is large, then so is $h(g, f^*)$ for any $g \in [u_{\eta_0}, v_{\eta_0}]$. 

For completeness, we briefly describe Hellinger bracketing entropy. For a pair of functions $v \le u$, the bracket $[v,u]$ is the set of functions that is sandwiched between $v$ and $u$, i.e., $\{f : v(x) \le f(x) \le u(x) \text{ for all } x\}$. The Hellinger size of such a bracket is defined as the Hellinger distance $h(v,u)$, where we are assuming that $v$ and $u$ are nonnegative. The Hellinger bracketing entropy of a class of densities $\mathcal{F}$ at a scale $\zeta$ is denoted $\mathcal{H}_{[]}(\mathcal{F}, h, \zeta)$, and it equals the logarithm of the smallest number of brackets of size at most $\zeta$ that cover the class $\mathcal{F}$. Importantly, the boundary functions $u,v$ need not belong to $\mathcal{F}$ itself (and so need not themselves be densities). We also make the observation that $\mathcal{H}_{[]}(\cdot, \cdot, \zeta)$ is nonincreasing in $\zeta$.

We shall also use the following property of Hellinger divergence: For any pair of non-negative functions $u,v$, \begin{equation}\label{eqn:dH_squared_is_smaller_than_L1}
     h^2(u,v) \le \frac{1}{2} \|u - v\|_1,
\end{equation}
where the $1$-norm is in the $L_1(\mathrm{Leb}_d)$ sense, and $h^2(u,v) = (h(u,v))^2$. To see this, observe that \begin{align*} h^2(u,v) &= \frac{1}{2} \int (\sqrt{u(x)} - \sqrt{v(x)})^2 \mathrm{d}x \\ &\le \frac{1}{2} \int |\sqrt{u(x)} - \sqrt{v(x)}| \cdot |\sqrt{u(x)} + \sqrt{v(x)}| \mathrm{d}x \\ &= \frac{1}{2} \int |u(x) - v(x)| \mathrm{d}x. \end{align*}

\lemPointwiseControlImpliesBlowup*

\begin{proof}[Proof of Lemma~\ref{lemma:pointwise_control_implies_blowup}]

Since the bracket $[u_{\eta_0}, v_{\eta_0}]$ satisfies Lemma~\ref{lemma:mle_close_to_M_projection}, note that $\flc \in [u_{\eta_0}, v_{\eta_0}]$ and $\int_{\mathbb{R}^d} \left(v_{\eta_0}(x) - u_{\eta_0}(x) \right) dx \leq \eta_0$.

    This proof will use the following generic result, where we slightly weaken the constants for convenience.
    \begin{thm} \label{thm:wong1995restate}
        \emph{\citep[Theorem 1]{wong1995probability}} Let $\mathcal{F}$ be a class of densities, and let $\zeta_0$ be such that \begin{equation}\label{eqn:wong_shen_condition}
             \int_{\zeta_0^2/2^8}^{\sqrt{2} \zeta_0} \mathcal{H}^{1/2}_{[]}(\mathcal{F}, h, \zeta/10) \mathrm{d}\zeta \le \frac{1}{2^{11}}\sqrt{n} \zeta_0^2.
        \end{equation}
        Then there exists a universal constant $C > 2^{-14}$ such that if $X_i \sim f_0$ iid, then \begin{equation}\label{eqn:wong_shen_conclusion} \PP_{X_i \overset{\mathrm{iid}}{\sim} f_0}\left( \inf_{g \in \mathcal{F} : h(f_0, g) \ge \zeta_0} \prod_{i = 1}^{n/2}  \frac{f_0(X_i)}{g(X_i)} \le \exp( n\zeta_0^2/48) \right) \le 4 \exp(- C n \zeta_0^2).\end{equation}
    \end{thm}

    Now, recall that $\varepsilon = h(f^*, \flc) > 0.$   Let us choose $\eta_0$ such that for the bracket $[u_{\eta_0}, v_{\eta_0}]$ from the proof of Lemma~\ref{lemma:mle_close_to_M_projection}, \begin{itemize}
        \item $\forall g \in [u_{\eta_0}, v_{\eta_0}]$, $h(g, f^*) \ge \varepsilon \sqrt{24/25},$
        \item $\sqrt{\eta_0/2} \le (\varepsilon^2 / 10) (24/25) \cdot 2^{-8}$ 
    \end{itemize}
    
    It is evident that the second criterion can be met by taking $\eta_0$ small enough. For the first, observe that for any $g \in [u_{\eta_0}, v_{\eta_0}],$ 
    \[ h(g, f^*) \ge h(\flc, f^*) - h(g, \flc),\] due to the triangle inequality (which applies since $h$ is a metric). Further, observe that \begin{align*}
        h(g, \flc) &\le \sqrt{\frac{1}{2} \int |g(x) - \flc(x) | \,\mathrm{d}x}\\
        &= \sqrt{\frac{1}{2} \int \left(\max( g(x), \flc(x)) - \min(g(x), \flc(x))\right)  \,\mathrm{d}x} \\
        &\le \sqrt{\frac{1}{2} \int \left(v_{\eta_0}(x) - u_{\eta_0}(x)\right) \,\mathrm{d}x} \leq \sqrt{\eta_0/2},
    \end{align*} 
    where we have used the fact that both $\flc$ and $g$  lie in  $[u_{\eta_0}, v_{\eta_0}].$ Since the upper bound decays as $\sqrt{\eta_0},$ taking $\eta_0$ small enough also yields that $h(g, f^*) \ge \varepsilon \sqrt{24/25}$.
    
    But, with this in hand, we observe that since $[u_{\eta_0}, v_{\eta_0}]$ is a bracket with Hellinger size at most $\sqrt{\eta_0/2} < (\varepsilon^2/10) (24/25) \cdot 2^{-8},$ defining $\zeta_0 := \varepsilon \sqrt{24/25} ,$ we find that   $$\mathcal{H}_{[]}([u_{\eta_0}, v_{\eta_0}], h, (\zeta_0^2 / 2^8) (1/10)) = \log(1) = 0.$$
    
    As a result, the conclusion of Theorem~\ref{thm:wong1995restate} above applies to $[u_{\eta_0}, v_{\eta_0}]$ with the above value of $\zeta_0$. (Since bracketing entropies decrease with $\zeta,$ the integral in the condition (\ref{eqn:wong_shen_condition}) evaluates to $0$ under our choice of $\zeta_0$.) Instantiating $f_0 = f^*$ and $\mathcal{F} = [u_{\eta_0}, v_{\eta_0}]$ in (\ref{eqn:wong_shen_conclusion}), we find that \[ \PPP{ \inf_{g \in [u_{\eta_0}, v_{\eta_0}], h(g, f^*) \ge \zeta_0} \prod_{i \in \dee_{0,n}} \frac{f^*(Y_i)}{g(Y_i)} \le \exp( n \zeta_0^2/48)} \le 4\exp(- Cn \zeta_0^2). \]
    
    The conclusion then follows on observing that we chose $\eta_0$ such that $h(g,f^*) \ge \varepsilon \sqrt{24/25} = \zeta_0$ for all $g \in [u_{\eta_0}, v_{\eta_0}].$ So the infimum in the probability expression extends to all $g \in [u_{\eta_0}, v_{\eta_0}].$ Finally, we observe that $\zeta_0^2 /48= \varepsilon^2/50.$\qedhere

\end{proof}


\paragraph{Likelihood ratios do not blow up if KL divergence is small.} In this section, we restate and prove Lemma~\ref{lem:R_n_small_markov}. This final piece of the puzzle is a generic application of Tchebycheff's inequality to a log-likelihood ratio between two densities that are close in the sense of Assumption~\ref{assumption:estimability}.

\lemRnSmallMarkov*

\begin{proof}[Proof of Lemma~\ref{lem:R_n_small_markov}]

    
    
    Observe that 
    \begin{align*}
        \PPP{ R_n \ge \exp(n \varepsilon^2/100) \: \middle| \: \mathcal{Y}_{1,n} } &= \PPP{ \prod_{i \in \dee_{0,n} } \frac{f^*(Y_i)}{\onehatf(Y_i)} \ge \exp(n \varepsilon^2/100) \: \middle| \: \mathcal{Y}_{1,n} }\\
        &= \PPP{ \sum_{i \in \dee_{0,n} } \log \frac{f^*(Y_i)}{\onehatf(Y_i)} \ge n \varepsilon^2/100 \: \middle| \: \mathcal{Y}_{1,n} } \\
        &= \PPP{ \frac{2}{n} \sum_{i \in \dee_{0,n} } \left(\log \frac{f^*(Y_i)}{\onehatf(Y_i)} - \mu\right) \ge \varepsilon^2/50 - \mu \: \middle| \: \mathcal{Y}_{1,n}  },
    \end{align*}
    where we use $\mu := \mathbb{E}[ \log (f^*(Y_i) / \onehatf(Y_i)) \mid \mathcal{Y}_{1,n}]$ for $i \in \dee_{0,n}$. (Since the data are iid, this is invariant to the choice of $i$.) Now, if $\mu \ge \varepsilon^2/50,$ then we may simply upper bound this final probability by 1. On the other hand, if $\mu < \varepsilon^2/50,$ then we observe that the right hand side is positive, and so we may upper bound the above quantity by noting \begin{align*}
        \PPP{\frac{4}{n^2} \left( \sum_{i \in \dee_{0,n}} \left( \log \frac{f^*(Y_i)}{\onehatf(Y_i)} - \mu \right) \right)^2 \ge (\varepsilon^2/50 - \mu)^2 \: \middle| \: \mathcal{Y}_{1,n}} \\ \le \frac{4\mathbb{E}\left[ \left( \sum_{i \in \dee_{0,n}} \left(\log (f^*(Y_i)/\onehatf(Y_i)) - \mu \right) \right)^2 \: \middle| \: \mathcal{Y}_{1,n} \right] }{n^2 (\varepsilon^2/50 - \mu)^2}.
    \end{align*} 
    
    We observe that due to the independence of the data, the summands in \[\sum_{i \in \dee_{0,n}} \left(\log(f^*(Y_i)/\onehatf(Y_i)) - \mu\right) \] are centered and iid given $\mathcal{Y}_{1,n}$. Therefore the conditional mean of the square of the sum is its conditional variance. The additivity of variance over sums of independent random variables yields that if $\mu < \varepsilon^2/50,$ then \[ \PPP{ \prod_{i \in \dee_{0,n} } \frac{f^*(Y_i)}{\onehatf(Y_i)} \ge \exp(n \varepsilon^2/100) \: \middle| \: \mathcal{Y}_{1,n} } \le \frac{ 4\int f^*(x) (\log (f^*(x)/\onehatf(x)) - \mu)^2 \mathrm{d}x }{n (\varepsilon^2/50 - \mu)^2} .\]
    
    The conclusion now follows upon recalling that $\mu = \int f^*(x) \log (f^*(x) / \onehatf(x)) \mathrm{d}x = \dkl(f^*\|\onehatf),$ and that variance is less than or equal to a raw second moment. \qedhere
\end{proof}

\section{Additional Normal Mixture Simulations} \label{app:addl_sims}

\subsection{Visualizing Log-concave MLEs} \label{app:addl_visualize}

In Section~\ref{sec:solving}, we visualize the log-concave MLEs of samples from two-component Gaussian mixtures of the form $$f^*(x) = 0.5\phi_d(x) + 0.5\phi_d(x-\mu).$$ Section~\ref{sec:solving} considers the $n=5000$ and $d=1$ setting for both log-concave ($\|\mu\|\leq 2$) and not log-concave ($\|\mu\| > 2$) true densities. We provide visualizations in several additional settings.

In the one-dimensional setting, we compute the log-concave MLEs $\hat{f}_n$ on samples $\{x_1, \ldots, x_n\}$. Figures \ref{fig:logconc_densities_n50_d1} and \ref{fig:logconc_densities_n5000_d1_app} show the true and log-concave MLE densities for samples with $n=50$ and $n=5000$, respectively. These simulations use both the \texttt{LogConcDEAD} and \texttt{logcondens} packages. \texttt{logcondens} only works in one dimension but is much faster than \texttt{LogConcDEAD}. Visually, we see that these two packages produce approximately the same densities. Furthermore, we include values of $n^{-1} \sum_{i=1}^n \log(f^*(x_i))$ on the true density plots and $n^{-1} \sum_{i=1}^n \log(\hat{f}_n(x_i))$ on the log-concave MLE plots. The log likelihood is approximately the same for the two density estimation methods. 

When $\|\mu\| = 0$ or $\|\mu\| = 2$, the true density is log-concave. As we increase from $n=50$ to $n = 5000$, the log-concave MLE becomes a better approximation to the true density. We see this improvement both visually and numerically. That is, $n^{-1} \sum_{i=1}^n \log(\hat{f}_n(x_i))$ is closer to $n^{-1} \sum_{i=1}^n \log(f^*(x_i))$ for larger $n$. When $\|\mu\| = 4$, the underlying density is not log-concave. The log-concave MLE at $\|\mu\| = 4$ and  $n=5000$ seems to have normal tails, but it is nearly uniform in the middle.

\begin{figure}[H]
\centering
\includegraphics[scale=.6]{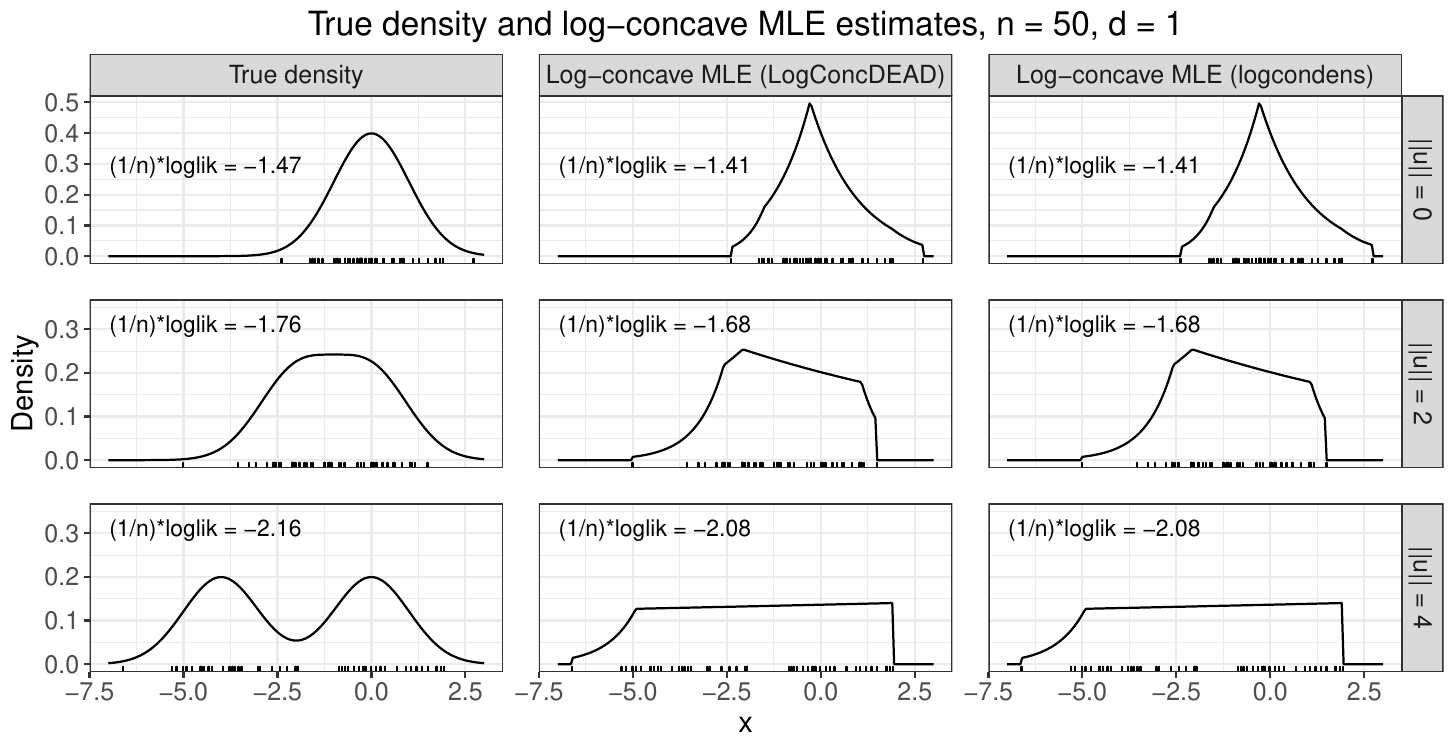}
\caption{Density plots from fitting log-concave MLE on $n = 50$ observations. Tick marks represent the observations. The true density is the Normal mixture $f^*(x) = 0.5\phi_1(x) + 0.5\phi_1(x-\mu)$. In all settings, the \texttt{LogConcDEAD} and \texttt{logcondens} packages return similar results.}
\label{fig:logconc_densities_n50_d1}
\end{figure} 

\begin{figure}[H]
\centering
\includegraphics[scale=.6]{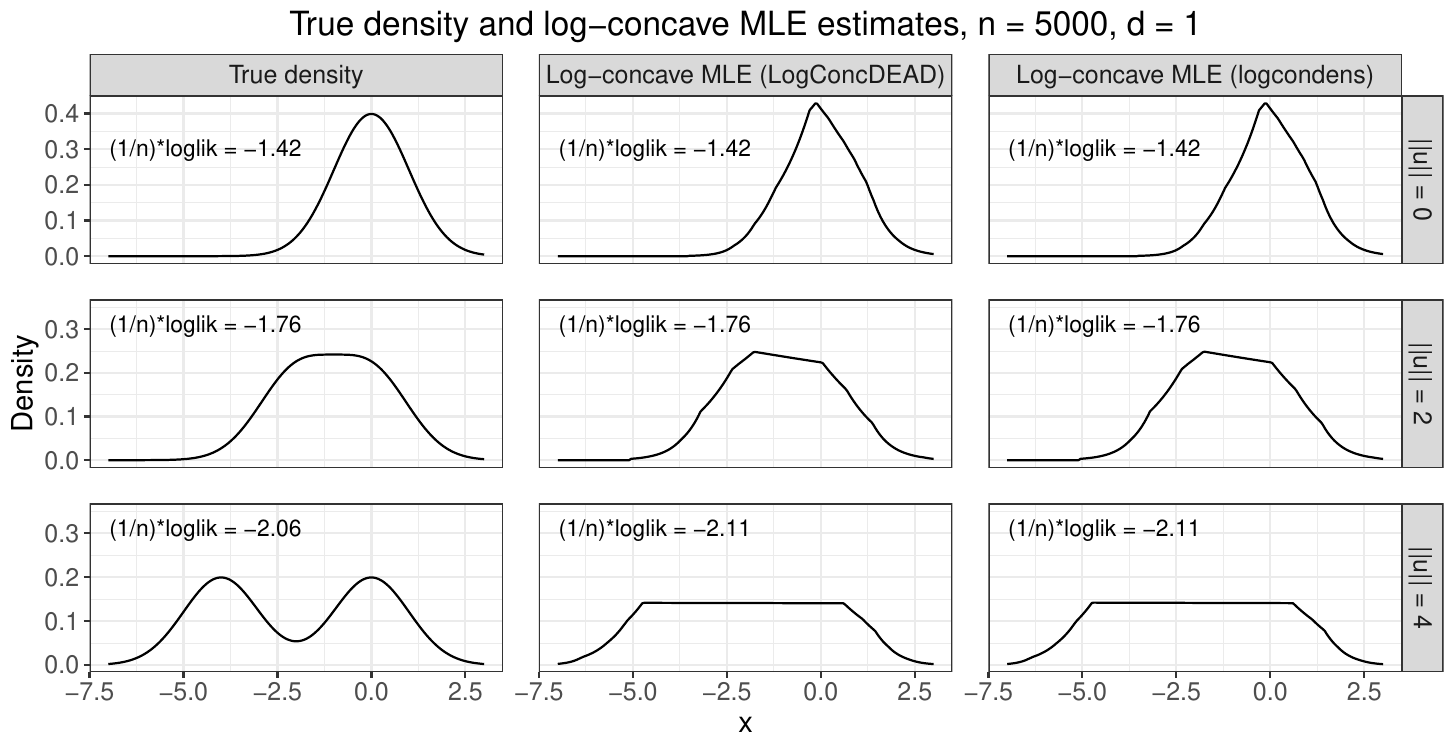}
\caption{Density plots from fitting log-concave MLE on $n = 5000$ observations. The true density is the Normal mixture $f^*(x) = 0.5\phi_1(x) + 0.5\phi_1(x-\mu)$. In all settings, the \texttt{LogConcDEAD} and \texttt{logcondens} packages return similar results. In the $\|\mu\| = 0$ and $\|\mu\| = 2$ log-concave settings, the log-concave MLE on 5000 observations is close to the true density. In the $\|\mu\| = 4$ non-log-concave setting, the log-concave densities appear to have normal tails and uniform centers.}
\label{fig:logconc_densities_n5000_d1_app}
\end{figure} 

We observe similar behavior in the two-dimensional setting. In two dimensions, we use $\mu = (-\|\mu\|, 0)$. Figures~\ref{fig:logconc_contours_n50_d2} and \ref{fig:logconc_contours_n500_d2} show two-dimensional contour plots for the true and log-concave MLEs with $n=50$ and $n=500$. In Figure~\ref{fig:logconc_contours_n50_d2}, we can clearly see that the support of the log-concave MLE is the convex hull of the observed sample. For $\|\mu\| = 0$ and $\|\mu\| = 2$, the true density and log-concave MLE have more similar appearances when $n=500$. In addition, $n^{-1} \sum_{i=1}^n \log(\hat{f}_n(x_i))$ is closer to $n^{-1} \sum_{i=1}^n \log(f^*(x_i))$ for larger $n$. When $\|\mu\| = 4$, the log-concave MLE density is nearly flat in the center of the density.
 
\begin{figure}[H]
\centering
\includegraphics[scale=.55]{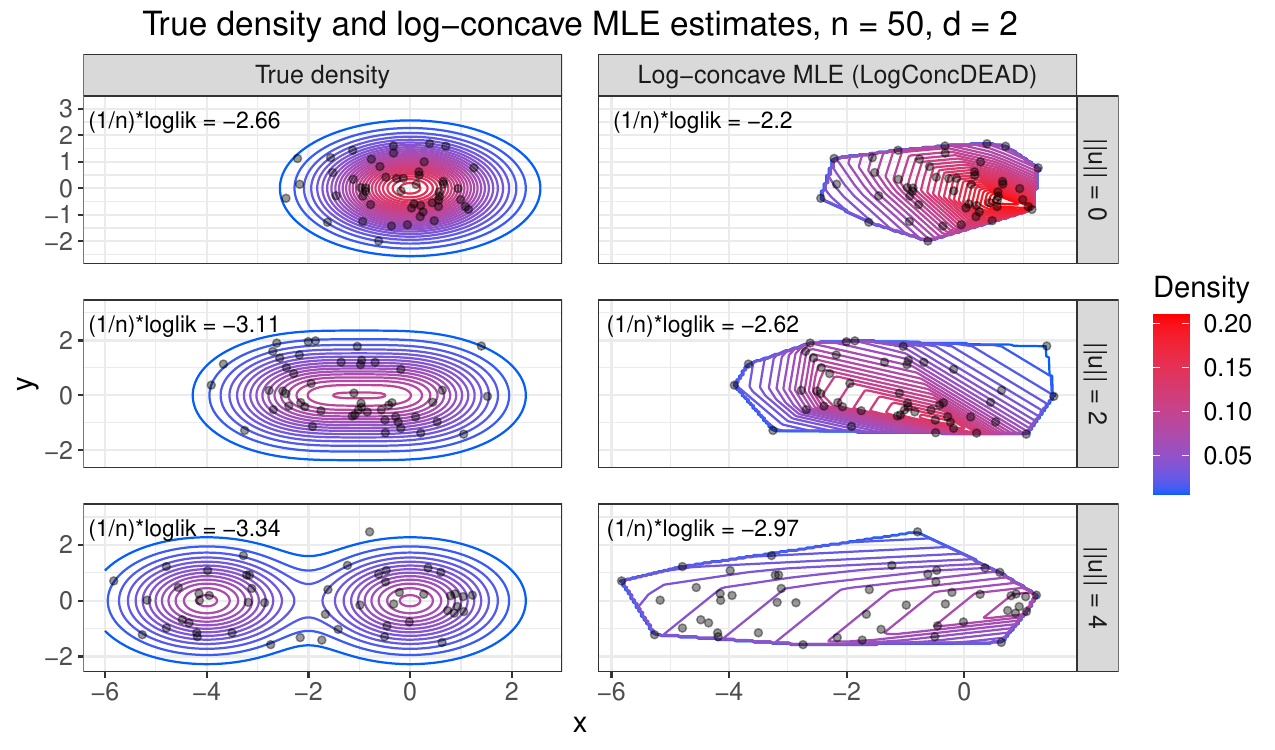}
\caption{Contour plots from fitting log-concave MLE on $n = 50$ observations. Points represent the 50 observations. The true density is the Normal mixture $f^*(x) = 0.5\phi_2(x) + 0.5\phi_2(x-\mu)$. The log-concave MLE has a density of zero outside of the convex hull of the observations.}
\label{fig:logconc_contours_n50_d2}
\end{figure} 

\begin{figure}[H]
\centering
\includegraphics[scale=.55]{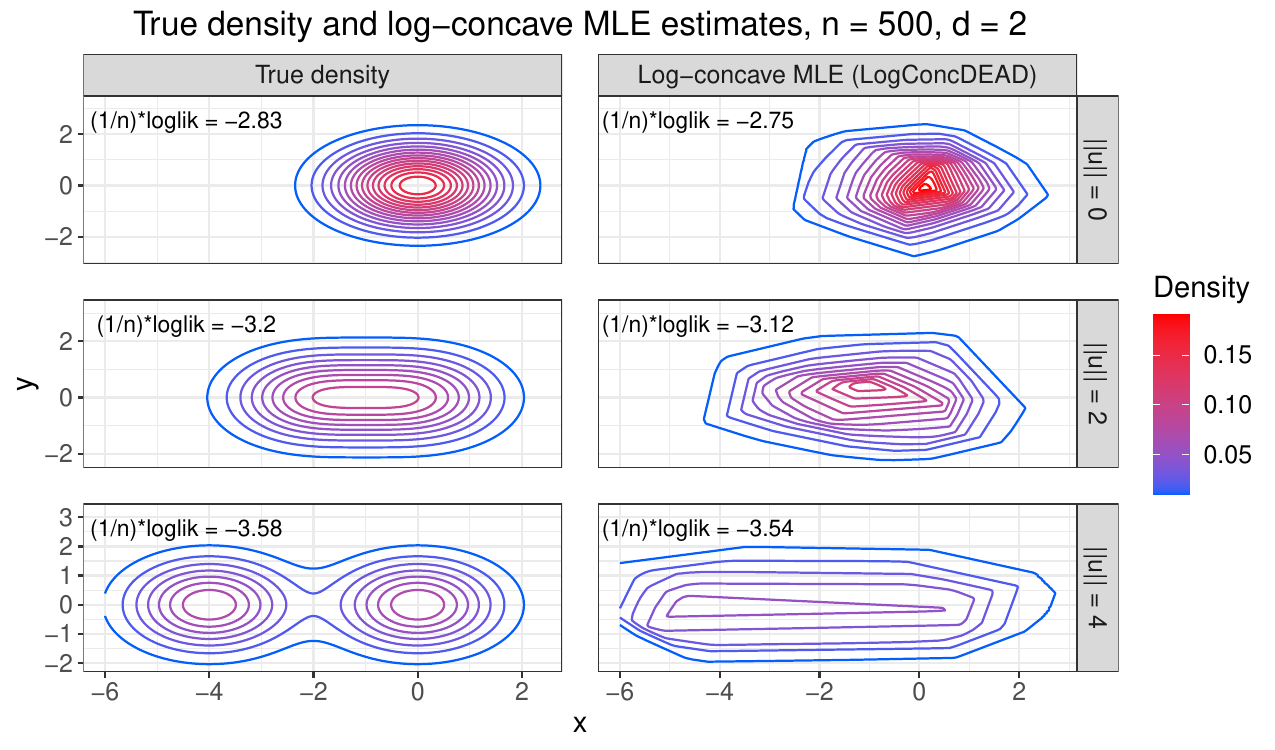}
\caption{Contour plots from fitting log-concave MLE on $n = 500$ observations. The true density is the Normal mixture $f^*(x) = 0.5\phi_2(x) + 0.5\phi_2(x-\mu)$. In the $\|\mu\| = 0$ and $\|\mu\| = 2$ settings, the log-concave MLE over 500 observations has a similar appearance to the true density. When $\|\mu\| = 4$, the interior of the log-concave MLE density appears to be nearly uniform, similar to the $d=1$ case.}
\label{fig:logconc_contours_n500_d2}
\end{figure} 

\subsection{Permutation Test under Additional Parameter Settings} \label{app:perm_test_addl}

Figure~\ref{fig:perm_randproj_n100} demonstrated that the permutation test for log-concavity was not valid for $d\geq 4$ at $n=100$. We consider whether these results still hold with a larger sample size. Figure~\ref{fig:perm_test_reject_n250} simulates the permutation test at $n = 250$. Compared to the $n = 100$ setting, this setting has slightly higher power at $\|\mu\| = 4$ and $\|\mu\| = 5$ when $d = 1$ or $d = 2$. We still see that the rejection proportion is much higher than 0.10 for $\|\mu\| \leq 2$ at $d = 4$ and $d = 5$.

\begin{figure}[H]
\centering
\includegraphics[scale=.7]{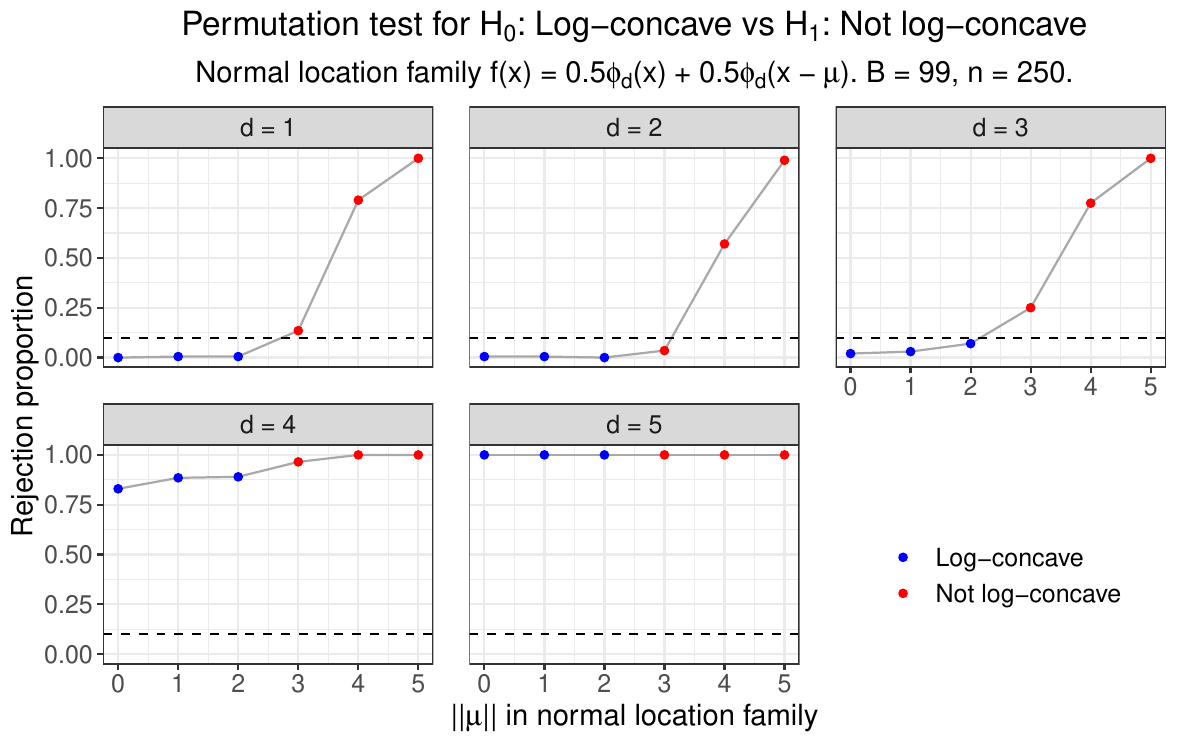}
\caption{Rejection proportions for test of $H_0: f^*$ is log-concave versus $H_1: f^*$ is not log-concave, using the permutation test from \cite{cule2010}. We set $\alpha = 0.10$ and $n=250$, and we perform 200 simulations at each combination of ($d, \|\mu\|$). The test permutes the observations $B = 99$ times. The results are similar to Figure~\ref{fig:perm_randproj_n100}. The permutation test is valid or approximately valid for $d\leq 3$, but it is not valid for $d\geq 4$.}
\label{fig:perm_test_reject_n250}
\end{figure}

Next, we consider whether the permutation test results hold if we increase $B$, the number of times that we shuffle the sample. In Figure~\ref{fig:perm_test_reject_B_vary}, we show the results of simulations at $B \in \{100, 200, 300, 400, 500\}$ on $n=100$ observations. Each row corresponds to the same set of simulations performed at five values of $B$. Looking across each row, we do not see an effect as $B$ increases from 100 to 500. In these analyses, the lack of validity at $d=4$ and $d=5$ remains as we increase $n$ or increase $B$.

\begin{figure}[H]
\centering
\includegraphics[scale=.8]{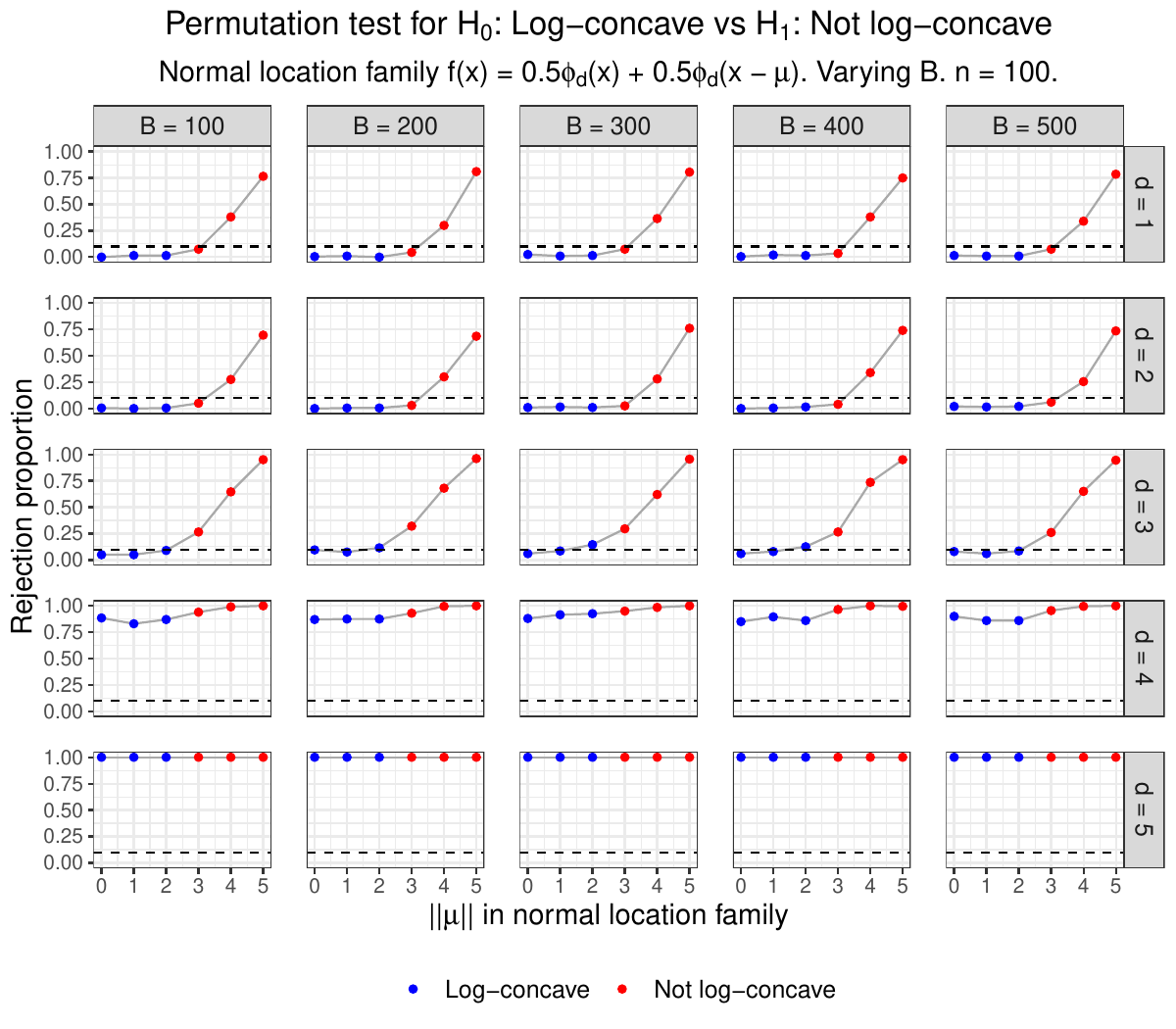}
\caption{Rejection proportions for test of $H_0: f^*$ is log-concave versus $H_1: f^*$ is not log-concave, using the permutation test from \cite{cule2010}. We set $\alpha = 0.10$ and $n=100$, and we perform 200 simulations at each combination of ($B, d, \|\mu\|$). At these larger numbers of shuffles $B$, the permutation test still is not valid for $d\geq 4$.}
\label{fig:perm_test_reject_B_vary}
\end{figure}

Recall that the test statistic is $T = \sup_{A \in \mathcal{A}_0} |P_n(A) - P^*_n(A)|$, and the test statistic on a shuffled sample is $T_b^* = \sup_{A\in\mathcal{A}_0} |P_{n,b}(A) - P_{n,b}^*(A)|$. Both $P$ and $P^*$ are proportions (out of $n$ observations), so $T$ and $T_b^*$ can only take on finitely many values. We consider whether the conservativeness of the test (e.g., $d=1$) or the lack of validity of test (e.g., $d=5$) is due to this discrete nature. Figure~\ref{fig:perm_test_quantiles} plots the distribution of shuffled test statistics $T_b^*$ across eight simulations. The left panels consider the $d = 1$ case at all combinations of $\|\mu\| \in \{0, 2\}$ and $B \in \{100, 500\}$. We see that ``bunching'' of the quantiles is not responsible for the test being conservative in this case. (For instance, if the $90^{th}$ percentile were equal to the $99^{th}$ percentile, then it would make sense for the method to be conservative at $\alpha = 0.10$.) Instead, the 0.90, 0.95, and 0.99 quantiles (dashed blue lines) are all distinct, and the original test statistic (solid black line) is less than each of these values. We also consider the behavior in the $d=5$ case (right panels). Again, these three quantiles are all distinct. In this case, though, the original test statistic is in the far right tail of the distribution of shuffled data test statistics.

\begin{figure}[H]
\centering
\includegraphics[scale=.8]{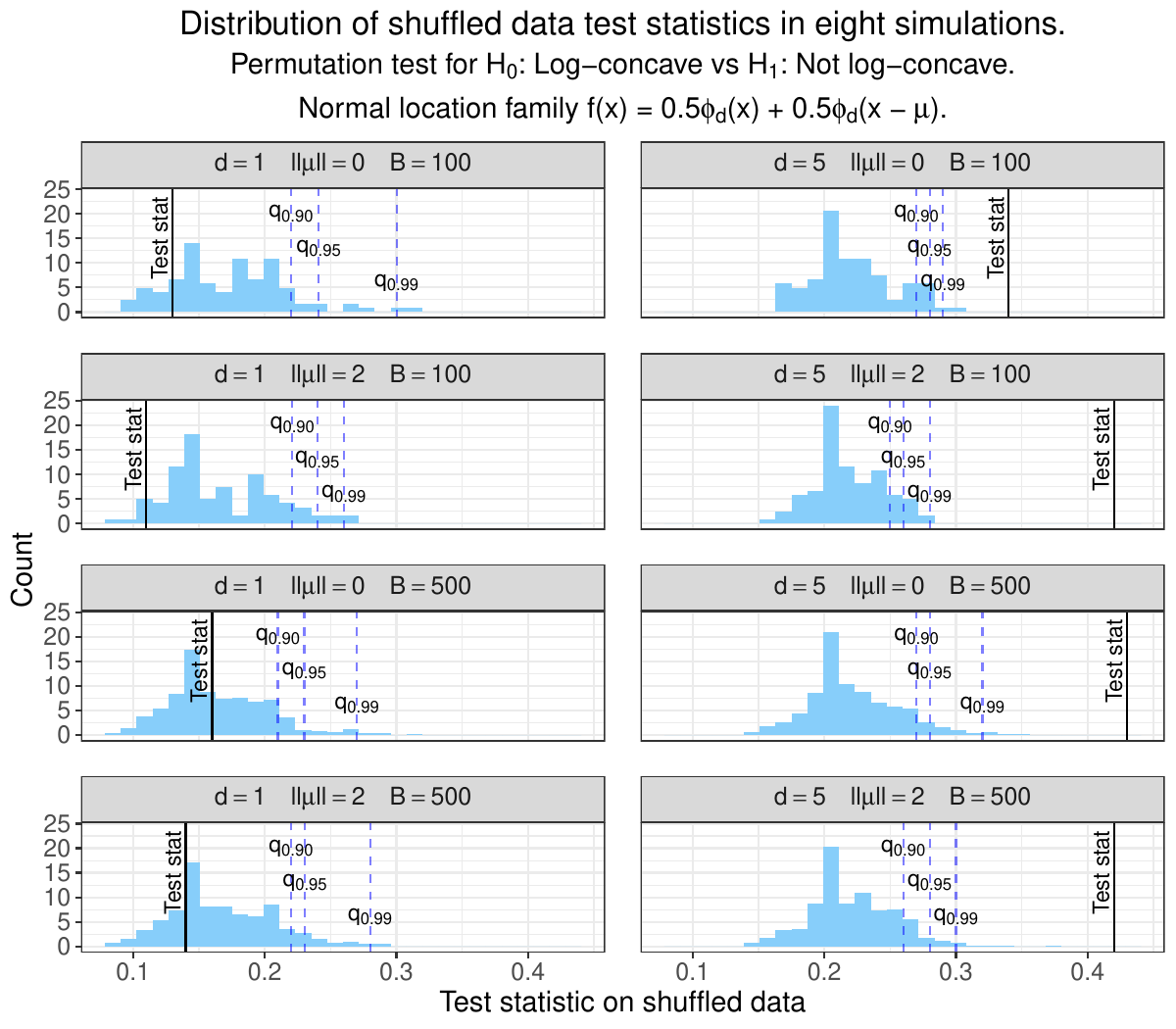}
\caption{Distribution of $T_{n,b}^*$ across eight simulations. The dashed blue lines correspond to the quantiles of the distribution of shuffled data test statistics. The solid black lines correspond to the original test statistics in each simulation. We note that the conservative nature of the permutation test at $d=1$ and the anticonservative nature of the permutation test at $d=5$ are not due to the discreteness of the test statistics.}
\label{fig:perm_test_quantiles}
\end{figure}

\subsection{Relationship between power, $\|\mu\|$, and $d$ in Full Oracle Test} \label{app:power_mu_d}

Unlike the permutation test, the full oracle universal test controls the type I error both theoretically and in simulations. In Section~\ref{sec:fulloracle}, we note that $\|\mu\|$ needs to grow exponentially with $d$ to maintain a certain level of power in the full oracle test. Figure~\ref{fig:power_vary_d} demonstrates this relationship, by exploring how $\|\mu\|$ needs to grow with $d$ to maintain power of approximately 0.90. For each value of $d$, we vary $\|\mu\|$ in increments of 1 and estimate the power through 200 simulations. We choose the value of $\|\mu\|$ with power closest to 0.90. If none of the $\|\mu\|$ values have power in the range of $[0.88, 0.92]$ at a given $d$, then we use finer-grained values of $\|\mu\|$. From the best fit curve, it appears that $\|\mu\|$ needs to grow at an exponential rate in $d$ to maintain the same power. Thus, while the full oracle approach offers an improvement in validity over the permutation test, the power becomes substantially worse in higher dimensions.

\begin{figure}[H]
\centering
\includegraphics[scale=.65]{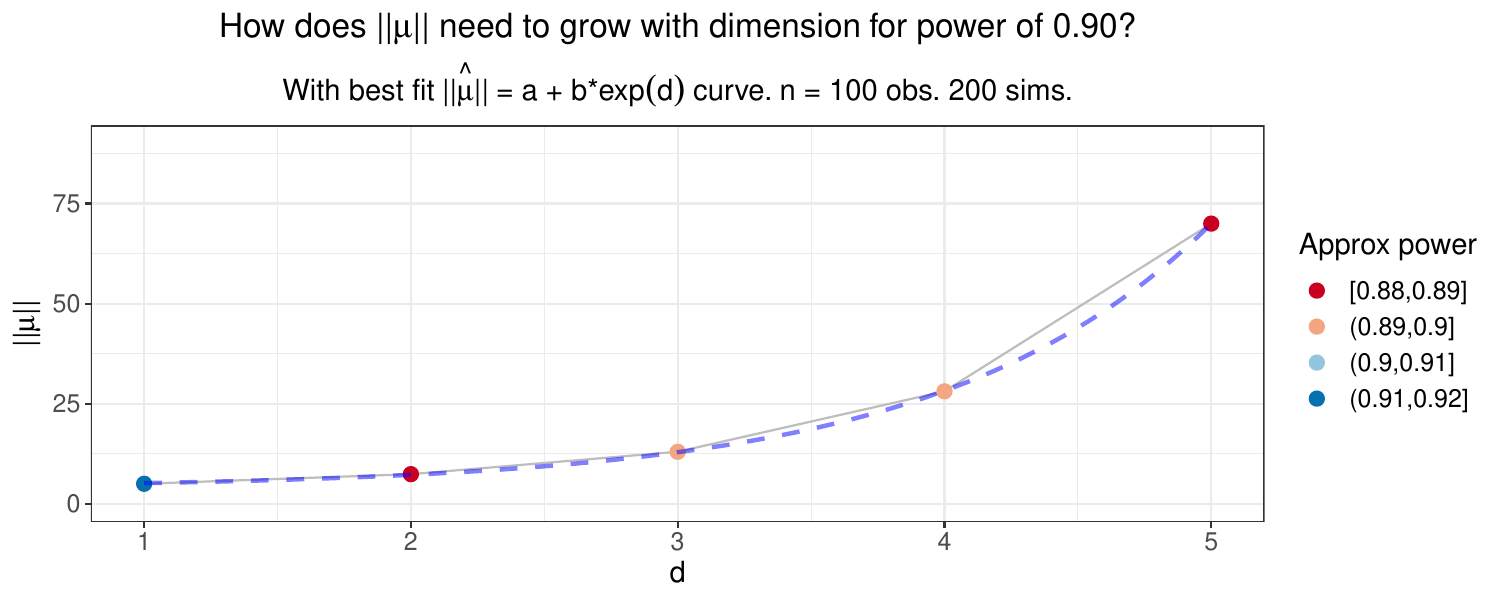}
\caption{Power of the universal test of $H_0: f^*$ is log-concave versus $H_1: f^*$ is not log-concave. Simulations use the true density numerator, $B = 100$ subsamples, and $n = 100$ observations. $\|\mu\|$ needs to grow exponentially in $d$ to maintain power of 0.90. }
\label{fig:power_vary_d}
\end{figure} 

\section{Comparing Full-dimensional and Projection Tests}\label{app:full_vs_proj}

\subsection{Normal Location Mixture Model}

Figures~\ref{fig:unequal_spaced} and \ref{fig:equal_spaced} show that for the normal location mixture in $d\geq 2$, the projection tests for 
log-concavity have higher power than the $d$-dimensional tests. We now offer more insight into this behavior, specifically for the two-dimensional density $f^*(x) = 0.5\phi_2(x) + 0.5\phi_d(x-\mu)$ with $\mu = -(6, 0)$. From the simulations in the upper right panel of Figure~\ref{fig:unequal_spaced}, the partial oracle two-dimensional test has estimated power of 0.025, and the partial oracle random projections test has estimated power of 0.975. 

Figure~\ref{fig:d1_vs_d2_test_stats} shows the distribution of log test statistics over 1000 simulations. Each simulation uses $B = 100$ data splits. The random projection test uses a single projection for illustration, though Figures~\ref{fig:unequal_spaced} and \ref{fig:equal_spaced} use 100 random projections for each data split. From the upper panel of Figure~\ref{fig:d1_vs_d2_test_stats}, only 4.5\% of the log test statistics from the two-dimensional partial oracle test are greater than or equal to $\log(1/0.1) \approx 2.3$. In contrast, 44.3\% of the log test statistics from the random projection partial oracle test are greater than or equal to $\log(1/0.1)$. This means that many of the individual projections provide stronger evidence against log-concavity than the two-dimensional density estimate. 

Furthermore, 19\% of the random projection test statistics are greater than 1000 (log test statistics greater than $\log(1000) \approx 6.9$). If the test uses 100 projections and even one projection exceeds 1000, then the average test statistic will automatically exceed 10, and the test will reject at $\alpha = 0.1$. A single projection has a test statistic above 1000 most often when the first component of the random projection vector $(\theta_1, \theta_2)$ is large. In particular, the test statistic exceeds 1000 in 65\% of simulations where $|\theta_1| > 0.9$ (where $\theta_1^2 + \theta_2^2 = 1$). When $(\theta_1, \theta_2)$ is drawn uniformly from the boundary of the unit circle (as in these simulations), we can work with the geometric properties of circles to show that $P(|\theta_1| > t) = (2/\pi) \text{arccos}(t)$, $t\in [0,1]$. From this formula, $P(|\theta_1| > 0.9) \approx 0.29$. Hence, repeated projections will likely identify some of these directions with strong evidence against log-concavity. Even one projection onto a direction with strong evidence against log-concavity may provide enough evidence to reject the null hypothesis. 
 
\begin{figure}[H]
\centering
\includegraphics[scale=.6]{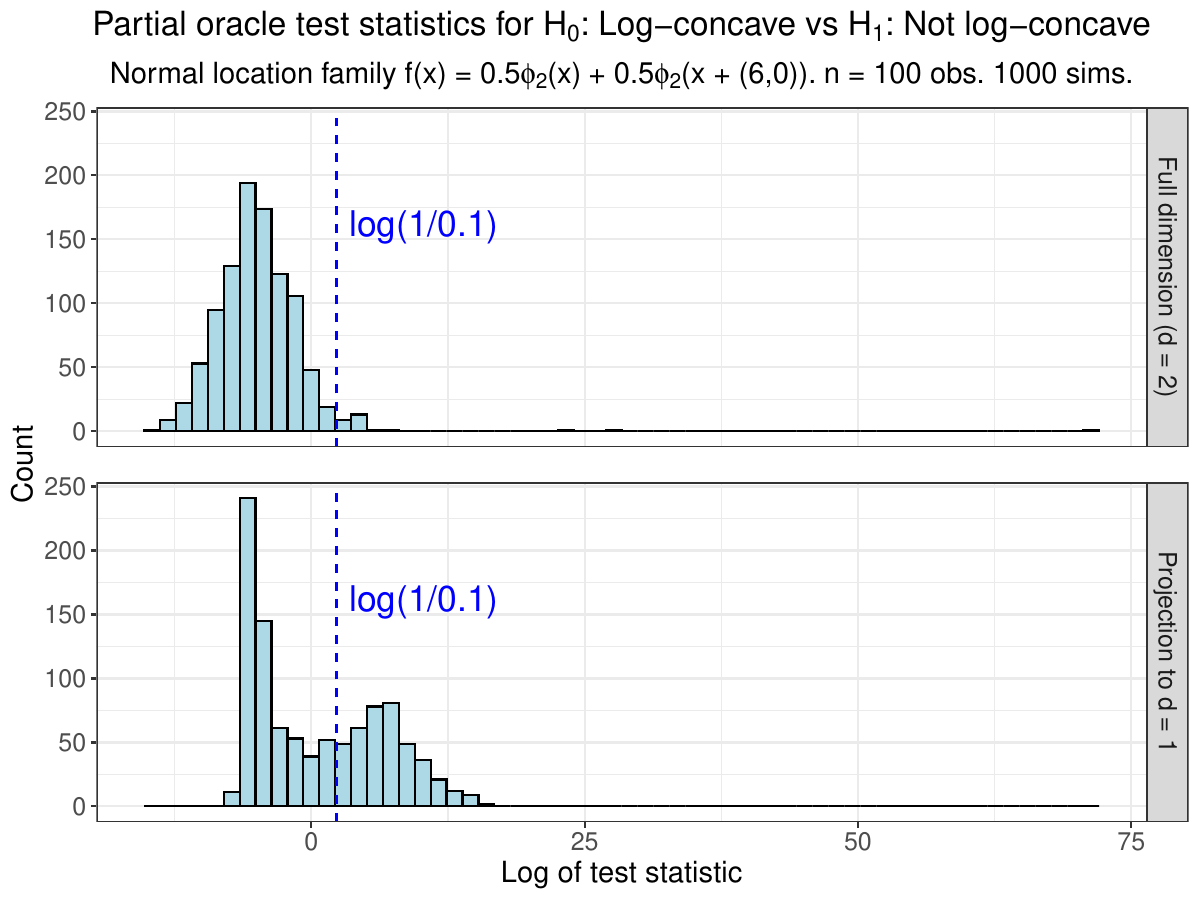}
\caption{Log test statistics for full dimensional and random projection partial oracle tests for log-concavity. To compute each test statistic, we generate 100 observations from the two-dimensional density $0.5\phi_2(x) + 0.5\phi_2(x + (6,0))$, and we compute the test statistic over $B = 100$ random splits of the data. The projection test statistics each use a single projection.}
\label{fig:d1_vs_d2_test_stats}
\end{figure} 

\subsection{Discrete Distribution}

We can also consider the power of the two-dimensional test versus a one-dimensional projection test for an example that is more analytically tractable. The setup for this example is modified from \cite{samworth2018recent}. Suppose the true distribution is a discrete uniform distribution on the three points $\{(-1, 0), (0, \sqrt{3}), (1, 0) \}$. That is, 
\begin{displaymath}
f^*(x, y) = \left\{
\begin{array}{ll}
1/3 & : (x,y) \in \{(-1, 0), (0, \sqrt{3}), (1, 0) \} \\
0 & : \text{ else.}
\end{array}
\right.
\end{displaymath}
Since this distribution is discrete with multiple atoms, it is not log-concave. The log-concave MLE in two dimensions is a uniform distribution over the equilateral triangle formed by the three points. Since the area of the triangle is $\sqrt{3}$, the true log-concave MLE density is $f_{d=2}^\text{LC}(x,y) = 1/\sqrt{3}$ for $(x,y)$ within the triangle (inclusive of the boundary). If we project onto the $x$-axis, then the true log-concave MLE density of the projection is the uniform density $f_{d=1}^\text{LC}(x,y) = 1/2$ over $x\in [-1, 1].$

The relative power of a two-dimensional test versus a one-dimensional projection test will depend on the choice of numerator. Suppose the two-dimensional test uses an alternative density estimate $\hat{f}_{1, d=2}$ that places uniform blocks of width 0.1 and height 0.1 around the three points. For the density to integrate to 1, $\hat{f}_{1, d=2}$ has the form
\begin{displaymath}
\hat{f}_{1, d=2}(x,y) = \left\{
\begin{array}{ll}
1/0.03 & : (x,y) \in [-1.05,-0.95]\times [-0.05,0.05] \: \cup \\
& \hspace*{5em} [-0.05, 0.05] \times [\sqrt{3} - 0.05, \sqrt{3}+0.05] \: \cup \\
& \hspace*{5em} [0.95, 1.05] \times [-0.05, 0.05]  \\
0 & : \text{ else.}
\end{array}
\right.
\end{displaymath}
Using an analogous setup in the projection test, we would use a density estimate of the form 
\begin{displaymath}
\hat{f}_{1, d=1}(x,y) = \left\{
\begin{array}{ll}
1/0.3 & : x \in [-1.05, -0.95] \: \cup \: [-0.05, 0.05] \: \cup \: [0.95, 1.05]  \\
0 & : \text{ else.}
\end{array}
\right.
\end{displaymath}
For large $n$, the estimated log-concave MLE will be approximately equal to the true log-concave MLE. Hence, for a single data split, the two-dimensional test statistic will be $$T_{n, d=2} \approx \frac{(1/0.03)^{n/2}}{(1/\sqrt{3})^{n/2}} \approx (57.74)^{n/2}$$  and the one-dimensional projection test statistics will be 
$$T_{n, d=1} \approx \frac{(1/0.3)^{n/2}}{(1/2)^{n/2}} \approx (6.67)^{n/2}.$$ Both tests will have high power, but the full-dimension test will have larger test statistics.

As another option, suppose the two-dimensional test uses an alternative density estimate $\hat{f}_{1, d=2}$ that places uniform blocks of width 0.1 and height 1 around the three points. Then $\hat{f}_{1, d=2}$ has the form
\begin{displaymath}
\hat{f}_{1, d=2}(x,y) = \left\{
\begin{array}{ll}
1/0.3 & : (x,y) \in [-1.05,-0.95]\times [-0.5,0.5] \: \cup \\
& \hspace*{5em} [-0.05, 0.05] \times [\sqrt{3} - 0.5, \sqrt{3}+0.5] \: \cup \\
& \hspace*{5em} [0.95, 1.5] \times [-0.05, 0.5]  \\
0 & : \text{ else.}
\end{array}
\right.
\end{displaymath}
The corresponding alternative density estimate for the $x$-axis projection test will be the same density as previous. Then for a single data split, the two-dimensional test statistic will be $$T_{n, d=2} \approx \frac{(1/0.3)^{n/2}}{(1/\sqrt{3})^{n/2}} \approx (5.77)^{n/2}$$  and the one-dimensional projection test statistics will be 
$$T_{n, d=1} \approx \frac{(1/0.3)^{n/2}}{(1/2)^{n/2}} \approx (6.67)^{n/2}.$$ In this case, the projection test will have larger test statistics.

Through this simple example, we have seen the interconnection between the true log-concave MLEs and the choice of numerator densities in determining which test has higher power. In practice, it may be sensible to conduct simulations to understand the power of these tests in a setup similar to the given data. Alternatively, as mentioned in Section~\ref{sec:counterpoint}, if one is not sure whether the projected or full-dimensional test will have higher power, one can simply run both, average the resulting test statistics, and threshold the average at $1/\alpha$. Since the average of e-values is an e-value, such a test is still valid, and the test is consistent if either of the original tests is consistent.

\section{Example: Testing Log-concavity of Beta Density} \label{app:beta}

In the one-dimensional normal mixture case, we saw that the full oracle universal test sometimes had higher power than the permutation test. We consider whether this holds in another one-dimensional setting.

The Beta$(\alpha, \beta)$ density has the form $$f(x; \alpha, \beta) = \frac{\Gamma(\alpha + \beta)}{\Gamma(\alpha) \Gamma(\beta)} x^{\alpha - 1} (1-x)^{\beta - 1}, \quad x \in (0, 1),$$ where $\alpha > 0$ and $\beta > 0$ are shape parameters. 

As noted in \cite{cule2010}, Beta$(\alpha, \beta)$ is log-concave if $\alpha \geq 1$ and $\beta \geq 1$. We can see this in a quick derivation:
\begin{align*}
\frac{\partial^2}{\partial x^2} \log f(x; \alpha, \beta) &= \frac{\partial^2}{\partial x^2} \left[\log\left(\frac{\Gamma(\alpha + \beta)}{\Gamma(\alpha) \Gamma(\beta)} \right) + (\alpha - 1)\log(x) + (\beta - 1) \log(1 - x) \right] \\
&= \frac{\partial}{\partial x} \left[\frac{\alpha - 1}{x} + \frac{1-\beta}{1-x} \right] \\
&= \frac{1-\alpha}{x^2} + \frac{1-\beta}{(1-x)^2}.
\end{align*}
This second derivative is less than or equal to 0 for all $x \in (0,1)$ only if both $\alpha \geq 1$ and $\beta \geq 1$. The Beta($\alpha, \beta$) distribution is hence log-concave when $\alpha \geq 1$ and $\beta \geq 1$. This means that tests of $H_0: f^*$ is log-concave versus $H_1: f^*$ is not log-concave should reject $H_0$ if $\alpha < 1$ or $\beta < 1$.

\subsection{Understanding Limiting Log-concave MLEs}

In general, it is non-trivial to solve for the limiting log-concave function $f^\text{LC} = \argmin{f\in\mathcal{F}_d} D_\text{KL}(f^* \| f)$. We try to determine $f^\text{LC}$ in a few specific cases. In Figure~\ref{fig:beta_true_logcondens}, we consider two choices of shape parameters $(\alpha, \beta)$ such that the Beta$(\alpha, \beta)$ densities are not log-concave. On the left panels, we plot the Beta densities. For the right panels, we simulate 100,000 observations from the corresponding Beta$(\alpha, \beta)$ density, we fit the log-concave MLE on the sample using \texttt{logcondens}, and we plot this log-concave MLE density. Thus, the right panels should be good approximations to $f^\text{LC}$ in these two settings.

\begin{figure}[H]
\centering
\includegraphics[scale=.75]{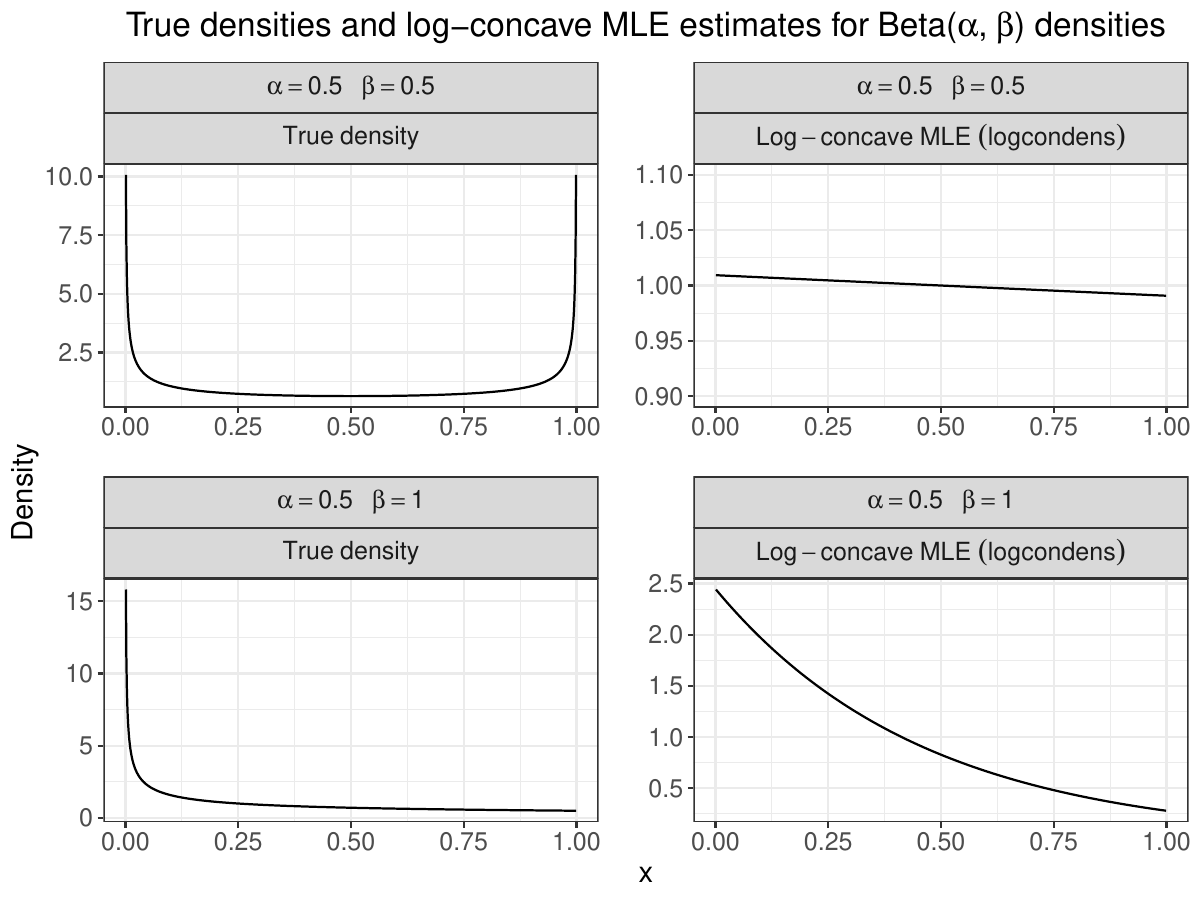}
\caption{Two non-log-concave Beta densities and their corresponding log-concave MLEs, as estimated over $n=100,000$ observations. The limiting log-concave MLE of the Beta(0.5, 0.5) density appears to be Unif(0, 1). The limiting log-concave MLE of the Beta(0.5, 1) density has an exponential appearance, which Figure~\ref{fig:truncated_exp} examines.}
\label{fig:beta_true_logcondens}
\end{figure} 

In the first setting ($\alpha = 0.5, \beta = 0.5$), it appears that the log-concave MLE is the Unif(0, 1) density. We consider the second setting ($\alpha = 0.5, \beta = 1$) in more depth. The density in row 2, column 2 looks similar to an exponential density, but $x$ can only take on values between 0 and 1. The truncated exponential density is given by 
$$f(x; \lambda, b) = \frac{\lambda \exp(-\lambda x)}{1 - \exp(-\lambda b)}, \quad 0 < x \leq b.$$ In this setting, we can try to fit a truncated exponential density with $b = 1$. In Figure~\ref{fig:truncated_exp}, we see that a truncated exponential density with $\lambda = 2.18$ and $b = 1$ provides a good fit for the log-concave MLE.

\begin{figure}[H]
\centering
\includegraphics[scale=.7]{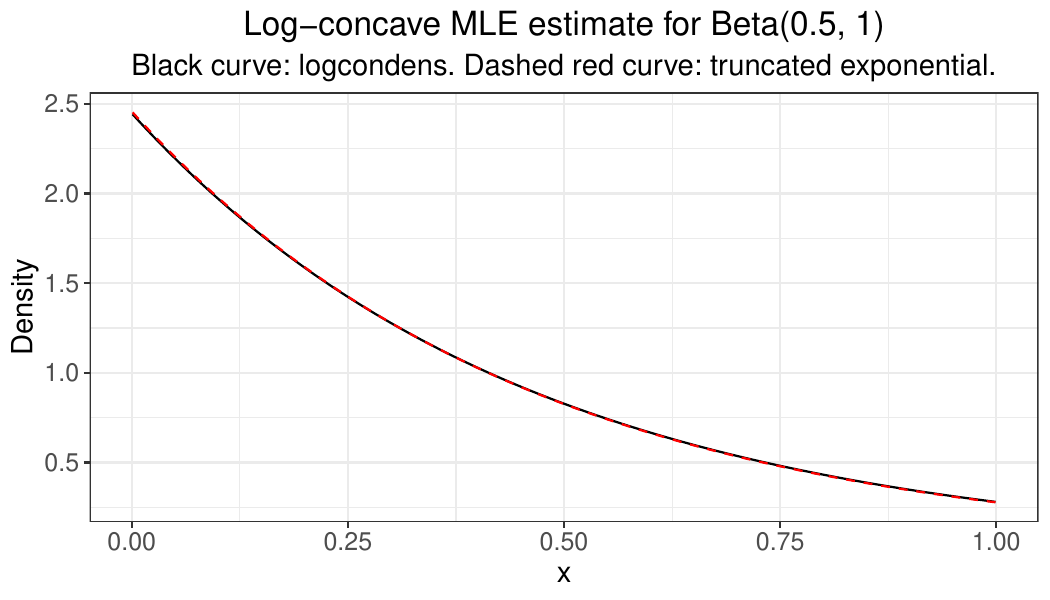}
\caption{The log-concave MLE (black solid) estimated over $n=100,000$ observations from the Beta(0.5, 1) density nearly perfectly matches the truncated exponential density with $\lambda = 2.18$ and $b = 1$ (red dashed). This suggests that this truncated exponential density may be the limiting log-concave MLE for this Beta density.}
\label{fig:truncated_exp}
\end{figure} 

We can also see that the truncated exponential density is log-concave:
\begin{align*}
\frac{\partial^2}{\partial x^2} \log f(x; \lambda, b) &= \frac{\partial^2}{\partial x^2} \left[\log(\lambda) -\lambda x - \log(1 - \exp(-\lambda b)) \right] \\
&= \frac{\partial}{\partial x} [-\lambda] \\
&= 0.
\end{align*}

\subsection{Universal Tests can have Higher Power than Permutation Tests}

Figure~\ref{fig:beta_densities} shows examples of both log-concave and not log-concave Beta densities. We use similar $\alpha$ and $\beta$ parameters in the simulations where we test for log-concavity. This shows that our simulations are capturing a variety of Beta density shapes.

\begin{figure}[H]
\centering
\includegraphics[scale=.65]{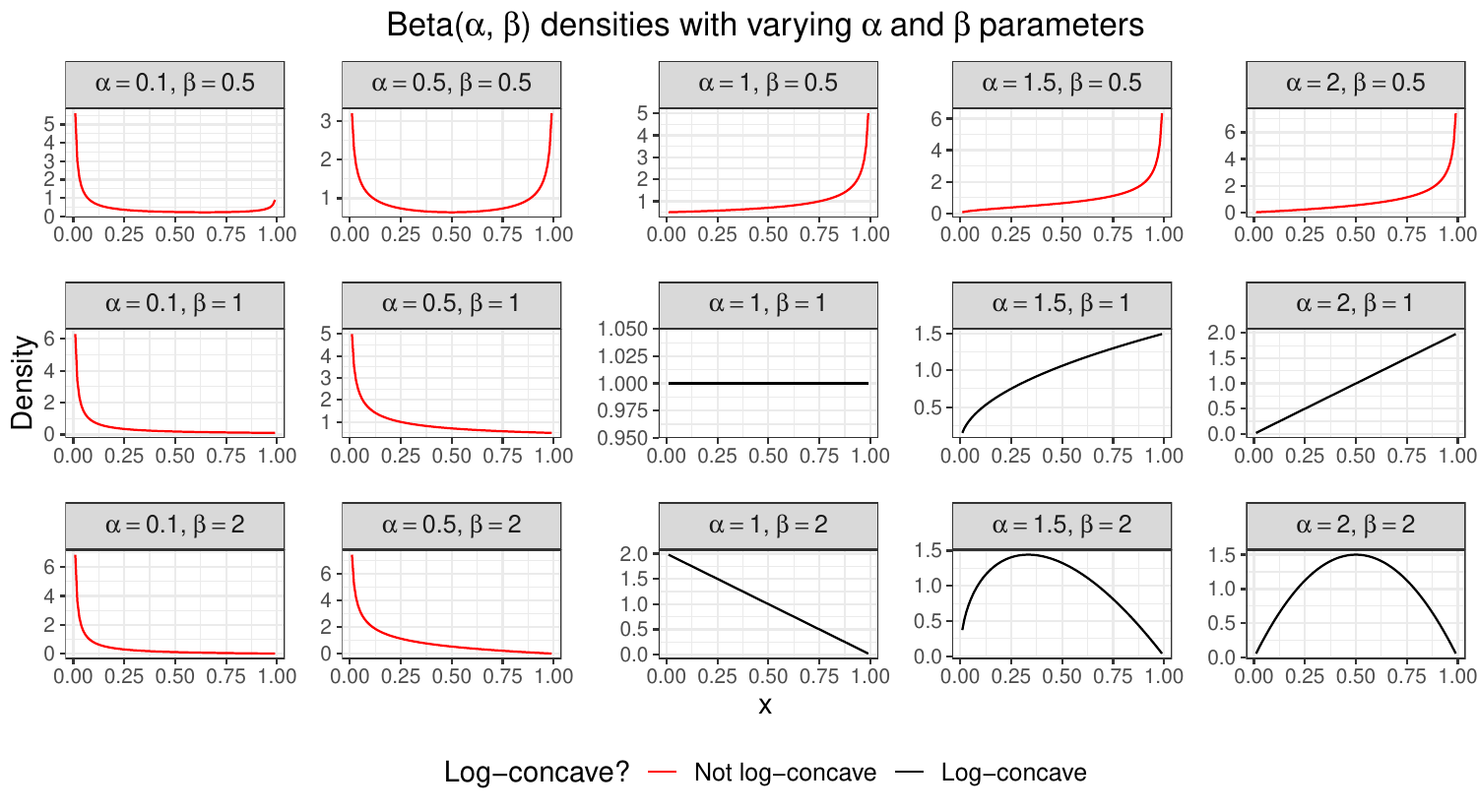}
\caption{Beta densities across a variety of $\alpha$ and $\beta$ parameters. We implement our log-concavity tests on Beta densities across these ranges of parameters, encompassing a variety of density shapes.}
\label{fig:beta_densities}
\end{figure} 

We now implement the full oracle LRT (universal), partial oracle LRT (universal), fully nonparametric LRT (universal), and permutation test. The full oracle LRT uses the true density in the numerator. The partial oracle LRT uses the knowledge that the true density comes from the Beta family. We use the \texttt{fitdist} function in the \texttt{fitdistrplus} library to find the MLE for $\alpha$ and $\beta$ on $\{Y_i : i\in\mathcal{D}_1\}$ computationally \citep{fitdistrplus}. Then the numerator of the partial oracle LRT uses this Beta MLE density. The fully nonparametric approach fits a kernel density estimate on $\{Y_i : i\in\mathcal{D}_1\}$. In particular, we use the \texttt{kde1d} function from the \texttt{kde1d} library, and we restrict the support of the KDE to $[0, 1]$ \citep{kde1d}. This restriction is particularly important in the Beta family case, since some of the non-log-concave Beta densities assign high probability to observations near 0 or 1. (See Figure~\ref{fig:beta_densities}.) The numerator of the fully nonparametric approach uses the KDE. 

Figure~\ref{fig:reject_beta_four} compares the four tests of $H_0: f^*$ is log-concave versus $H_1: f^*$ is not log-concave. We set $n=100$, and we perform 200 simulations to determine each rejection proportion. The universal methods subsample at $B = 100$, and the permutation test uses $B = 99$ shuffles. In the first panel, $\beta = 0.5$, so the density is not log-concave for any choice of $\alpha$. In the second and third panels, $\beta = 1$ and $\beta = 2$. In these cases, the density is log-concave only when $\alpha \geq 1$ as well.  

We observe that the permutation test is valid in all settings, but the three universal tests often have higher power. As expected, out of the universal tests, the full oracle approach has the highest power, followed by the partial oracle approach and then the fully nonparametric approach. When $\beta = 0.5$, all of the universal LRTs have power greater than or equal to the permutation test. When $\beta \in \{1, 2\}$, the universal approaches have higher power for some values of $\alpha$. Again, we see that even when the permutation test is valid, it is possible for universal LRTs to have higher power.

\begin{figure}[H]
\centering
\includegraphics[scale=.6]{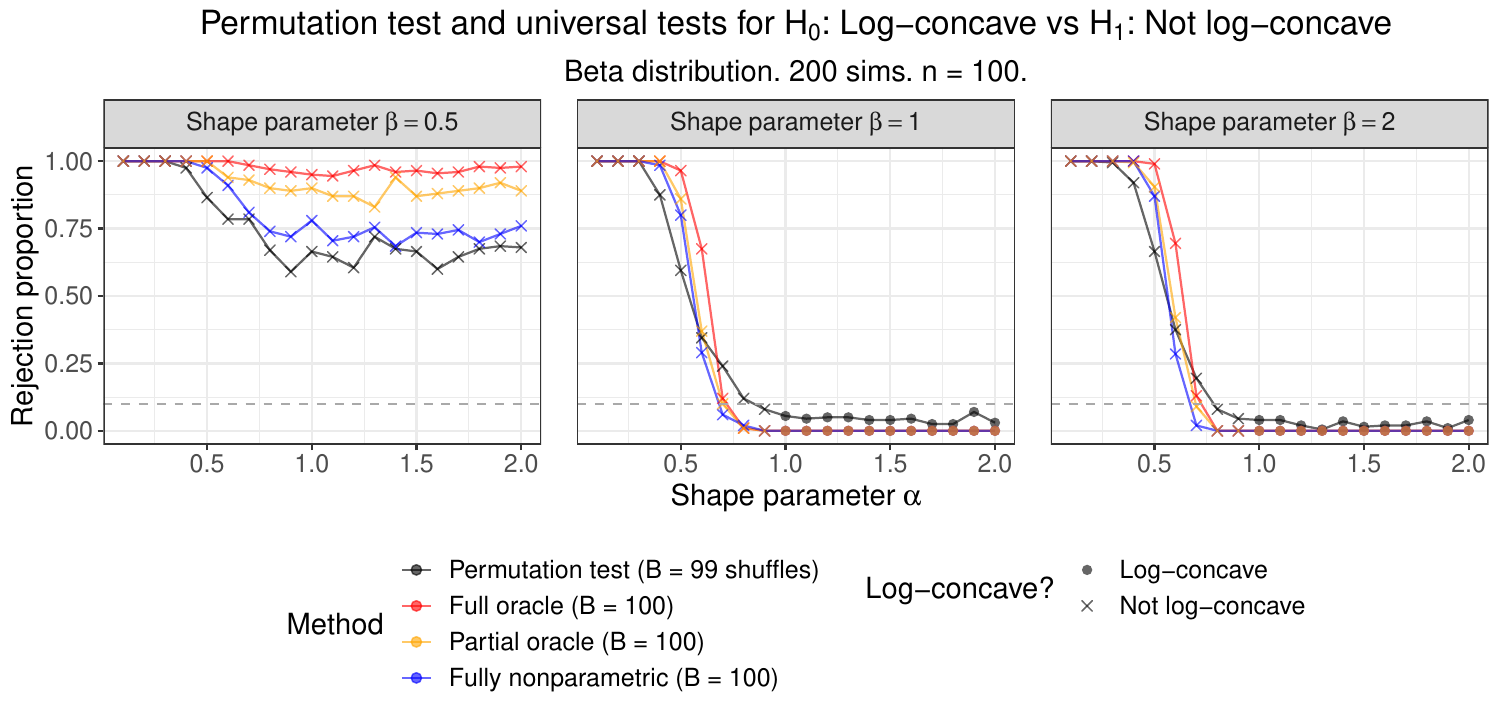}
\caption{Rejection proportions for four tests of $H_0: f^*$ is log-concave versus $H_1: f^*$ is not log-concave, on Beta($\alpha, \beta$) density. The permutation test is valid in all simulations, but the universal tests have both higher power and guaranteed validity.}
\label{fig:reject_beta_four}
\end{figure} 

\section{Permutation Test and Trace Test for Log-concavity} \label{app:perm_trace}

\subsection{Permutation Test} \label{app:perm_test_desc}

\cite{cule2010} construct a permutation test of the hypothesis $H_0: f^*\in\mathcal{F}_d$ versus $H_1: f^*\notin\mathcal{F}_d$. We now discuss this test in more detail.

Algorithm~\ref{alg:permtest} explains the permutation test.

\begin{algorithm}[ht]
\caption{Permutation test for $H_0: f^* \in\mathcal{F}_d$ versus $H_1: f^* \notin \mathcal{F}_d$}
\hspace*{\algorithmicindent} \textbf{Input:} $n$ iid $d$-dimensional observations $Y_1, \ldots, Y_n$ from unknown density $f^*$, \\
\hspace*{\algorithmicindent} \hskip 12pt number of shuffles $B$, significance level $\alpha$. \\
\hspace*{\algorithmicindent} \textbf{Output:} Decision of whether to reject $H_0: f^*\in\mathcal{F}_d$.
\begin{algorithmic}[1]
\State Fit the log-concave MLE $\hat{f}_n$ on $\mathcal{Y} = \{Y_1, \ldots, Y_n\}$. \label{algstep:mlelcd}
\State Draw another sample $\mathcal{Y}^* = \{Y_1^*, \ldots, Y_n^*\}$ from the log-concave MLE $\hat{f}_n$. \label{algstep:rlcd}
\State Compute the test statistic $T = \sup_{A \in \mathcal{A}_0} |P_n(A) - P^*_n(A)|$, where $\mathcal{A}_0$ is the set of all balls centered at a point in $\mathcal{Y} \cup \mathcal{Y}^*$, $P_n(A)$ is the proportion of observations in ball $A$ out of all observations in $\mathcal{Y}$, and $P_n^*(A)$ is the proportion of observations in ball $A$ out of all observations in $\mathcal{Y}^*$ . 
\For {$b=1,2,\ldots,B$}
\State ``Shuffle the stars'' to randomly place $n$ observations from $\mathcal{Y} \cup \mathcal{Y}^*$ into $\mathcal{Y}_b$.
\State Place the remaining $n$ observations in $\mathcal{Y}_b^*$.
\State Using these new samples, compute $T_b^* = \sup_{A\in\mathcal{A}_0} |P_{n,b}(A) - P_{n,b}^*(A)|$. $P_{n,b}(A)$ and   \par
$P_{n,b}^*(A)$ are defined similarly to $P_n(A)$ and $P_n^*(A)$, using $\mathcal{Y}_b$ and $\mathcal{Y}_b^*$. 
\EndFor
\State Arrange the test statistics $(T^*_1, T^*_2, \ldots, T^*_B)$ into the order statistics $(T^*_{(1)}, T^*_{(2)}, \ldots, T^*_{(B)})$.
\State \textbf{return} Reject $H_0$ if $T > T^*_{(\lceil(B+1)(1-\alpha)\rceil)}$.
\end{algorithmic}
\label{alg:permtest}
\end{algorithm}

Intuitively, this test assumes that if $H_0$ is true, the samples $\mathcal{Y}$ and $\mathcal{Y}^*$ will be similar, so $T$ will not be particularly large relative to $T_1^*, \ldots, T_B^*$. Alternatively, if $H_0$ is false, $\mathcal{Y}$ and $\mathcal{Y}^*$ will be dissimilar, and the converse will hold. This approach is not guaranteed to control the type I error level. We observe cases both where the permutation test performs well and where the permutation test's false positive rate is much higher than $\alpha$. 

We provide several computational notes on Algorithm~\ref{alg:permtest}. Steps~\ref{algstep:mlelcd} and \ref{algstep:rlcd} use functions from the \texttt{LogConcDEAD} library. To perform step~\ref{algstep:mlelcd}, we can use the \texttt{mlelcd} function, which estimates the log-concave MLE density from a sample. To perform step~\ref{algstep:rlcd}, we can use the \texttt{rlcd} function, which samples from a fitted log-concave density. Where $\mathcal{A}_0$ is the set of all balls centered at a point in $\mathcal{Y} \cup \mathcal{Y}^*$, $|P_n(A) - P_n^*(A)|$ only takes on finitely many values over $A\in\mathcal{A}_0$. To see this, consider fixing a point at some value $y \in \mathcal{Y} \cup \mathcal{Y}^*$, letting $A_r(y)$ be the sphere of radius $r$ centered at $y$, and increasing $r$ from 0 to infinity. As $r\to\infty$, $|P_n(A_r(y)) - P_n^*(A_r(y))|$ only changes when $A_r(y)$ expands to include an additional observation in $\mathcal{Y} \cup \mathcal{Y}^*$. Hence, it is possible to compute $\sup_{A \in \mathcal{A}_0} |P_n(A) - P^*_n(A)|$ by considering all sets $A$ centered at some $y\in\mathcal{Y} \cup \mathcal{Y}^*$ and with radii equal to the distances between the center of $A$ and all other observations. For large $n$, it may be necessary to approximate the test statistics $T, T_1^*, T_2^*, \ldots, T_B^*$ by varying the radius of $A$ across a smaller set of fixed increments. In each of our simulations, we compute the test statistics exactly.

\subsection{Trace Test} \label{app:trace}

To test $H_0: f^*$ is log-concave versus $H_1: f^*$ is not log-concave, we now briefly consider the trace test from Section~3 of \cite{chen2013smoothed}. The trace test is similar to the permutation test, but its test statistic is the trace of the difference in covariance matrices between the observed data and the fitted log-concave MLE density estimator. The \texttt{hatA} function in the \texttt{LogConcDEAD} library computes this statistic. In $B$ bootstrap repetitions, the test draws a new sample from the observed data's log-concave MLE, fits the log-concave MLE of the new data, and computes the trace statistic. The test compares the original statistic to the bootstrapped statistics. At $d = 4$, $\|\mu\| = 2$, $n=100$, $B = 99$, and $\alpha = 0.1$, the trace test falsely rejected $H_0$ at level $\alpha = 0.1$ in 20 out of 20 simulations. Similar to the results reported by \cite{chen2013smoothed}, at $d=2$ and the same $\|\mu\|$, $n$, $B$, and $\alpha$ as above, the trace test falsely rejects $H_0$ in 19 out of 200 simulations. Hence, similar to the permutation test, simulations suggest that the trace test is valid for $d=2$, but it does not control type I error for $d=4$. This test is also more computationally intensive than the permutation test. The 20 simulations at $d=4$ took about 8 hours to run over 4 cores.

\end{document}